\documentclass[12pt]{article}

\usepackage[english]{babel}
\usepackage[utf8]{inputenc}
\usepackage{amsmath}
\usepackage{graphicx}
\usepackage{subfigure}
\usepackage{amssymb}
\usepackage{amsthm}
\usepackage{bm}
\usepackage{tikz-cd}
\usepackage{mathrsfs}
\usepackage[colorinlistoftodos]{todonotes}
\usepackage{enumitem}
\usepackage{yfonts}
\usepackage{ dsfont }
\usepackage{bbm}
\usepackage{geometry}
\usepackage{setspace}
\usepackage{tikz}
\usepackage{pgfplots}
\usepackage{authblk}
\usepackage{cite}
\pgfplotsset{compat=newest}
\usepgfplotslibrary{fillbetween}
\usetikzlibrary{positioning}
\usetikzlibrary{decorations.pathreplacing}
\usetikzlibrary{decorations,arrows}
\usetikzlibrary{decorations.markings}
\usetikzlibrary{patterns}
\usetikzlibrary{quotes,angles}
\usetikzlibrary{arrows}

\numberwithin{equation}{section}
\title{Generalized bulk-interface correspondence for non-quantized spin transport}

\author{Jiayu Qiu\thanks{Department of Mathematics, 
 HKUST,  Clear Water Bay, Kowloon, Hong Kong SAR, China, jqiuaj@connect.ust.hk.}\,\,\,\, Hai Zhang\thanks{Department of Mathematics, HKUST, Clear Water Bay, Kowloon, Hong Kong SAR, China, haizhang@ust.hk. HZ was partially supported by Hong Kong RGC grant GRF 16307024 and NSFC grant 12371425.}}

\date{\today}

\newtheorem{theorem}{Theorem}[section]
\newtheorem{lemma}[theorem]{Lemma}

\newtheorem{definition}[theorem]{Definition}

\newtheorem{corollary}[theorem]{Corollary}

\newtheorem{proposition}[theorem]{Proposition}
\newtheorem{assumption}[theorem]{Assumption}
\allowdisplaybreaks[4]

\begin{document}

\maketitle

\begin{abstract}
This paper establishes a rigorous mathematical framework for a generalized bulk-interface correspondence (BIC) in electronic systems with possibly nonconserved spin charge, where the Hamiltonian and spin operator do not commute. We first introduce the bulk spin conductance as a character of the bulk medium, which is defined as a potential-current correlation function and is not quantized if the spin charge is nonconserved. Then we establish the principle of BIC, which states that the difference of bulk spin conductances across an interface equals the sum of two quantities associated with the spin transport along the interface: the spin-drift conductance, which captures spin transport carried by interface modes, and the spin-torque conductance, which accounts for spin generation near the interface due to the non-conservation of spin. Furthermore, when the spin charge is conserved, our result recovers the existing BIC based on the spin Chern number or Fu-Kane-Mele $\mathbb{Z}_2$ index. Our findings demonstrate that the principle of BIC is not restricted to systems with quantized characters and provides new insights into spin transport phenomena.
\end{abstract}

\section{Introduction and main results}
\subsection{Introduction}
Starting from the discovery of the celebrated quantum Hall effect, the study of topological matter and topological systems has garnered significant interest due to their exotic properties and unique physical phenomena \cite{Klitzing80qhe,vonKlitzing2020forty_years,stone1992quantum}. For instance, systems with broken time-reversal (TR) symmetry are often characterized by an integer-valued topological invariant known as the Chern number \cite{Vanderbilt2019Berry_phase}. One of the most remarkable phenomena associated with this characterization is the emergence of edge states at the interface between two systems with different Chern numbers, known as the bulk-interface correspondence (BIC) principle. The principle of BIC states that the net number of edge states equals precisely the difference of Chern numbers of the two bulk systems. The rigorous mathematical theory on the principle of BIC based on the Chern number is classified into two groups: the first group studies electronic systems in discrete structures (the so-called tight-binding models) \cite{avila2013shortrange+transfer,elgart2005shortrange+functional,graf2018shortrange+transfer,graf2013shortrange+scattering,Kellendonk02landau+ktheory,ludewig2020shortrange+coarse,EG2002bec_discrete_functional,drouot2024bec_curvedinterfaces,de2016spectral_flow,Bra2019bec_discrete_K,BKR2017bec_discrete_K,AMZ2020bec_discrete_K,kubota2017bec_discrete_K}, and the second focuses on continuous systems which are governed by differential operators, including Schrödinger operators \cite{taarabt2014landau+functional,cornean2021landau+functional,CG2005bec_schrodinger_functional,DGR2011bec_schrodinger_functional,Kellendonk2004landau+ktheory,BR2018bec_continuous_K,gontier2023edge}, Dirac operators \cite{bal2019dirac+functional,QB2024bec_dirac_functional,bal2022bec_dirac_functional_symbol}, and those arising in other physical contexts such as photonics and phononics \cite{qiu2025bulk,lin2022transfer,thiang2023transfer,drouot2021microlocal}. When the systems possess an unbroken TR symmetry, the topological characterization based on the Chern number fails. In that case, a $\mathbb{Z}_2$ topological character is available for certain systems with an
internal degree of freedom (DOF) such as spin, particle-hole, or sublattice structure \cite{de2016spectral_flow}, and the principle of BIC based on the $\mathbb{Z}_2$ index has been mathematically studied in \cite{avila2013shortrange+transfer,BKR2017bec_discrete_K,AMZ2020bec_discrete_K,kubota2017bec_discrete_K}. Beyond the intriguing mathematics associated with the quantized index, the principle of BIC has attracted attention for its practical application: determined by the quantized index, the edge states are robust under perturbation applied to the system, which leads to novel devices used in energy and information transportation \cite{bernivig13topo_insulator}.

Nonetheless, the BIC based on the Chern number or $\mathbb{Z}_2$ index does not apply when a system lacks such quantized characters. For example, consider an electronic system that preserves TR symmetry and lacks internal DOF. The first condition ensures that the Chern number vanishes, while the second prevents the $\mathbb{Z}_2$ index from serving as a well-defined bulk topological invariant \cite{de2016spectral_flow}.

%Nevertheless, the practical value of the BIC based on the Chern number or $\mathbb{Z}_2$ index, which characterizes the so-called strong topological insulator (STI) \cite{de2016spectral_flow}, is limited by the difficulty in obtaining systems with a quantized topological character in real experiments. For example, achieving a nontrivial Chern character in photonic systems usually requires huge external magnetic fields to break the time-reversal symmetry and to generate the magneto-optic effect \cite{ozawa19topological_photonics}. Such large fields are challenging to generate in practical applications and can introduce nonlinearities that make the optical properties of the system difficult to control. On the other hand, the characterization based on the $\mathbb{Z}_2$ index relies on the internal degree of freedom (DOF) such as spin, particle-hole, or sublattice structure, making it difficult to generalize to physical systems lacking these DOFs \cite{de2016spectral_flow}. 
%Furthermore, for interacting systems (e.g., fractional quantum Hall systems), achieving a quantized characterization is even more challenging \cite{tong2016lecturesquantumhalleffect}.

\begin{figure}
\centering
\begin{tikzpicture}[scale=0.35]
%left bulk
\foreach \x in {-8,-6,-4} { 
        \foreach \y in {-7,-5,-3,-1,1,3,5,7} {
            \draw[fill=black,opacity=1] (\x,\y) ellipse(0.5 and 0.5);
            \draw[dashed] ({\x-1},{\y+1}) rectangle ({\x+1},{\y-1});
        }
    }
\node[above,scale=0.8] at (-6,9) {Left bulk};

%right bulk
\foreach \x in {4,7} { 
        \foreach \y in {7,2,-3} {
            \draw[fill=black,opacity=1] (\x,\y) ellipse(0.3 and 0.8);
            \draw[fill=black,opacity=1] ({\x+1},{\y-3}) ellipse(0.8 and 0.3);
            \draw[dashed] ({\x-1},{\y+1}) rectangle ({\x+2},{\y-4});
        }
    }
\node[above,scale=0.8] at (6,9) {Right bulk};

%interface region
\draw[dashed] (0,-8.5)--(0,8.5);
%%disk elements
\draw[fill=black,opacity=1] (-2,6) ellipse(0.5 and 0.5);
\draw[fill=black,opacity=1] (1,3) ellipse(0.5 and 0.5);
\draw[fill=black,opacity=1] (-1,-3) ellipse(0.5 and 0.5);
\draw[fill=black,opacity=1] (1,-4) ellipse(0.5 and 0.5);
%%ellipse elements
\draw[fill=black,opacity=1] (0,2) ellipse(0.3 and 0.8);
\draw[fill=black,opacity=1] (2,1) ellipse(0.3 and 0.8);
\draw[fill=black,opacity=1] (-2,-2) ellipse(0.3 and 0.8);
\draw[fill=black,opacity=1] (-1,-7) ellipse(0.3 and 0.8);
\draw[fill=black,opacity=1] (-1,5) ellipse(0.8 and 0.3);
\draw[fill=black,opacity=1] (-2,2) ellipse(0.8 and 0.3);
\draw[fill=black,opacity=1] (0,-5) ellipse(0.8 and 0.3);
\draw[fill=black,opacity=1] (1,-6) ellipse(0.8 and 0.3);
%%square elements
\draw[fill=black,opacity=1] ({-1-0.3},{4+0.3}) rectangle ({-1+0.3},{4-0.3});
\draw[fill=black,opacity=1] ({1-0.3},{6+0.3}) rectangle ({1+0.3},{6-0.3});
\draw[fill=black,opacity=1] ({0-0.3},{-1+0.3}) rectangle ({0+0.3},{-1-0.3});
\draw[fill=black,opacity=1] ({2-0.3},{-3+0.3}) rectangle ({2+0.3},{-3-0.3});

\node[right,scale=0.7] at (0,0) {$(0,0)$};
\node[above,scale=0.8] at (0,9) {Interface region};
\end{tikzpicture}
\caption{An interface model: two bulk mediums are joined along an interface. The bulk media are assumed to have rectangular lattice structures, and impurities are allowed near the interface.}
\label{fig_interface_model}
\end{figure}
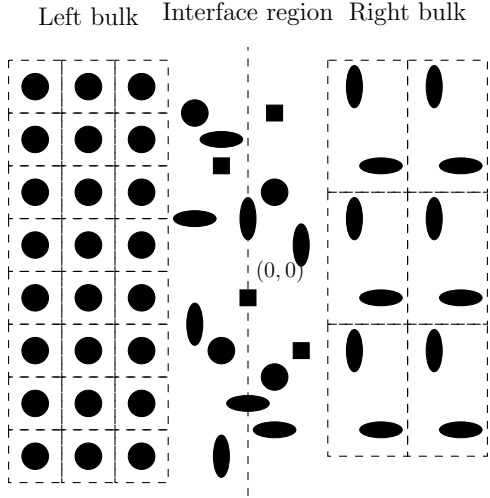

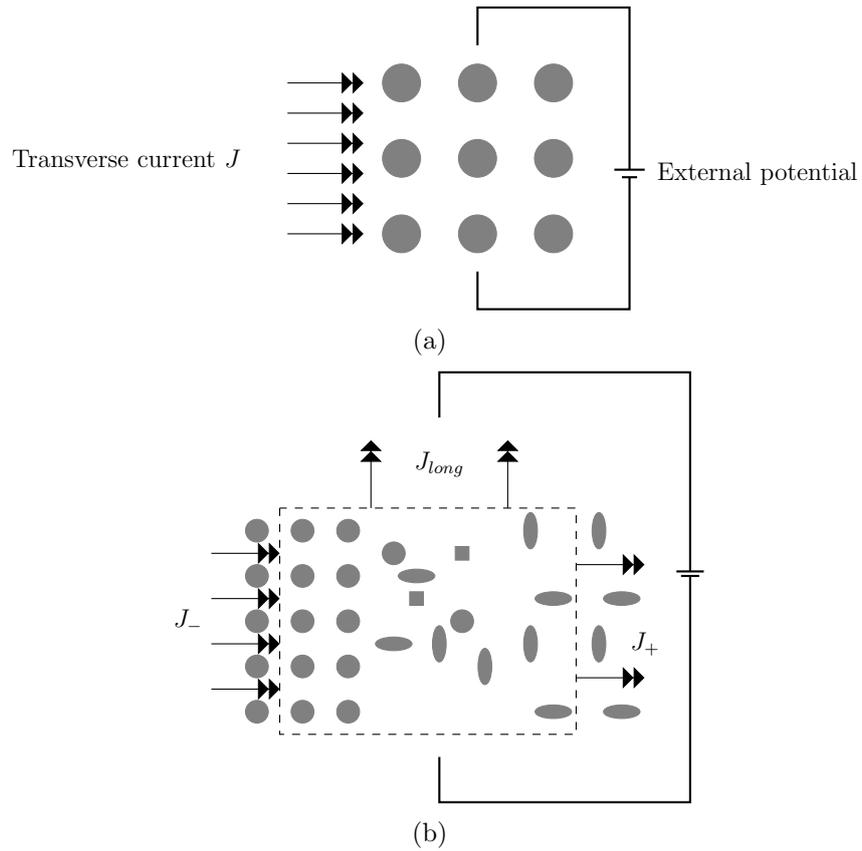
\begin{figure}
\centering
\subfigure[]{ 
\label{fig_transverse current}
\begin{tikzpicture}[scale=0.5]
%bulk
\foreach \x in {-2,0,2} { 
        \foreach \y in {0,2,4} {
            \draw[fill=black,opacity=0.5] (\x,\y) ellipse(0.5 and 0.5);
        }
    }
%voltage
\draw[thick] (0,-1)--(0,-2)--(4,-2)--(4,1.5);
\draw[thick] (3.8,1.5)--(4.2,1.5);
\draw[thick] (3.6,1.7)--(4.4,1.7);
\draw[thick] (4,1.7)--(4,6)--(0,6)--(0,5);
\node[right,scale=0.8] at (4.5, 1.6) {External potential};

%induced current
\foreach \y in {0,0.8,1.6,2.4,3.2,4} { 
        \draw[>=triangle 90, ->>] (-5,\y)--(-3,\y);
    }
\node[left,scale=0.8] at (-6,2) {Transverse current $J$};
\end{tikzpicture}
}
\subfigure[]{ 
\label{fig_thought experiment}
\begin{tikzpicture}[scale=0.3]
%left bulk
\foreach \x in {-8,-6,-4} { 
        \foreach \y in {-1,1,3,5,7} {
            \draw[fill=black,opacity=0.5] (\x,\y) ellipse(0.5 and 0.5);
        }
    }
    
%right bulk
\foreach \x in {4,7} { 
        \foreach \y in {7,2} {
            \draw[fill=black,opacity=0.5] (\x,\y) ellipse(0.3 and 0.8);
            \draw[fill=black,opacity=0.5] ({\x+1},{\y-3}) ellipse(0.8 and 0.3);
        }
    }

%interface region
%%disk elements
\draw[fill=black,opacity=0.5] (-2,6) ellipse(0.5 and 0.5);
\draw[fill=black,opacity=0.5] (1,3) ellipse(0.5 and 0.5);
%%ellipse elements
\draw[fill=black,opacity=0.5] (0,2) ellipse(0.3 and 0.8);
\draw[fill=black,opacity=0.5] (2,1) ellipse(0.3 and 0.8);
\draw[fill=black,opacity=0.5] (-1,5) ellipse(0.8 and 0.3);
\draw[fill=black,opacity=0.5] (-2,2) ellipse(0.8 and 0.3);
%%square elements
\draw[fill=black,opacity=0.5] ({-1-0.3},{4+0.3}) rectangle ({-1+0.3},{4-0.3});
\draw[fill=black,opacity=0.5] ({1-0.3},{6+0.3}) rectangle ({1+0.3},{6-0.3});

%imaginary box
\draw[dashed] (-7,8) rectangle (6,-2);
%induced current
\foreach \y in {0,2,4,6} { 
        \draw[>=triangle 90, ->>] (-10,\y)--(-7,\y);
    }
\node[left,scale=0.8] at (-10,3) {$J_{-}$};
\foreach \y in {0.5,5.5} { 
        \draw[>=triangle 90, ->>] (6,\y)--(9,\y);
    }
\node[right,scale=0.8] at (8,2) {$J_{+}$};
\foreach \x in {-3,3} { 
        \draw[>=triangle 90, ->>] (\x,8)--(\x,11);
    }
\node[above,scale=0.8] at (0,9) {$J_{long}$};

%voltage
\draw[thick] (0,-3)--(0,-5)--(11,-5)--(11,5);
\draw[thick] (10.6,5)--(11.4,5);
\draw[thick] (10.4,5.2)--(11.6,5.2);
\draw[thick] (11,5.2)--(11,14)--(0,14)--(0,12);
\end{tikzpicture}
}
\caption{(a) Transverse (Hall) Current in a periodic medium induced by an external potential along the longitudinal direction. (b) A thought experiment: an external potential is applied to the interface model, and the charge transport is analyzed within an imaginary box containing the interface.}
\label{fig_collapsed}
\end{figure}

A natural question is whether it is possible to establish the principle of BIC for systems lacking a quantized characterization. To address this question, we need to revisit the Chern-based BIC in electronic systems to extract the fundamental idea and the underlying physics. The first step is to recall the physical meaning of the Chern number. Suppose the system is periodic, i.e., governed by a periodic Hamiltonian $\mathcal{H}$, and assume the Hamiltonian has a gap $\Delta$ in its spectrum. In this case, the system is said to be insulating at energies within $\Delta$ and its ground state is given by the spectral projection $P=\mathbbm{1}_{(-\infty,\inf  \Delta ]}(\mathcal{H})$. By the celebrated Kubo formula, the Chern number $\mathcal{C}$ associated with the projection $P$ corresponds to the transverse electric current measured in the ground state when the system is subjected to an external potential inducing a unit voltage drop, as illustrated in Figure \ref{fig_transverse current} \cite{tong2016lecturesquantumhalleffect}. Hence the Chern number characterizes the \textit{bulk conductance} of the system. The bulk conductance has been studied in various systems and admits the following double-commutator expression \cite{Avron1994charge,elgart2005shortrange+functional,drouot2024bec_curvedinterfaces}:
\begin{equation} \label{eq_intro_1}
\sigma=-ei\text{Tr}\Big(P\big[[P,\Lambda_1],[P,\Lambda_2]\big] \Big) \quad (e:\text{the electric charge})
\end{equation}
Here $\Lambda_i$ are the switch functions as defined in Definition \ref{def_switch functions}. The second key insight is that the emergence of interface modes is a consequence of the conservation law. Consider the thought experiment depicted in Figure \ref{fig_thought experiment}, where two media with distinct bulk conductance (denoted as $\sigma_+$ and $\sigma_-$) are joined along an interface (see Figure \ref{fig_interface_model}) and an external potential with a unit voltage drop is applied to this interface model. Within a large box containing the interface, a transverse electric current of magnitude $J_{+}=\sigma_+ E$ will flow outside the box from the right, while an electric current of magnitude $J_{-}=\sigma_- E$ flows inside from the left. Since there is no accumulation of electric charge in the steady state \cite{Vanderbilt2019Berry_phase}, there must be longitudinal electric current fluxing outside the box to balance the difference of transverse currents $J_{-}-J_{+}=(\sigma_- -\sigma_+)E$ (recall that $\sigma_+\neq \sigma_-$). Since there are no transport channels in the bulk (recall that the bulk medium is insulating), this longitudinal current must be carried by modes localized near the interface. This reasoning illustrates the emergence of interface modes. Informally, the principle of BIC can be summarized as the following formula
\begin{equation} \label{eq_intro_2}
\text{Interface modes}
\overset{\text{conservation law}}{\Longleftrightarrow}
\text{Bulk conductance $\sigma_{-}-\sigma_{+}$}
\overset{\text{Kubo formula}}{=}
\text{Chern number $\mathcal{C}_{-}-\mathcal{C}_{+}$}
\end{equation}

In conclusion, the mechanism underlying the principle of BIC is the conservation law, which does not presume a nontrivial topological characterization of the system. The quantization of the bulk conductance (the Chern number) arises because it describes the quantum transport of a special quantity: the electric charge. However, other quantities, such as spin, can also be transported in electronic systems along with the motion of electrons. For example, the spin conductance, which measures the transverse transport of spin induced by an external electric field, is well-defined and contains a Kubo-like term \cite{monaco2020spin,marcelli2021new,Marcelli2019spin_conductivity,Giovanna22charge_to_spin}
\begin{equation} \label{eq_intro_3}
\sigma^{S,Kubo}= i\text{Tr}\Big(P\big[[P,S\Lambda_1],[P,\Lambda_2]\big] \Big) \quad (S:\text{the spin operator})
\end{equation}
We note that: 1) the formula \eqref{eq_intro_3} reduces to \eqref{eq_intro_1} by replacing $S$ with $-e\mathbbm{1}$ (highlighting the special role of electric charge as it corresponds to the identity operator in quantum mechanics), and 2) the Kubo-like expression \eqref{eq_intro_3} contributes only partially to the bulk spin conductance; see the discussion in Section 1.2 and the references \cite{monaco2020spin,marcelli2021new,Giovanna22charge_to_spin}. In the special case that the spin operator commutes with the Hamiltonian, i.e. $[\mathcal{H},S]=0$ (which physically implies the total spin charge is conserved), the spin conductance \eqref{eq_intro_3} is quantized, and its parity (even or odd) corresponds to the well-known Fu-Kane-Mele $\mathbb{Z}_2$ index \cite{Marcelli2019spin_conductivity}; see Proposition \ref{prop_kubo_formula}. However, in the general case $[\mathcal{H},S]\neq 0$, as is typical for models in the study of quantum spin Hall effect such as the Bernevig-Hughes-Zhang (BHZ) model and Kane-Mele model with nonzero spin-orbit coupling \cite{BHZ06,KaneMele05Z_2}, the spin conductance is not quantized \cite{monaco2020spin}. In this case, a rigorous mathematical framework of BIC explaining the quantum transport based on the bulk property remains elusive. Nevertheless, from the physical perspective, we infer that the principle of BIC can indeed be established in these systems lacking a quantized characterization due to the following reasons:
\begin{itemize}
    \item[i)] The physical quantity (spin) flows with the interface modes, and its associated bulk conductance is well defined.
    \item [ii)] The conservation law is universally valid for the transport of physical quantities.
\end{itemize}

In this paper, we establish a generalized BIC for electron spin systems with TR symmetry and possibly nonconserved spin charge. Such spin systems lack the conventional topological characterization (e.g., Chern number or $\mathbb{Z}_2$ index) due to the TR symmetry and non-commutation between the Hamiltonian and the spin operator. 
%, if we preclude any other internal DOF like the particle-hole and sublattice DOF. 
We first define the bulk spin conductance as a potential-current correlation function; see Definition \ref{def_spin_conducatance} and the discussion below it. Remarkably, our definition of spin conductance recovers the well-known Kubo formula when the total spin is conserved; see Proposition \ref{prop_kubo_formula}. Next, we introduce the interface spin-drifting conductance (Definition \ref{def_interface_spin_drift}), which is well understood and describes the spin transport carried by possible interface modes, and more importantly, the interface spin-torque conductance (Definition \ref{def_interface_torque}). This spin-torque conductance describes the spin generation near the interface driven by the external potential, whose existence is a natural consequence of the non-conservation of spin. Finally, we prove the main result of this paper, the principle of BIC, which states that the difference of bulk spin conductances at two sides of an interface model equals precisely the interface spin-drift conductance plus the spin-torque conductance; see Theorem \ref{thm_bic}. Our results demonstrate that the interface between two bulk media with different spin conductances remains physically nontrivial, even in the absence of a quantized bulk characterization. Specifically, for any Fermi energy within the band gap, either interface transport channels must emerge or spin torque will be observed; see the discussion under Theorem \ref{thm_bic}. This idea is discussed in the physics literature \cite{jo2024spintronics} and is rigorously justified by our theorem. When the spin is conserved, Theorem \ref{thm_bic} recovers the previous result on the principle of BIC based on the spin Chern number or Fu-Kane-Mele $\mathbb{Z}_2$ index; see \cite[Theorem 3]{avila2013shortrange+transfer}. We expect that our result can be extended to other lattice structures as well as continuum models; see Section 1.3 and the reference therein.

\subsection{Setting and Main results}
Throughout the paper, we consider non-interacting electrons moving on a discrete set on the plane, i.e. $\mathcal{D}\subset \mathbb{R}^2$. The wave functions of electrons form the Hilbert space $H_{\mathcal{D}}=\ell^2(\mathcal{D})\otimes \mathbb{C}^2$, where $\mathbb{C}^2$ represents the spin DOF. The motion of the electron is governed by the Hamiltonian $\mathcal{H}$, which is a bounded self-adjoint operator on $H_{\mathcal{D}}$. All the Hamiltonians considered in this paper are assumed to be tight-binding, or equivalently, near-sighted or exponentially short-ranged \cite{Marcelli2019spin_conductivity,drouot2024bec_curvedinterfaces,Kohn96density}:
\begin{definition}[tight-binding operator] \label{def_tight_bind}
A bounded operator $A$ on $H_{\mathcal{D}}$ is called tight-binding if there exists $\lambda,C_{\lambda}>0$ such that the matrix elements of $A$ satisfy $|A(\bm{x},\bm{y})|\leq C_{\lambda}e^{-\lambda\|\bm{x}-\bm{y}\|_{1}}$ for any $\bm{x},\bm{y}\in\mathcal{D}$, where $\|\cdot\|_1$ denotes the $l^1$-norm of $\mathbb{R}^2$. 
\end{definition}
In particular, we are interested in the motion of elections in an interface model, where two different periodic media (referred to as the two bulk materials) are separated by an interface along the axis $\{x_1=0\}$, as shown in Figure \ref{fig_interface_model}. For simplicity, we assume both two bulks are periodic to rectangular lattices, denoted as
\begin{equation*}
\mathcal{D}_{\pm}=\mathbb{Z}(a^{\pm}_{1}\bm{e}_1)\oplus \mathbb{Z}(a^{\pm}_{2}\bm{e}_2) \quad (a^{\pm}_{1},a^{\pm}_{2}\in\mathbb{N})
\end{equation*}
where the subscripts $+/-$ represent the right/left bulk, respectively. The interface structure is denoted as $\mathcal{D}_{e}$. As illustrated in Figure \ref{fig_interface_model}, we assume
\begin{equation} \label{eq_interface_structure}
\mathcal{D}_{e}\cap \{\pm x_1>L\}=\mathcal{D}_{\pm}\cap \{\pm x_1>L\} \quad \text{for some } L>0 .
\end{equation}
Note that this setup allows for the presence of defects near the interface, as shown in Figure \ref{fig_interface_model}. The Hamiltonian modeling the electron's motion in the interface model is denoted as $\mathcal{H}_{e}$, while the bulk Hamiltonians are denoted as $\mathcal{H}_{\pm}$, respectively. We assume that the interface Hamiltonian differs from the bulk Hamiltonians only by a tight-binding term $H_{\delta}$ near the interface as follows (see also \cite{elgart2005shortrange+functional,drouot2024bec_curvedinterfaces}): 
\begin{equation} \label{eq_interface_hamiltonian}
\mathcal{H}_{e}=\mathbbm{1}_{\{x_1>L\}}\mathcal{H}_{+}\mathbbm{1}_{\{x_1>L\}}+\mathbbm{1}_{\{x_1<-L\}}\mathcal{H}_{+}\mathbbm{1}_{\{x_1<-L\}}+\mathcal{H}_{\delta}.
\end{equation}
Here $\mathcal{H}_{\delta}$ is tight-binding in the sense of Definition \ref{def_tight_bind} and 
\begin{equation} \label{eq_interface_error}
|\mathcal{H}_{\delta}(\bm{x},\bm{y})|\leq C_{\nu}e^{-\nu(|x_1|+|y_1|)}
\end{equation}
for some $\nu,C_{\nu}>0$ and all $\bm{x},\bm{y}\in \mathcal{D}_{e}$. Furthermore, we assume the two bulks are insulating at a common energy level, i.e., $\mathcal{H}_{+}$ and $\mathcal{H}_{-}$ possess a common spectral gap. Specifically, we summarize these assumptions as follows.
\begin{assumption} \label{assum_bulk_hamiltonian}
\begin{itemize}
    \item[i)] The bulk Hamiltonians $\mathcal{H}_{\pm}$ defined on the lattice $\mathcal{D}_{\pm}$ are tight-binding. Without loss of generality, we assume that $\mathcal{H}_{*}(\bm{x},\bm{y})$ ($*\in\{+,-,e,\delta\}$) satisfy the estimate in Definition \ref{def_tight_bind} with the same decay length $\lambda>0$, and that $\nu=\lambda$ in \eqref{eq_interface_error}.
    \item[ii)] $\mathcal{H}_{\pm}$ are periodic to the lattice $\mathcal{D}_{\pm}$, respectively. That is $\mathcal{H}_{\pm}(\bm{x}+\bm{n},\bm{y}+\bm{n})=\mathcal{H}_{\pm}(\bm{x},\bm{y})$ for all $\bm{n}\in \mathcal{D}_{\pm}$.
    \item[iii)] $\mathcal{H}_{\pm}$ possess a common spectral gap, i.e. there exists an open interval $\Delta\subset \mathbb{R}$ such that $\Delta\cap \text{Spec}(\mathcal{H}_{+})=\Delta\cap \text{Spec}(\mathcal{H}_{-})=\emptyset$.
\end{itemize}
\end{assumption}
This paper aims to relate the spin transport of electrons near the interface to the properties of the bulk Hamiltonians, as captured by the bulk spin conductance \cite{shi2006proper,monaco2020spin,Giovanna22charge_to_spin,marcelli2021new,xiao2021conserved}. Physically, the bulk spin conductance characterizes the transverse spin current observed in the ground state when the system is subjected to an external potential with a unit voltage drop, as illustrated in Figure \ref{fig_transverse current}. The mathematical formulation of bulk spin conductance has been studied in \cite{monaco2020spin,Giovanna22charge_to_spin,marcelli2021new}. Here, we propose a new definition for spin conductance that is particularly convenient for studying BIC.

To begin, we introduce some notions. The first is the spin operator on the Hilbert space $H_{\mathcal{D}}=\ell^2(\mathcal{D})\otimes \mathbb{C}^2$, defined as $S=\text{Id}\otimes \frac{1}{2}\sigma_{z}$ ($\sigma_{z}$ is the third Pauli matrix). As defined, the spin operator $S$ only acts on the spin DOF. The second notion is the density function of the ground state: let $\rho\in C^{\infty}_{c}(\mathbb{R})$ be a smooth function that satisfies
\begin{equation} \label{eq_density_function}
\rho(x)=1 \quad \text{for }x\in (b,\inf \Delta),\quad \rho(x)=0 \quad \text{for }x\in(\sup \Delta,+\infty)
\end{equation}
where $b=\inf\big\{ \cup_{*\in\{+,-,e\}}\text{Spec}(\mathcal{H}_{*})\big\}$ denotes the infimum of the spectra of all Hamiltonians under consideration. Hence $\rho(\mathcal{H}_{\pm})$ represents the ground state of $\mathcal{H}_{\pm}$ when the Fermi level lies in the gap $\Delta$; in fact, $\rho$ is the mollified version of the Heaviside function with jump at the Fermi level. The third notion is the principal-value trace along the $x_1$-direction:
\begin{definition}[principal-value trace along $x_1$] \label{def_pv_trace}
Let $A$ be an operator acting on $H_{\mathcal{D}}$ such that $\mathbbm{1}_{\Omega_{n,a}\cap \mathcal{D}}A$ is trace-class on $H_{\mathcal{D}}$ for any $n\in\mathbb{N}$ and some $a>0$ ($\mathbbm{1}_{\Omega_{n,a}\cap \mathcal{D}}$ is the indicator function of $\Omega_{n,a}:=(n a,(n+1) a)\times \mathbb{R}$, a strip with width equal to $a$). The principal-value trace of $A$ along $x_1$ (in domain $\mathcal{D}$ with step size $a$) is defined as follows whenever the limit exists:
\begin{equation*}
Tr_{a,\mathcal{D}}^{pv,1}(A)
:=\lim_{N\to\infty}\sum_{n=-N}^{N-1}Tr_{\mathcal{D}}(\mathbbm{1}_{\Omega_{n,a}\cap \mathcal{D}}A)
\end{equation*}
where $Tr_{\mathcal{D}}$ is the conventional trace associated with trace-class operators on $H_{\mathcal{D}}$.
\end{definition}
The principal-value trace is studied in detail in \cite{Marcelli2019spin_conductivity,marcelli2021new} as well as in \cite{Avron1994charge}, which is introduced as the weak version of trace for certain operators that are not trace-class. It serves as an important tool for our study of BIC. We briefly review its basic properties in Section 2, with an emphasis on its link to the conventional trace and its action on periodic operators. The last notion is the switch function:
\begin{definition}[switch functions] \label{def_switch functions}
Fix $j\in\{1,2\}$. A switch function in the $x_j$-direction is a measurable function $\Lambda_{j}:\mathbb{R}^2\to [0,1]$ that depends only on the variable $x_j$ and satisfies
\begin{equation*}
\Lambda_{j}(\bm{x})=1\quad \text{for }x_j>b_1,\quad
\Lambda_{j}(\bm{x})=0\quad \text{for }x_j<b_2,
\end{equation*}
for some real numbers $b_1>b_2$.
\end{definition}
With these notions, we define the spin conductance of the bulk Hamiltonians as follows.
\begin{definition}[Bulk spin conductance] \label{def_spin_conducatance}
Let $\mathcal{H}_{\pm}$ be the bulk Hamiltonians satisfying Assumption \ref{assum_bulk_hamiltonian}, and $\rho$ be the density function satisfying \eqref{eq_density_function}. The spin conductance corresponding to the ground state $\rho(\mathcal{H}_{\pm})$ and the switch function $\Lambda_2$ is defined as
\begin{equation} \label{eq_spin_conduc_1}
\sigma_{\pm}^{\Lambda_2}:=\frac{1}{2\pi}Tr_{ka_1^{\pm},\mathcal{D}_{\pm}}^{pv,1}(\Sigma_{\pm}^{\Lambda_2,\Lambda_1 S})
\quad (k\in\mathbb{N})
\end{equation}
with
\begin{equation} \label{eq_spin_conduc_2}
\begin{aligned}
\Sigma_{\pm}^{\Lambda_2,\Lambda_1 S}
:=\int_{\mathbb{C}}dm(z)\frac{\partial \tilde{\rho}}{\partial \overline{z}}\Big[
&R_{\pm}(z) \big[\mathcal{H}_{\pm},\Lambda_2\big] R_{\pm}(z) \big[\mathcal{H}_{\pm},\Lambda_1 S\big] R_{\pm}(z) \\
&- R_{\pm}(z) \big[\mathcal{H}_{\pm},\Lambda_1 S\big] R_{\pm}(z) \big[\mathcal{H}_{\pm},\Lambda_2\big] R_{\pm}(z) \Big],
\end{aligned}
\end{equation}
where $R_{\pm}(z)=(\mathcal{H}_{\pm}-z)^{-1}$ is the resolvent of the bulk Hamiltonian, $\tilde{\rho}$ is the almost-analytic extension of $\rho$ (introduced in Section 2.2), $dm(z)$ is the Lebesgue measure on the complex plane.
\end{definition}
Before showing $\sigma_{\pm}^{\Lambda_2}$ is well-defined, we briefly explain this definition and in particular, remark on the physics idea underlying this expression. We first note that the two commutators in \eqref{eq_spin_conduc_2}, i.e. $[\mathcal{H}_{\pm},\Lambda_2]$ and $[\mathcal{H}_{\pm},\Lambda_1 S]$ have distinct physics meanings. $[\mathcal{H}_{\pm},\Lambda_2]=-i\partial_{t}\big|_{t=0}e^{it\mathcal{H}_{\pm}}\Lambda_2 e^{-it\mathcal{H}_{\pm}}$ generates the time evolution when the system $\mathcal{H}_{\pm}$ is subjected to an external electric potential that induces a unit voltage drop along the $x_2$-direction as in Figure \ref{fig_transverse current}. It's the driving force of all observed physical phenomena inside the system. On the other hand, $[\mathcal{H}_{\pm},\Lambda_1 S]$ corresponds exactly to the observable of our interest: the spin current. Loosely speaking, $\Lambda_1 S$ measures the total spin (or the so-called `spin charge' \cite{ezawa2014symmetry}) inside the support of $\Lambda_1$ (i.e. a half plane), and its time-derivative $[\mathcal{H}_{\pm},\Lambda_1 S]$ characterizes the spin transport between the two half planes $\{x_1>0\}$ and $\{x_1<0\}$. 
%(This illustration is not rigorous since the total spin in an unbounded region can be infinite; nonetheless, it just captures the main idea). 
In other words, Definition \ref{def_spin_conducatance} characterizes the \textit{correlation between external potential and the spin current measured in the ground state of the bulk Hamiltonian}, or \textit{the spin-current response driven by the external potential}. The correlation form \eqref{eq_spin_conduc_2} also appeared in the study of charge transportation \cite{elgart2005shortrange+functional}, although it was not explored there in detail. We find this formulation particularly convenient for analyzing BIC in the context of spin transport.

Mathematically, one must utilize the principal-value trace rather than the conventional trace to define the bulk spin conductance, since the operator $\Sigma_{\pm}^{\Lambda_2,\Lambda_1 S}$ in \eqref{eq_spin_conduc_2} is not trace-class. Intuitively, this is because $\Sigma_{\pm}^{\Lambda_2,\Lambda_1 S}$ lacks localization in the $x_1$-direction, a consequence of the possible non-commutation $[\mathcal{H}_{\pm},S]\neq 0$. In fact, by the Leibniz rule
\begin{equation} \label{eq_leibniz}
[\mathcal{H}_{\pm},\Lambda_1 S]=[\mathcal{H}_{\pm},\Lambda_1]S+\Lambda_1 [\mathcal{H}_{\pm},S].
\end{equation}
The second term on the right does not decay as $x_1\to \infty$, which prevents $\Sigma_{\pm}^{\Lambda_2,\Lambda_1 S}$ to be trace-class. This contrasts with the well-understood case of electric charge conductance \cite{Avron1994charge,elgart2005shortrange+functional,drouot2024bec_curvedinterfaces}. The reason that the principal-value trace of $\Sigma_{\pm}^{\Lambda_2,\Lambda_1 S}$ is well-defined is due to the physical observation that the spin torque (corresponding to the second term of \eqref{eq_leibniz}) vanishes mesoscopically in the bulk; see the discussion under Proposition \ref{prop_spin_conduct_justified}.

We also note that spin conductance has been extensively studied from the perspective of linear response theory in \cite{marcelli2021new,Giovanna22charge_to_spin,monaco2020spin,Marcelli2019spin_conductivity}, employing the concept of non-equilibrium almost-stationary states (NEASS). Although their definition of spin conductance differs from ours, the two approaches share certain similarities; in particular, both definitions coincide with the well-known Kubo formula when total spin is conserved (see Proposition \ref{prop_kubo_formula}). We anticipate a deeper connection between these definitions, the exploration of which we leave to interested readers.

The spin conductance defined in Definition \ref{def_spin_conducatance} is justified as follows. The proof is left to Section 4.
\begin{proposition} \label{prop_spin_conduct_justified}
For each switch function $\Lambda_2$, $\sigma_{\pm}^{\Lambda_2}$ in Definition \ref{def_spin_conducatance} is well-defined as the the limit \eqref{eq_spin_conduc_1} exists and is independent of the choice of switch function $\Lambda_1$, density function $\rho$ and the step size $k\in\mathbb{N}$.
\end{proposition}

We remark that $\sigma_{\pm}^{\Lambda_2}$ generally depends on the choice of switch function $\Lambda_2$ because the total spin is not necessarily conserved, i.e. $[\mathcal{H}_{\pm},S]\neq 0$. In that case, a local potential defect may induce spin torque in the bulk, and hence nonzero spin current by \eqref{eq_leibniz}. For the special case of $[\mathcal{H},S]=0$, we show that the spin conductance $\sigma_{\pm}^{\Lambda_2}$ recovers the familiar Kubo formula \cite{marcelli2021new,Marcelli2019spin_conductivity} (Proposition \ref{prop_kubo_formula}) and it is independent of $\Lambda_2$.

We briefly sketch the proof of Proposition \ref{prop_spin_conduct_justified}. First, we select a special switch function $\Lambda^{\star}_1$ that is constant on each vertical strip $\Omega_{n,ka_1^{\pm}}\cap \mathcal{D}_{\pm}$ of the lattice $\mathcal{D}_{\pm}$ (for example, the Heaviside function in the first variable, i.e., $\Lambda^{\star}_1(\bm{x})=H(x_1)$), then to show that the limit \eqref{eq_spin_conduc_1} exists for $\Lambda_1=\Lambda^{\star}_1$. To this end, we move $\Lambda_1^{\star}$ in \eqref{eq_spin_conduc_2} cyclically to the right, obtaining
\begin{equation} \label{eq_main_idea_proof_bulk_conduct_3}
\begin{aligned}
R_{\pm}(z) \big[\mathcal{H}_{\pm},\Lambda_2\big] R_{\pm}(z) \big[\mathcal{H}_{\pm},S\Lambda^{\star}_1 \big] R_{\pm}(z)
&=R_{\pm}(z) \big[\mathcal{H}_{\pm},\Lambda_2\big] R_{\pm}(z) \Big(\big[\mathcal{H}_{\pm}, S\big]\Lambda^{\star}_1+S\big[\mathcal{H}_{\pm},\Lambda^{\star}_1 \big]\Big) R_{\pm}(z) \\
&=R_{\pm}(z) \big[\mathcal{H}_{\pm},\Lambda_2\big] R_{\pm}(z) \big[\mathcal{H}_{\pm}, S\big] R_{\pm}(z)\Lambda^{\star}_1
+\text{I}+\text{II}
\end{aligned}
\end{equation}
with
\begin{equation*}
\begin{aligned}
&\text{I}=R_{\pm}(z) \big[\mathcal{H}_{\pm},\Lambda_2\big] R_{\pm}(z) \big[\mathcal{H}_{\pm}, S\big]\Big[\Lambda^{\star}_1, R_{\pm}(z)\Big], \\
&\text{II}=R_{\pm}(z) \big[\mathcal{H}_{\pm},\Lambda_2\big] R_{\pm}(z)S\big[\mathcal{H}_{\pm},\Lambda^{\star}_1 \big] R_{\pm}(z).
\end{aligned}
\end{equation*}
Both the operators $\text{I}$ and $\text{II}$ are trace-class since the commutators with $\Lambda^{\star}_1$ provide localization in the $x_1$-direction, and $\big[\mathcal{H}_{\pm},\Lambda_2\big]$ is localized in the $x_2$-direction. Hence we can write the operator $\Sigma_{\pm}^{\Lambda_2,\Lambda^{\star}_1 S}$ as
\begin{equation*}
\Sigma_{\pm}^{\Lambda_2,\Lambda^{\star}_1 S}
=\Sigma_{\pm}^{\Lambda_2, S}\Lambda^{\star}_1+(\text{trace-class operators})
\end{equation*}
where $\Sigma_{\pm}^{\Lambda_2, S}$ is defined similarly as in \eqref{eq_spin_conduc_2} by replacing $\Lambda_1 S$ with $S$:
\begin{equation} \label{eq_torque_correlation}
\begin{aligned}
\Sigma_{\pm}^{\Lambda_2,S}
:=\int_{\mathbb{C}}dm(z)\frac{\partial \tilde{\rho}}{\partial \overline{z}}\Big[
&R_{\pm}(z) \big[\mathcal{H}_{\pm},\Lambda_2\big] R_{\pm}(z) \big[\mathcal{H}_{\pm}, S\big] R_{\pm}(z) \\
&- R_{\pm}(z) \big[\mathcal{H}_{\pm},S\big] R_{\pm}(z) \big[\mathcal{H}_{\pm},\Lambda_2\big] R_{\pm}(z)
\Big].
\end{aligned}
\end{equation}
Note that the operator $[\mathcal{H}_{\pm},S]$ is the spin torque operator \cite{xiao2021conserved,sun2024nonconserved}, which measures the non-conservation of spin; hence the potential-torque correlation operator $\Sigma_{\pm}^{\Lambda_2, S}$ characterizes the spin torque response in the bulk induced by the external potential. By Proposition \ref{prop_pv_usual_trace} of Section 2, the principal-value trace of trace-class operators equals its conventional trace, which exists and is finite. Since $\Lambda^{\star}_1$ is constant on each vertical strip $\Omega_{n,ka_1^{\pm}}\cap \mathcal{D}_{\pm}$, we have
\begin{equation} \label{eq_main_idea_proof_bulk_conduct_1}
\begin{aligned}
Tr_{ka_1^{\pm},\mathcal{D}_{\pm}}^{pv,1}(\Sigma_{\pm}^{\Lambda_2,\Lambda_1 S})
&=Tr_{ka_1^{\pm},\mathcal{D}_{\pm}}^{pv,1}(\Sigma_{\pm}^{\Lambda_2, S}\Lambda^{\star}_1) +(\text{a finite quantity}) \\
&=\lim_{N\to\infty} \sum_{n=-N}^{N-1}Tr_{\mathcal{D}_{\pm}}(\mathbbm{1}_{\Omega_{n,ka_1^{\pm}}\cap \mathcal{D}_{\pm}}\Sigma_{\pm}^{\Lambda_2, S})\times (\text{value of $\Lambda^{\star}_1$ on $\Omega_{n,ka_1^{\pm}}$}) \\
&\quad +(\text{a finite quantity}).
\end{aligned}
\end{equation}
An important physical insight is that the mesoscopic average (i.e., trace over a unit cell) of the spin torque response in the bulk induced by the external potential equals zero \cite{Giovanna22charge_to_spin,Marcelli2019spin_conductivity}; see Proposition \ref{prop_vanish_bulk_torque_Kubo} and \ref{prop_vanish_bulk_torque_correlation} in Section 3. More precisely, 
\begin{equation} \label{eq_main_idea_proof_bulk_conduct_2}
Tr_{\mathcal{D}_{\pm}}(\mathbbm{1}_{\Omega_{n,ka_1^{\pm}}\cap \mathcal{D}_{\pm}}\Sigma_{\pm}^{\Lambda_2, S})=0 \quad (\forall n\in\mathbb{N}). 
\end{equation}
Therefore, the principal-value trace at the left of \eqref{eq_main_idea_proof_bulk_conduct_1} exists. This concludes the first step of the proof. Next, we show that the value of \eqref{eq_spin_conduc_1} is independent of $\Lambda_1$, completing the major part of the proof. This step is based on standard manipulation of traces for commutators. We leave all the details to Section 4.

For completeness, we demonstrate that the spin conductance defined in Definition \ref{def_spin_conducatance} is nontrivial by showing that it recovers the well-known Kubo formula when the total spin is conserved. This property is also shared by the spin conductance defined in \cite{marcelli2021new,Giovanna22charge_to_spin,monaco2020spin,Marcelli2019spin_conductivity}, which differs from ours, using the tool of non-equilibrium almost-stationary state (NEASS). It's shown clearly in their paper that the spin conductance contains terms contributed by the non-conservation of spin, which all vanish when $[\mathcal{H}_{\pm},S]=0$.
\begin{proposition} \label{prop_kubo_formula}
When $[\mathcal{H}_{\pm},S]=0$, $\sigma_{\pm}^{\Lambda_2}$ is expressed by the Kubo formula of spin conductance:
\begin{equation} \label{eq_kubo_formula}
\sigma_{\pm}^{\Lambda_2}=iTr_{\mathcal{D}_{\pm}}\Big(P_{\pm}\Big[ [P_{\pm},S\Lambda_1],[P_{\pm},\Lambda_2] \Big] \Big)
=\frac{1}{2}(\mathcal{C}(P_{\pm}^{\uparrow})-\mathcal{C}(P_{\pm}^{\downarrow}))
\end{equation}
where $P_{\pm}=\rho(\mathcal{H}_{\pm})$ is the spectral projection to the ground state of $\mathcal{H}_{\pm}$, $\mathcal{C}(P_{\pm}^{\uparrow/\downarrow})$ is the Chern number associated with the spin-up/spin-down ground state (the component of $P_{\pm}$ lying in the $\frac{1}{2}/\frac{-1}{2}$ eigenspace of the spin operator $S=\text{Id}\otimes \frac{1}{2}\sigma_z$). In that case, $\sigma_{\pm}=\sigma_{\pm}^{\Lambda_2}$ is independent of the switch function $\Lambda_2$. Moreover, when the Hamiltonian $\mathcal{H}_{\pm}$ is TR symmetric, $\sigma_{\pm} \text{ mod }2$ gives the Fu-Kane-Mele $\mathbb{Z}_2$ index.
\end{proposition}

The proof is left to Section 5.

We now introduce the quantities characterizing the spin transport near the interface. The first quantity is the interface spin-drifting conductance, which describes the spin transport carried by possible interface modes:
\begin{definition}[Interface spin-drifting conductance] \label{def_interface_spin_drift}
The interface spin-drifting conductance of the interface Hamiltonian $\mathcal{H}_{e}$ for energies in the bulk spectral gap $\Delta$ is defined as
\begin{equation} \label{eq_interface_spin_drift}
\sigma_{e}^{drift,\Lambda_2,\rho}=-iTr_{\mathcal{D}_{e}}\big(\{[\mathcal{H}_{e},\Lambda_2],S\}\rho^{\prime}(\mathcal{H}_{e}) \big),
\end{equation}
where the bracket $\{A,B\}=AB+BA$ denotes the anti-commutator.
\end{definition}
This definition is illustrated physically as follows. First, the commutator $[\mathcal{H}_{e},\Lambda_2]$ measures the number of electrons moving from the lower half plane $\{\Lambda_2(\bm{x})=0\}$ to $\{\Lambda_2(\bm{x})=1\}$ in unit time \cite{elgart2005shortrange+functional,Avron1994charge,drouot2021microlocal}. Hence the symmetrized product $\{[\mathcal{H}_{e},\Lambda_2],S\}=[\mathcal{H}_{e},\Lambda_2]S+S[\mathcal{H}_{e},\Lambda_2]$ measures the spin transport contributed by the center-of-mass drift of electrons \cite{marcelli2021new}. On the other hand, $\rho^{\prime}(\mathcal{H}_{e})$ is a density of state for possible modes in the bulk spectral gap (recall that $\rho$ is the mollified version of the Heaviside function as in \eqref{eq_density_function}; hence $\rho^{\prime}$ approximates the Dirac-delta function centering at the energy inside the gap $\Delta$). In conclusion, $\sigma_{e}^{drift,\Lambda_2,\rho}$ characterizes the spin transport due to electron drift carried by possible in-gap modes. Mathematically, it is straightforward to prove that the operator in \eqref{eq_interface_spin_drift} is trace-class following the proof of \cite[Proposition 3]{drouot2024bec_curvedinterfaces}. Intuitively, it's because $\rho^{\prime}(\mathcal{H}_{e})$ is localized in the the $x_1$-direction (since $\rho^{\prime}(\mathcal{H}_{\pm})=0$ as we assume $\rho^{\prime}$ is supported in the spectral gap) and $[\mathcal{H}_{e},\Lambda_2]$ is localized in the $x_2$-direction (because it's a commutator with $\Lambda_2$); both facts are proved in \cite{drouot2024bec_curvedinterfaces}. Since the spin operator $S=\text{Id}\otimes \frac{1}{2}\sigma_z$ does not act on the spatial DOF, it's clear that the operator $\{[\mathcal{H}_{e},\Lambda_2],S\}\rho^{\prime}(\mathcal{H}_{e})$ is localized in both directions and hence is trace-class.

However, a critical observation is that the center-of-mass drift captured by $\sigma_{e}^{drift,\Lambda_2,\rho}$ is not the only contribution to the spin transport if the spin torque operator $[\mathcal{H}_e,S]$ is nonzero. Intuitively, this is because the spin axis is not fixed but rotates along the electron's motion. We characterize this contribution as the interface spin-torque response. It is defined in the form of a correlation function as in Definition \ref{def_spin_conducatance} as follows: 
\begin{definition}[Interface spin-torque conductance]
\label{def_interface_torque}
The spin-torque conductance of the interface Hamiltonian $\mathcal{H}_{e}$ is defined as
\begin{equation} \label{eq_interface_torque_1}
\sigma_{e}^{torque,\Lambda_2,\rho}:=\frac{1}{2\pi}Tr_{a_1^{com},\mathcal{D}_{e}}^{pv,1}(\Sigma_{e}^{\Lambda_2,S})
\quad (k\in\mathbb{N})
\end{equation}
with
\begin{equation} \label{eq_interface_torque_2}
\begin{aligned}
\Sigma_{e}^{\Lambda_2, S}
:=\int_{\mathbb{C}}dm(z)\frac{\partial \tilde{\rho}}{\partial \overline{z}}\Big[
&R_{e}(z) \big[\mathcal{H}_{e},\Lambda_2\big] R_{e}(z) \big[\mathcal{H}_{e}, S\big] R_{e}(z) \\
&- R_{e}(z) \big[\mathcal{H}_{e}, S\big] R_{e}(z) \big[\mathcal{H}_{e},\Lambda_2\big] R_{e}(z)
\Big]
\end{aligned}
\end{equation}
and $a_1^{com}:=\text{lcm}(a_1^{+},a_1^{-})$ being the least common multiple of $a_1^{+}$ and $a_1^{-}$.
\end{definition}
We show that $\sigma_{e}^{torque,\Lambda_2,\rho}$ indeed characterizes the interface property:
\begin{proposition} \label{prop_interface_spin_torque_justified}
For each switch function $\Lambda_2$ and density function $\rho$, the integral kernel of $\Sigma_{e}^{\Lambda_2, S}$ decays to zero away from the interface. Moreover, $\sigma_{e}^{torque,\Lambda_2,\rho}$ is well-defined as the the limit \eqref{eq_interface_torque_1} exists.
\end{proposition}
The proof of Proposition \ref{prop_interface_spin_torque_justified} is left to Section 6. Its main idea, based on physics, is illustrated below. As is suggested by Definition \ref{def_interface_torque}, $\sigma_{e}^{torque,\Lambda_2,\rho}$ is the expectation of spin-torque response in the interface model driven by an external potential along the $x_2$-direction. Since the spin-torque response vanishes mesoscopically in the bulk by \eqref{eq_main_idea_proof_bulk_conduct_2}, the main contribution to the principal-value trace \eqref{eq_interface_torque_1} is from the near-interface region; hence \eqref{eq_interface_torque_1} is finite.

Finally, we state our main result concerning the generalized BIC. A detailed proof, along with an intuitive explanation based on the conservation law (as discussed in the Introduction), is provided in Section 7.

\begin{theorem} \label{thm_bic}
For any switch function $\Lambda_2$ and density function $\rho$ satisfying \eqref{eq_density_function}, it holds that
\begin{equation} \label{eq_bic}
\sigma_{e}^{drift,\Lambda_2,\rho}+\sigma_{e}^{torque,\Lambda_2,\rho}=\sigma_{+}^{\Lambda_2}-\sigma_{-}^{\Lambda_2}.
\end{equation}
\end{theorem}
Theorem \ref{thm_bic} demonstrates the principle of BIC can be extended to electronic systems with nonconserved spin charge. In fact, for TR symmetric systems that lack spin conservation and other DOF, such as the Bernevig-Hughes-Zhang (BHZ) model and Kane-Mele model with nonzero spin-orbit coupling \cite{BHZ06,KaneMele05Z_2}, there is currently no widely accepted, nontrivial, and quantized bulk characterization. Nevertheless,  we can still define the `character' of the bulk system as the spin conductance $\sigma_{\pm}^{\Lambda_2}$. Remarkably, the difference between these bulk characters continues to determine the quantum transport near the interface. That is, as long as $\sigma_{+}^{\Lambda_2}\neq \sigma_{-}^{\Lambda_2}$, the interface separating the two bulk media is physically nontrivial: either interface modes emerge ($\rho^{\prime}(\mathcal{H}_e)\neq 0$) or spin distribution be detected near the interface (the spin response $\sigma_{e}^{torque,\Lambda_2,\rho}\neq 0$). This result implies that the principle of BIC can be extended to a broader context, and is not limited to the topological regime. Unlike the well-understood drift transport by interface modes, the interface spin-torque is less studied. The emergence of spin torque near the interface can serve as a platform for designing spin-torque devices; see \cite{vzelezny2018spin_torque} and the references therein.

As shown in Section 7, the proof of Theorem \ref{thm_bic} reflects the physical mechanism behind the principle of BIC: the conservation law. Specifically, corresponding to the thought experiment in Figure \ref{fig_thought experiment}, the bulk spin conductance $\sigma_{-}^{\Lambda_2}/\sigma_{+}^{\Lambda_2}$ in \eqref{eq_bic} measure the bulk-to-box/box-to-bulk spin current from the left/right side, respectively. On the other hand, the spin-drift conductance $\sigma_{e}^{drift,\Lambda_2,\rho}$ measures the net longitudinal current flowing outside the imaginary box and $\sigma_{e}^{torque,\Lambda_2,\rho}$ measures the spin generation within the box. These four terms necessarily sum to zero by the conservation law.

When $[\mathcal{H}_{e},S]=0$, the interface spin-torque conductance $\sigma_{e}^{torque,\Lambda_2,\rho}$ vanishes by definition. In that case, Theorem \ref{thm_bic} recovers the previous result on the principle of BIC based on the spin Chern number or Fu-Kane-Mele $\mathbb{Z}_2$ index; see \cite[Theorem 3]{avila2013shortrange+transfer}.

\subsection{Futher discussions}

In this paper, we have assumed that both bulk media separated by the interface are periodic with respect to rectangular lattices for simplicity.  However, we expect that our results can be extended to other lattice structures, such as the honeycomb lattice. Achieving this extension would require a careful analysis of the spin-torque response, particularly verifying the validity of Proposition \ref{prop_spin_conduct_justified} and Proposition \ref{prop_interface_spin_torque_justified}. We leave this interesting direction for future research.

Another potential avenue for future study is the extension of these results to PDE settings, i.e., continuum models. A significant difference between the continuum models and the discrete models considered here is that, in this paper, (almost) all operators involved are bounded. When moving to a PDE setting, one needs to be careful about the domain of operators when they are unbounded, especially those commutators. Nonetheless, we believe this extension is achievable, as the underlying physical mechanism behind the BIC is expected to remain consistent across both continuum and discrete models.

A more interesting problem is to study the BIC associated with the transport of physical quantities other than the spin, such as the \textit{orbital angular momentum $O=x_1[\mathcal{H}_{\pm},x_2]-x_2[\mathcal{H}_{\pm},x_1]$} \cite{bhowal2021orbital}. This problem is particularly interesting as it could lead to a new framework for understanding the orbital or valley Hall effects in both electronic and photonic systems that lack internal DOF \cite{wu2017valley,xiao2007valley,lee2020valley,li2024interface,qiu2024mathematical}. However, in such cases, it is currently unclear whether we can still use an analog of \eqref{eq_spin_conduc_2} (by replacing the spin operator $S$ with $O$) to characterize the bulk medium. A primary complication arises from the position-dependence of the operator $O$, resulting in the non-commutation relation $[O,x_j]\neq 0$. This non-commutation issue makes the orbital transport differ greatly from the spin transport, causing some important steps that establish the BIC in this paper to fail. For example, the mesoscopic average of orbital-torque response (obtained by replacing $S$ by $O$ in \eqref{eq_main_idea_proof_bulk_conduct_2}) may not vanish. These issues may be resolved when the structure possesses certain special symmetries, such as reflection symmetry. We leave this problem for future study.

\subsection{Notations}

We summarize some frequently used notations in this subsection:

\begin{itemize}
    \item[i)] Vector spaces: the basic Hilbert space on the discrete set $\mathcal{D}$ is $H_{\mathcal{D}}=\ell^2(\mathcal{D})\otimes \mathbb{C}^2$, the trace ideal in $H_{\mathcal{D}}$ is denoted as $\mathscr{T}_{\mathcal{D}}$.
    \item[ii)] Norms: $\|\cdot\|_{1}$ and $\|\cdot\|_{2}$ denote the $\ell^1$ and $\ell^2$ norm in $\mathbb{R}^2$, respectively. $\text{dist}(\cdot,U)$ is the $\ell^1$ distance from a point to a set $U\subset \mathbb{R}^2$. $\|\cdot\|$ denotes the operator norm on $H_{\mathcal{D}}$. $\|\cdot\|_{\mathscr{T}_{\mathcal{D}}}$ denotes the trace norm on $\mathscr{T}_{\mathcal{D}}$.
    \item[iii)] Operators and functionals: $\mathcal{H}_{\pm}$ are the bulk Hamiltonians for the right/left bulks (Assumption \ref{assum_bulk_hamiltonian}), respectively. and $\mathcal{H}_{e}$ is the interface Hamiltonian in \eqref{eq_interface_hamiltonian}. Their resolvents are denoted as $R_{*}(z)=(\mathcal{H}_{*}-z)^{-1}$ ($*\in\{+,-,e\}$). The trace of $A\in \mathscr{T}_{\mathcal{D}}$ is denoted as $Tr_{\mathcal{D}}$. The trace of a matrix on $\mathbb{C}^2$ is denoted as $tr$.
    \item[iv)] Geometry: $\mathcal{D}$ denotes a discrete set on $\mathbb{R}^2$. In particular $\mathcal{D}_{e}$, $\mathcal{D}_{+}$, $\mathcal{D}_{-}$ are the underlying set for the interface model, right and left bulks (Figure \ref{fig_interface_model}), respectively. $\Omega_{n,a}:=(n a,(n+1) a)\times \mathbb{R}$ denotes a strip parallel to $x_2-$axis with width equal to $a$ ($n\in\mathbb{Z}$). $\Omega_{R}=[-R,R]\times \mathbb{R}$ denotes the strip centered at the origin with width equal to $R>0$. $\Omega_{\pm}=\{\pm x_1>0\}$ denote the right/left half plane, respectively.
    \item[v)] Special functions and constants: $\Lambda_i$ denote the switch function in the $x_i$-direction (Definition \ref{def_switch functions}). $\rho$ is the density function satisfying \eqref{eq_density_function}, and its almost-analytic extension is denoted as $\tilde{\rho}$. Throughout this paper, $\lambda$ denotes the decay length of Hamiltonians $\mathcal{H}_*$ for $*\in\{+,-,e,\delta\}$ (Assumption \ref{assum_bulk_hamiltonian}(i)) exclusively.
\end{itemize}

\section{Preliminaries}

\subsection{Some trace-class properties}

Let $A$ be a bounded operator on the Hilbert space $H_{\mathcal{D}}=\ell^2(\mathcal{D})\otimes \mathbb{C}^2$, with its integral kernel denoted as $A(\bm{x},\bm{y})$. We say $A$ is trace-class if its trace norm is finite: $\|A\|_{\mathscr{T}_{\mathcal{D}}}:= \sum_{k}(|A|e_k,e_k)_{H_{\mathcal{D}}}<\infty$, where $|A|=\sqrt{A^* A}$ and $e_k$ are an orthonormal basis of $H_{\mathcal{D}}$. The trace-class operators form an ideal of $H_{\mathcal{D}}$, denoted as $\mathscr{T}_{\mathcal{D}}$. For $A\in \mathscr{T}_{\mathcal{D}}$, the trace of $A$ is defined as the sum of diagonal elements, i.e. $Tr_{\mathcal{D}}(A)=\sum_{\bm{x}\in\mathcal{D}}A(\bm{x},\bm{x})$.

We summarize some basic properties of the trace-class operators as follows. The proof of Lemma \ref{lem_He-Hpm}-\ref{lem_trace_norm_kernel} and Corollary \ref{corol_trace_criterion}, except Lemma \ref{lem_resolv_ground_tight}(b), can be found in \cite[Section 2]{drouot2024bec_curvedinterfaces}. The proof of Lemma \ref{lem_resolv_ground_tight}(b), Proposition \ref{prop_cond_cyc_periodic} and \ref{prop_pv_usual_trace} can be found in \cite[Section 3 and 5]{Marcelli2019spin_conductivity}. The statements in this section are slightly different from the ones in the reference; nonetheless, all the proofs are the same.

The first Lemma follows naturally from the property of the interface Hamiltonian \eqref{eq_interface_hamiltonian}, which states that the difference of Hamiltonians $(\mathcal{H}_e-\mathcal{H}_{\pm})(\bm{x},\bm{y})$ is tight binding and decays exponentially as $x_1,y_1\to\pm \infty$.
\begin{lemma} \label{lem_He-Hpm}
There exist $C,D>0$, which depend only on the decay length $\lambda$, such that $|(\mathcal{H}_e-\mathcal{H}_{\pm})(\bm{x},\bm{y})|\leq Ce^{-D(\|\bm{x}-\bm{y}\|_{1}+\text{dist}(x,\Omega_{\mp})+\text{dist}(y,\Omega_{\mp}) )}$.
\end{lemma}

The second Lemma indicates that the resolvents and the ground state projections $P_{\pm}$ of the edge/bulk Hamiltonians are tight-binding, as is well-known:
\begin{lemma} \label{lem_resolv_ground_tight}
There exists $C,D>0$, which depend only on the decay length $\lambda$, such that
\begin{itemize}
    \item[(a)] For $|\text{Im}\, z|<1$, the resolvents $R_{*}(z)=(\mathcal{H}_{*}-z)^{-1}$ ($*\in\{+,-,e,\delta\}$) satisfy $$|R_{*}(z)(\bm{x},\bm{y})|\leq \frac{C}{|\text{Im }z|}e^{-D|\text{Im }z|\|\bm{x}-\bm{y}\|_{1}}.$$
    \item[(b)] The spectral projections $P_*=\rho (H_*)$ ($*\in\{+,-\}$) satisfy $$|P_*(\bm{x},\bm{y})|\leq C e^{-D\|\bm{x}-\bm{y}\|_{1}}.$$
\end{itemize}
\end{lemma}

The third Lemma indicates that the commutator of a tight-binding operator with the switch function $\Lambda_i$ is exponentially localized in the $x_i$-direction:

\begin{lemma} \label{lem_switch_commutator}
Suppose the operator $A$ is tight-binding on the discrete set $\mathcal{D}\subset \mathbb{R}^2$ with the estimates $|A(\bm{x},\bm{y})|\leq Ce^{-D\|\bm{x}-\bm{y}\|_{1}}$. For any switch function $\Lambda_i$ in the $x_i$-direction ($i=1,2$), there exist $C_1,D_1>0$, which depend only on $C,D,\Lambda_i$, such that
\begin{equation*}
\big|\big[A,\Lambda_i\big](\bm{x},\bm{y})\big| \leq C_1e^{-D_1 (\|\bm{x}-\bm{y}\|_{1}+|x_i|+|y_i|)} \quad (\forall \bm{x},\bm{y}\in \mathcal{D}).
\end{equation*}
\end{lemma}

The next Lemma justifies the following intuition: for a product of finitely many tight-binding operators, as long as there is an operator localized in a certain direction, the whole product is also localized in that direction.

\begin{lemma} \label{lem_product_localized}
Let $\mathcal{D}\subset \mathbb{R}^2$ be a discrete set and $U,V\subset \mathbb{R}^2$. Assume $A_j$ ($1\leq i=j\leq n$) are tight-binding operators on $\mathcal{D}$ with the estimates:
\begin{equation*}
\begin{aligned}
&\forall j\in [1,n],\quad |A_j(\bm{x},\bm{y})|\leq C_{j}e^{-D\|\bm{x}-\bm{y}\|_{1}}, \\
&\exists p\in [1,n],\quad |A_p(\bm{x},\bm{y})|\leq C_{p}e^{-D(\|\bm{x}-\bm{y}\|_{1}+\text{dist}(\bm{x},U)+\text{dist}(\bm{y},U))}, \\
&\exists q\in [1,n],\quad |A_q(\bm{x},\bm{y})|\leq C_{q}e^{-D(\|\bm{x}-\bm{y}\|_{1}+\text{dist}(\bm{x},V)+\text{dist}(\bm{y},V))}.
\end{aligned}
\end{equation*}
Then we have
\begin{equation*}
\Big|\big(\Pi_{j=1}^{n}A_j \big)(\bm{x},\bm{y}) \Big|
\leq \frac{256^{n-1}(\Pi_{j=1}^{n}C_j)}{D^{2(n-1)}}e^{-\frac{D}{4}(\|\bm{x}-\bm{y}\|_{1}+\text{dist}(\bm{x},U)+\text{dist}(\bm{y},U)+\text{dist}(\bm{x},V)+\text{dist}(\bm{y},V))}.
\end{equation*}
\end{lemma}

The next Lemma presents a convenient estimate of the trace norm of tight-binding operators based on the integral kernels:

\begin{lemma} \label{lem_trace_norm_kernel}
Let $A$ be a tight-binding operator on $\mathcal{D}$. Its trace norm is estimated as $\|A\|_{\mathscr{T}_{\mathcal{D}}}\leq C\sum_{\bm{x},\bm{y}\in \mathcal{D}}|A(\bm{x},\bm{y})|$, where $C$ depends only on $\mathcal{D}$.
\end{lemma}

The next statement, which follows from Lemma \ref{lem_product_localized} and \ref{lem_trace_norm_kernel}, justifies the intuition that any tight-binding operator localized in both $x_1$ and $x_2$ directions is trace-class.
\begin{corollary} \label{corol_trace_criterion}
Suppose $A_i$ ($1\leq i\leq n$) are tight-binding operators that satisfy the assumption of Lemma \ref{lem_product_localized} with $U=\{x_1=0\}$ and $V=\{x_2=0\}$. Then $\Pi_{j=1}^{n}$ is trace-class.
\end{corollary}

An important property of the trace-class operators is the cyclicity: $Tr_{\mathcal{D}}(AB)=Tr_{\mathcal{D}}(BA)$ if $AB,BA\in\mathscr{T}_{\mathcal{D}}$ (cf. \cite[Corollary 3.8]{simon2005trace}). For periodic operators, even though they are not trace-class, the cyclicity holds when the trace is calculated within a unit cell. This property is referred to as the conditional cyclicity:
\begin{proposition} \label{prop_cond_cyc_periodic}
Let $A,B$ be operators on the lattice $\mathcal{D}_{+}$ (or $\mathcal{D}_{-}$) that are periodic in the the $x_1$-direction, i.e. for any $n\in\mathbb{Z}$
\begin{equation*}
A(\bm{x}+na_{1}^{+}\bm{e}_1,\bm{y}+na_{1}^{+}\bm{e}_1)=A(\bm{x},\bm{y})
,\quad B(\bm{x}+na_{1}^{+}\bm{e}_1,\bm{y}+na_{1}^{+}\bm{e}_1)=B(\bm{x},\bm{y})
\end{equation*}
If both $\mathbbm{1}_{\Omega_{n,a_{1}^{+}}\cap \mathcal{D}_{+}}AB,\mathbbm{1}_{\Omega_{n,a_{1}^{+}}\cap \mathcal{D}_{+}}BA\in \mathscr{T}_{\mathcal{D}}$ for all $n\in\mathbb{N}$, then 
\begin{equation*}
Tr_{\mathcal{D}_{+}}(\mathbbm{1}_{\Omega_{n,a_{1}^{+}}\cap \mathcal{D}_{+}}AB)=Tr_{\mathcal{D}_{+}}(\mathbbm{1}_{\Omega_{n,a_{1}^{+}}\cap \mathcal{D}_{+}}BA).
\end{equation*}
\end{proposition}

The last statement connects the principal-value trace (Definition \ref{def_pv_trace}) and the conventional trace: the principal-value trace of a trace-class operator $A$ coincides with its conventional trace, as is expected.

\begin{proposition} \label{prop_pv_usual_trace}
Let $\mathcal{D}$ be a discrete set in $\mathbb{R}^2$. For $A\in \mathscr{T}_{\mathcal{D}}$ and any $a>0$, it holds that $Tr_{a,\mathcal{D}}^{pv,1}(A)=Tr_{\mathcal{D}}(A)$.
\end{proposition}

\subsection{Almost-analytic extension, Helffer-Sjöstrand formula}
Let $g\in C_c^{\infty}(\mathbb{R})$. An almost-analytic extension of $g$ is a function $\tilde{g}\in C_{c}^{\infty}(\mathbb{C})$ such that
\begin{equation} \label{eq_almost_analytic}
\tilde{g}|_{\mathbb{R}}=g,\quad \partial_{\overline{z}}\tilde{g}=\mathcal{O}(|\text{Im }z|^{\infty}),\quad
\text{supp }(\tilde{g})\subset \mathbb{C}\cap\{z:\, |\text{Im }z|\leq 1\} .
\end{equation}
The second condition above means that for any $N>0$ there exists $C_N>0$ such that $|\partial_{\overline{z}}\tilde{g}|\leq C_N|\text{Im }z|^{N}$. We recall the Hellfer-Sjöstrand formula as a tool in functional calculus (cf. \cite[Theorem 14.8]{zworski2012semiclassical}): if $A$ is a self-adjoint operator, we can express $g(A)$ as an absolutely convergent integral
\begin{equation} \label{eq_hellfer_sjostrand}
g(A)=\frac{1}{\pi i}\int_{\mathbb{C}}\frac{\partial\tilde{g}(z)}{\partial \overline{z}}(A-z)^{-1}dm(z).
\end{equation}
Here $m(z)$ is the Lebesgue measure on the complex plane. The following property of the Hellfer-Sjöstrand representation will be used extensively throughout this paper. Similar property has previously appeared, in various forms, in \cite{elgart2005shortrange+functional,drouot2024bec_curvedinterfaces}.
\begin{proposition} \label{prop_off_diag}
For any bounded operator $A$ on $H_{\mathcal{D}_{\pm}}$, we define
\begin{equation} \label{eq_off_diag_1}
A^{OD}_{\pm}:=\int_{\mathbb{C}}\frac{\partial\tilde{\rho}(z)}{\partial \overline{z}}\Big[R_{\pm}(z)AR_{\pm}(z)\Big]dm(z)
\end{equation}
where $R_{\pm}(z)=(\mathcal{H}_{\pm}-z)^{-1}$ and $\rho$ is a density function satisfying \eqref{eq_density_function}. Then $A^{OD}_{\pm}$ is off-diagonal in the orthogonal decomposition induced by the spectral projection $P_{\pm}=\rho(\mathcal{H}_{\pm})$ in the sense that
\begin{equation} \label{eq_off_diag_2}
P_{\pm}A^{OD}_{\pm}P_{\pm}=P_{\pm}^{\perp}A^{OD}_{\pm}P_{\pm}^{\perp}=0,
\end{equation}
where $P_{\pm}^{\perp}=1-P_{\pm}$.
\end{proposition}
\begin{proof}
First, since the reduced resolvent $P_{\pm}^{\perp}R_{\pm}(z)=R_{\pm}(z)P_{\pm}^{\perp}=(P_{\pm}^{\perp}\mathcal{H}_{\pm}P_{\pm}^{\perp}-z)^{-1}$ is analytic on $\text{supp}(\tilde{\rho})$ (note that $\text{Spec}(P_{\pm}^{\perp}\mathcal{H}_{\pm}P_{\pm}^{\perp})\subset \{\rho =0\}$), we can apply Stokes' theorem to obtain
\begin{equation} \label{eq_prelim_1}
P_{\pm}^{\perp}A^{OD}_{\pm}P_{\pm}^{\perp}=\oint_{\partial (\text{supp}(\tilde{\rho}))}\tilde{\rho}(z) \Big[P_{\pm}^{\perp}R_{\pm}(z)P_{\pm}^{\perp}AP_{\pm}^{\perp}R_{\pm}(z)P_{\pm}^{\perp}\Big]dz=0.
\end{equation}
Next, we prove $P_{\pm}A^{OD}_{\pm}P_{\pm}=0$, which requires more delicate analysis. Let $I_{\pm}=(\inf \text{Spec}(\mathcal{H}_{\pm}),\inf \Delta)$ be the part of spectrum of $\mathcal{H}_{\pm}$ that is below the gap $\Delta$ (i.e. the valence spectrum). Hence $\rho(x)\equiv 1$ for $x\in I_{\pm}$ by \eqref{eq_density_function}. Note that \eqref{eq_off_diag_1} converges absolutely by the almost-analyticity \eqref{eq_almost_analytic}. We see
\begin{equation*}
\begin{aligned}
P_{\pm}A^{OD}_{\pm}P_{\pm}
&=\lim_{\varepsilon\to 0}\int_{\mathbb{C}\backslash (I_{\pm}\times [-\varepsilon,\varepsilon])}\frac{\partial\tilde{\rho}(z)}{\partial \overline{z}}\Big[P_{\pm}R_{\pm}(z)AR_{\pm}(z)P_{\pm}\Big]dm(z) \\
&=\lim_{\varepsilon\to 0}\int_{\mathbb{C}\backslash (I_{\pm}\times [-\varepsilon,\varepsilon])}\frac{\partial\tilde{\rho}(z)}{\partial \overline{z}}\Big[P_{\pm}R_{\pm}(z)P_{\pm}AP_{\pm}R_{\pm}(z)P_{\pm}\Big]dm(z) .
\end{aligned}
\end{equation*}
For any $\varepsilon>0$, the operator $P_{\pm}R_{\pm}(z)P_{\pm}=(P_{\pm}\mathcal{H}_{\pm}P_{\pm}-z)^{-1}$ is analytic on $\mathbb{C}\backslash (I_{\pm}\times [-\varepsilon,\varepsilon])$. Hence by the Stokes' theorem
\begin{equation} \label{eq_prelim_2}
\begin{aligned}
P_{\pm}A^{OD}_{\pm}P_{\pm}
&=\lim_{\varepsilon\to 0}\int_{\mathbb{C}\backslash (I_{\pm}\times [-\varepsilon,\varepsilon])}\frac{\partial}{\partial \overline{z}}\Big[\tilde{\rho}(z)P_{\pm}R_{\pm}(z)P_{\pm}AP_{\pm}R_{\pm}(z)P_{\pm}\Big]dm(z) \\
&=\frac{1}{2}\lim_{\varepsilon\to 0}\oint_{\partial (I_{\pm}\times [-\varepsilon,\varepsilon])}\Big[\tilde{\rho}(z)P_{\pm}R_{\pm}(z)P_{\pm}AP_{\pm}R_{\pm}(z)P_{\pm}\Big] dz .
\end{aligned}
\end{equation}
We argue that we can replace $\tilde{\rho}(z)$ by $1$ in \eqref{eq_prelim_2} without changing the value of the left side. In fact, since $\rho(x)\equiv 1$ on $I_{\pm}$ and $\rho(x+iy)\in C_{c}^{\infty}(\mathbb{R}^2)$, it holds for any $N>0$ that 
$$
\sup_{x+iy\in I_{\pm}\times [-\varepsilon,\varepsilon]}|\partial_{x} \rho(x+iy)|=\mathcal{O}(\varepsilon^{N})
$$
by Taylor expansion. On the other hand, the almost-analyticity \eqref{eq_almost_analytic} indicates $$
\sup_{x+iy\in I_{\pm}\times [-\varepsilon,\varepsilon]}|(\partial_{x}+i\partial_{y}) \rho(x+iy)|=2\sup_{x+iy\in I_{\pm}\times [-\varepsilon,\varepsilon]}|\partial_{\overline{z}} \rho(x+iy)|=\mathcal{O}(\varepsilon^{N}).
$$ 
In conclusion, by the intermediate value theorem
\begin{equation*}
\big|\tilde{\rho}(z)-1 \big|
\leq |\text{Im }z|\sup_{z=x+iy\in I_{\pm}\times [-\varepsilon,\varepsilon]}\big(|\partial_{x} \rho(x+iy)|+|\partial_{y} \rho(x+iy)|\big)
=\mathcal{O} (\varepsilon^{N+1})
\end{equation*}
for any $z\in I_{\pm}\times [-\varepsilon,\varepsilon]$. By setting $N>2$ and using the standard bound $\|R_{\pm}(z)\|\leq 1/|\text{Im }z|^{-1}$, we have the following estimates
\begin{equation} \label{eq_prelim_3}
\begin{aligned}
\Big\|\oint_{\partial (I_{\pm}\times [-\varepsilon,\varepsilon])}\Big[\big(\tilde{\rho}(z)-1\big)P_{\pm}R_{\pm}(z)P_{\pm}AP_{\pm}R_{\pm}(z)P_{\pm}\Big] dz \Big\|
&\leq C\oint_{\partial (I_{\pm}\times [-\varepsilon,\varepsilon])}
|\text{Im }z|^{N+1}|\text{Im }z|^{-2}dz \\
&\leq C\varepsilon^{N-1},
\end{aligned}
\end{equation}
where $C$ depends only on $\tilde{\rho}$ and $\|A\|$. This justifies that we can replace $\tilde{\rho}$ by $1$ in \eqref{eq_prelim_2}. 

Next, we evaluate \eqref{eq_prelim_2} with $\tilde{\rho}$ being replaced by $1$. We claim
\begin{equation} \label{eq_prelim_4}
\oint_{\partial (I_{\pm}\times [-\varepsilon,\varepsilon])}\Big(P_{\pm}R_{\pm}(z)P_{\pm}AP_{\pm}R_{\pm}(z)P_{\pm}\Big) dz
=0\quad (\forall \varepsilon>0).
\end{equation}
To see this, we fix a point $x_0\in I_{\pm}$ and select $r>0$ large enough such that the circle $C_{r}(x_0)$ includes $I_{\pm}\times [-\varepsilon,\varepsilon]$. Since the integrand of $\eqref{eq_prelim_4}$ is analytic in the punctured disk $D_{r}(x_0)\backslash \overline{I_{\pm}\times [-\varepsilon,\varepsilon]}$, the Cauchy integral theorem yields
\begin{equation*}
\oint_{\partial (I_{\pm}\times [-\varepsilon,\varepsilon])}\Big(P_{\pm}R_{\pm}(z)P_{\pm}AP_{\pm}R_{\pm}(z)P_{\pm}\Big) dz
=\oint_{C_{r}(x_0)}\Big(P_{\pm}R_{\pm}(z)P_{\pm}AP_{\pm}R_{\pm}(z)P_{\pm}\Big) dz.
\end{equation*}
Since
\begin{equation*}
\big\|P_{\pm}R_{\pm}(z)P_{\pm}AP_{\pm}R_{\pm}(z)P_{\pm} \big\|
\leq \|A\|\|R_{\pm}(z)\|^2
\leq \frac{C}{\text{dist}^2(z,I_{\pm})}
\leq \frac{C}{r^2}
\end{equation*}
for $z\in C_{r}(x_0)$. 
\begin{equation*}
\Big|
\oint_{C_{r}(x_0)}\Big(P_{\pm}R_{\pm}(z)P_{\pm}AP_{\pm}R_{\pm}(z)P_{\pm}\Big) dz
\Big|
\leq \frac{C}{r^2}\oint_{C_{r}(x_0)} dz
\leq \frac{C}{r} \to 0 \quad (\text{as }r\to\infty).
\end{equation*}
This proves \eqref{eq_prelim_4}. Then the proof of $P_{\pm}A^{OD}_{\pm}P_{\pm}=0$ follows from \eqref{eq_prelim_2}-\eqref{eq_prelim_4}.
\end{proof}
We remark that the operator in the form \eqref{eq_off_diag_1} is closely linked to the inverse Liouvillian operator as studied in \cite{marcelli2021new,Giovanna22charge_to_spin}, which serves as a key ingredient for the theory of NEASS and the definition of the spin conductivity in their paper.

\section{Vanishing spin-torque response in the bulk}

In this section, we prove that the spin-torque response, expressed as a correlation function $\Sigma_{\pm}^{\Lambda_2, S}$ (defined in \eqref{eq_torque_correlation}), vanishes mesoscopically in the bulk media.  Our proof relies on the following result, which justifies the statement expressed in a Kubo-like formulation.
\begin{proposition} \label{prop_vanish_bulk_torque_Kubo}
For any switch function $\Lambda_2$ and $n\in\mathbb{Z}$,
\begin{equation*}
Tr_{\mathcal{D}_{\pm}}\Big(\mathbbm{1}_{\Omega_{n,a_1^{\pm}}\cap \mathcal{D}_{\pm}}P_{\pm}\big[[P_{\pm},S],[P_{\pm},\Lambda_2] \big] \Big)=0.
\end{equation*}
\end{proposition}
We refer the reader to \cite[Theorem 2.8 and equations (5.18)-(5.19)]{Marcelli2019spin_conductivity} for the proof. The key to the proof is that the spin operator $S$ commutes with position operators, and consequently, $[S,\Lambda_i]=0$. This reflects the physical intuition that the spin torque response should vanish in a suitable sense as the external potential (which is proportional to $\Lambda_2$) does not couple to the spin DOF \cite{Giovanna22charge_to_spin}.

We express Proposition \ref{prop_vanish_bulk_torque_Kubo} in our language of correlation function:
\begin{proposition} \label{prop_vanish_bulk_torque_correlation}
Let $\Sigma_{\pm}^{\Lambda_2, S}$ be the potential-toque correlation (associated with the bulk Hamiltonians) defined in \eqref{eq_torque_correlation}. Then for any $n\in\mathbb{Z}$, $k\in\mathbb{N}$,
\begin{equation} \label{eq_vanish_bulk_torque_correlation}
Tr_{\mathcal{D}_{\pm}}(\mathbbm{1}_{\Omega_{n,ka_1^{\pm}}\cap \mathcal{D}_{\pm}}\Sigma_{\pm}^{\Lambda_2, S})=0.
\end{equation}
\end{proposition}
\begin{proof}
We prove \eqref{eq_vanish_bulk_torque_correlation} for the `$+$' case (the right bulk) and $k=1$; the proof of the `$-$' case and other $k\in\mathbb{N}$ is similar. In particular, we show that
\begin{equation} \label{eq_vanish_bulk_torque_correlation_proof_1}
Tr_{\mathcal{D}_{+}}\Big(\mathbbm{1}_{\Omega_{n,a_1^{+}}\cap \mathcal{D}_{+}}\Big[\int_{\mathbb{C}}\frac{\partial \tilde{\rho}}{\partial \overline{z}}
R_{+}(z) \big[\mathcal{H}_{+},\Lambda_2\big] R_{+}(z) \big[\mathcal{H}_{+}, S\big] R_{+}(z) dm(z)\Big] \Big)=0.
\end{equation}
Similarly, one can prove
\begin{equation*}
Tr_{\mathcal{D}_{+}}\Big(\mathbbm{1}_{\Omega_{n,a_1^{+}}\cap \mathcal{D}_{+}}\Big[\int_{\mathbb{C}}\frac{\partial \tilde{\rho}}{\partial \overline{z}}R_{+}(z) \big[\mathcal{H}_{+},S\big] R_{+}(z) \big[\mathcal{H}_{+},\Lambda_2\big] R_{+}(z) dm(z)\Big] \Big)=0.
\end{equation*}
The proof is then completed by recalling the definition \eqref{eq_torque_correlation} of $\Sigma_{+}^{\Lambda_2,S}$.

{\color{blue}Step 1:} We first justify that the operator in \eqref{eq_vanish_bulk_torque_correlation_proof_1} is trace-class, which is intuitively straightforward as $\mathbbm{1}_{\Omega_{n,a_1^{+}}\cap \mathcal{D}_{+}}$ is localized in the $x_1$-direction and $\big[\mathcal{H}_{+},\Lambda_2\big]$ is localized in $x_2$. Specifically, by the identity $R_{+}(z)[\mathcal{H}_{+},S]R_{+}(z)=-[R_{+}(z),S]$, we have
\begin{equation} \label{eq_vanish_bulk_torque_correlation_proof_2}
\begin{aligned}
&\mathbbm{1}_{\Omega_{n,a_1^{+}}\cap \mathcal{D}_{+}}\Big(\int_{\mathbb{C}}\frac{\partial \tilde{\rho}}{\partial \overline{z}}
R_{+}(z) \big[\mathcal{H}_{+},\Lambda_2\big] R_{+}(z) \big[\mathcal{H}_{+}, S\big] R_{+}(z) dm(z)\Big) \\
&=-\int_{\mathbb{C}}\frac{\partial \tilde{\rho}}{\partial \overline{z}}\Big(\mathbbm{1}_{\Omega_{n,a_1^{+}}\cap \mathcal{D}_{+}}
R_{+}(z) \big[\mathcal{H}_{+},\Lambda_2\big] \big[R_{+}(z), S\big] \Big)dm(z) .
\end{aligned}
\end{equation}
Note that the following estimates hold for all $\bm{x},\bm{y}\in \mathcal{D}_{+}$ by Lemma \ref{lem_resolv_ground_tight}, \ref{lem_switch_commutator} and the definition of the spin operator $S=\text{Id}\otimes \frac{1}{2}\sigma_z$:
\begin{equation} \label{eq_vanish_bulk_torque_correlation_proof_3}
\left\{
\begin{aligned}
&|R_{+}(z)(\bm{x},\bm{y})|\leq \frac{C_1}{|\text{Im }z|}e^{-D_1|\text{Im }z|\|\bm{x}-\bm{y}\|_{1}}, \\
&|S(\bm{x},\bm{y})|
\leq C_2e^{-D_2\|\bm{x}-\bm{y}\|_{1}}, \\
&|[\mathcal{H}_{+},\Lambda_2](\bm{x},\bm{y})| \leq C_3e^{-D_3(\|\bm{x}-\bm{y}\|_{1}+|x_2|+|y_2|)}, \\
&|\mathbbm{1}_{\Omega_{n,a_1^{+}}\cap \mathcal{D}_{+}}(\bm{x},\bm{y})|
=|\mathbbm{1}_{\Omega_{n,a_1^{+}}\cap \mathcal{D}_{+}}(\bm{x})|\delta_{\bm{x},\bm{y}}\leq C_{4}e^{-D_4(\|\bm{x}-\bm{y}\|_{1}+|x_1|+|y_1|)},
\end{aligned}
\right.
\end{equation}
where $C_i,D_i$ depend only on $n$, $\Lambda_2$ and the decay length $\lambda$ of the Hamiltonians in Assumption \ref{assum_bulk_hamiltonian} (i). 

For each inequality in the sequel, we will omit the dependence of constants on the decay length $\lambda$ as it remains fixed throughout this paper. Letting $D=\min_{1\leq i\leq 4} D_i$, we see that \eqref{eq_vanish_bulk_torque_correlation_proof_3} holds by replacing $D_1|\text{Im }z|,D_i$ ($2\leq i\leq 4$) with $D|\text{Im }z|$ when $|\text{Im }z|\leq 1$. Hence we have the following estimate for the integrand in \eqref{eq_vanish_bulk_torque_correlation_proof_2} by Lemma \ref{lem_product_localized}:
\begin{equation} \label{eq_vanish_bulk_torque_correlation_proof_4}
\begin{aligned}
\Big|\Big(\mathbbm{1}_{\Omega_{n,a_1^{+}}\cap \mathcal{D}_{+}}
R_{+}(z) \big[\mathcal{H}_{+},\Lambda_2\big] \big[R_{+}(z), S\big] \Big)(\bm{x},\bm{y}) \Big| \leq \frac{C}{|\text{Im }z|^{10}}e^{-\frac{D|\text{Im }z|}{4}(\|\bm{x}-\bm{y}\|_{1}+|x_1|+|y_1|+|x_2|+|y_2|)},
\end{aligned}
\end{equation}
where $C,D$ depend only on $n$. Consequently, by Lemma \ref{lem_trace_norm_kernel},
\begin{equation} \label{eq_vanish_bulk_torque_correlation_proof_5}
\Big\|\mathbbm{1}_{\Omega_{n,a_1^{+}}\cap \mathcal{D}_{+}}
R_{+}(z) \big[\mathcal{H}_{+},\Lambda_2\big] \big[R_{+}(z), S\big]\Big\|_{\mathscr{T}_{\mathcal{D}_{+}}}\leq \frac{C}{|\text{Im }z|^{12}}.
\end{equation}
The almost-analyticity of $\tilde{\rho}$ (see \eqref{eq_almost_analytic}) further implies that
\begin{equation*}
\begin{aligned}
&\Big\|\int_{\mathbb{C}}\frac{\partial \tilde{\rho}}{\partial \overline{z}}\Big(\mathbbm{1}_{\Omega_{n,a_1^{+}}\cap \mathcal{D}_{+}}
R_{+}(z) \big[\mathcal{H}_{+},\Lambda_2\big] \big[R_{+}(z), S\big] \Big)dm(z)\Big\|_{\mathscr{T}_{\mathcal{D}_{+}}} \\
&\leq \int_{\mathbb{C}}\Big|\frac{\partial \tilde{\rho}}{\partial \overline{z}}\Big|\Big\|\mathbbm{1}_{\Omega_{n,a_1^{+}}\cap \mathcal{D}_{+}}
R_{+}(z) \big[\mathcal{H}_{+},\Lambda_2\big] \big[R_{+}(z), S\big] \Big\|_{\mathscr{T}_{\mathcal{D}_{+}}}dm(z)<\infty .
\end{aligned}
\end{equation*}
This concludes the proof that the operator \eqref{eq_vanish_bulk_torque_correlation_proof_2} is trace-class.

{\color{blue}Step 2:} We prove that
\begin{equation} \label{eq_vanish_bulk_torque_correlation_proof_6}
\begin{aligned}
&Tr_{\mathcal{D}_{+}}\Big(\mathbbm{1}_{\Omega_{n,a_1^{+}}\cap \mathcal{D}_{+}}\Big[\int_{\mathbb{C}}\frac{\partial \tilde{\rho}}{\partial \overline{z}}
R_{+}(z) \big[\mathcal{H}_{+},\Lambda_2\big] R_{+}(z) \big[\mathcal{H}_{+}, S\big] R_{+}(z) dm(z)\Big] \Big) \\
&=-Tr_{\mathcal{D}_{+}}\Big(\mathbbm{1}_{\Omega_{n,a_1^{+}}\cap \mathcal{D}_{+}}\Big[\int_{\mathbb{C}}\frac{\partial \tilde{\rho}}{\partial \overline{z}}
R_{+}(z) \big[\mathcal{H}_{+},\Lambda_2\big]  \big[R_{+}(z), S\big]  dm(z)\Big] \Big) \\
&=i\pi Tr_{\mathcal{D}_{+}}\Big(\mathbbm{1}_{\Omega_{n,a_1^{+}}\cap \mathcal{D}_{+}}\Big(P_{+}[P_{+},\Lambda_2]SP_{+}+P_{+}^{\perp}[P_{+},\Lambda_2]SP_{+}^{\perp}\Big) \Big) .
\end{aligned}
\end{equation}
Analogous to Step 1, one can show that \eqref{eq_vanish_bulk_torque_correlation_proof_5} remains to hold by inserting finitely many projections $P_{+}$ or $P_{+}^{\perp}$ inside the product. In particular, based on this observation, the following operators are trace-class
\begin{equation*}
\mathbbm{1}_{\Omega_{n,a_1^{+}}\cap \mathcal{D}_{+}}PA,\quad
\mathbbm{1}_{\Omega_{n,a_1^{+}}\cap \mathcal{D}_{+}}AP
\end{equation*}
for $A=\int_{\mathbb{C}}\frac{\partial \tilde{\rho}}{\partial \overline{z}}
R_{+}(z) \big[\mathcal{H}_{+},\Lambda_2\big]  \big[R_{+}(z), S\big]  dm(z)$ and $P\in \{P_{+},P_{+}^{\perp}\}$. Since both $A$, $P_{+}$ and $P_{+}^{\perp}$ are periodic in $\mathcal{D}_{+}$, we can apply the conditional cyclicity (Proposition \ref{prop_cond_cyc_periodic}) to obtain
\begin{equation}  \label{eq_vanish_bulk_torque_correlation_proof_7}
\begin{aligned}
&Tr_{\mathcal{D}_{+}}\Big(\mathbbm{1}_{\Omega_{n,a_1^{+}}\cap \mathcal{D}_{+}}\Big[\int_{\mathbb{C}}\frac{\partial \tilde{\rho}}{\partial \overline{z}}
R_{+}(z) \big[\mathcal{H}_{+},\Lambda_2\big]  \big[R_{+}(z), S\big]  dm(z)\Big] \Big) \\
&=Tr_{\mathcal{D}_{+}}\Big(\mathbbm{1}_{\Omega_{n,a_1^{+}}\cap \mathcal{D}_{+}} \big(P_{+}^2+(P_{+}^{\perp})^2 \big) \Big[\int_{\mathbb{C}}\frac{\partial \tilde{\rho}}{\partial \overline{z}}
R_{+}(z) \big[\mathcal{H}_{+},\Lambda_2\big]  \big(R_{+}(z)S-SR_{+}(z)\big)  dm(z)\Big] \Big) \\
&=Tr_{\mathcal{D}_{+}}\Big(\mathbbm{1}_{\Omega_{n,a_1^{+}}\cap \mathcal{D}_{+}}P_{+}\Big[\int_{\mathbb{C}}\frac{\partial \tilde{\rho}}{\partial \overline{z}}
R_{+}(z) \big[\mathcal{H}_{+},\Lambda_2\big]  \big(R_{+}(z)S-SR_{+}(z)\big)  dm(z)\Big]P_{+} \Big) \\
&\quad + Tr_{\mathcal{D}_{+}}\Big(\mathbbm{1}_{\Omega_{n,a_1^{+}}\cap \mathcal{D}_{+}}P_{+}^{\perp}\Big[\int_{\mathbb{C}}\frac{\partial \tilde{\rho}}{\partial \overline{z}}
R_{+}(z) \big[\mathcal{H}_{+},\Lambda_2\big]  \big(R_{+}(z)S-SR_{+}(z)\big)  dm(z)\Big]P_{+}^{\perp} \Big) .
\end{aligned}
\end{equation}
In \eqref{eq_vanish_bulk_torque_correlation_proof_7}, the terms involving $SR_{+}(z)$ vanish by Proposition \ref{prop_off_diag}. Hence
\begin{equation*}
\begin{aligned}
&Tr_{\mathcal{D}_{+}}\Big(\mathbbm{1}_{\Omega_{n,a_1^{+}}\cap \mathcal{D}_{+}}\Big[\int_{\mathbb{C}}\frac{\partial \tilde{\rho}}{\partial \overline{z}}
R_{+}(z) \big[\mathcal{H}_{+},\Lambda_2\big]  \big[R_{+}(z), S\big]  dm(z)\Big] \Big) \\
&=Tr_{\mathcal{D}_{+}}\Big(\mathbbm{1}_{\Omega_{n,a_1^{+}}\cap \mathcal{D}_{+}}P_{+}\Big[\int_{\mathbb{C}}\frac{\partial \tilde{\rho}}{\partial \overline{z}}
R_{+}(z) \big[\mathcal{H}_{+},\Lambda_2\big]  R_{+}(z)S  dm(z)\Big]P_{+} \Big) \\
&\quad + Tr_{\mathcal{D}_{+}}\Big(\mathbbm{1}_{\Omega_{n,a_1^{+}}\cap \mathcal{D}_{+}}P_{+}^{\perp}\Big[\int_{\mathbb{C}}\frac{\partial \tilde{\rho}}{\partial \overline{z}}
R_{+}(z) \big[\mathcal{H}_{+},\Lambda_2\big]  R_{+}(z)S  dm(z)\Big]P_{+}^{\perp} \Big) \\
&=-Tr_{\mathcal{D}_{+}}\Big(\mathbbm{1}_{\Omega_{n,a_1^{+}}\cap \mathcal{D}_{+}}P_{+}\Big[\int_{\mathbb{C}}\frac{\partial \tilde{\rho}}{\partial \overline{z}}
 \big[R_{+}(z),\Lambda_2\big]  S  dm(z)\Big]P_{+} \Big) \\
&\quad - Tr_{\mathcal{D}_{+}}\Big(\mathbbm{1}_{\Omega_{n,a_1^{+}}\cap \mathcal{D}_{+}}P_{+}^{\perp}\Big[\int_{\mathbb{C}}\frac{\partial \tilde{\rho}}{\partial \overline{z}}
 \big[R_{+}(z),\Lambda_2\big]  S  dm(z)\Big]P_{+}^{\perp} \Big) \\
&=-i\pi Tr_{\mathcal{D}_{+}}\Big(\mathbbm{1}_{\Omega_{n,a_1^{+}}\cap \mathcal{D}_{+}}P_{+}
 \big[\rho(\mathcal{H}_{+}),\Lambda_2\big]  S  P_{+} \Big)
 -i\pi Tr_{\mathcal{D}_{+}}\Big(\mathbbm{1}_{\Omega_{n,a_1^{+}}\cap \mathcal{D}_{+}}P_{+}^{\perp}
 \big[\rho(\mathcal{H}_{+}),\Lambda_2\big]  S  P_{+}^{\perp} \Big)
\end{aligned}
\end{equation*}
where the Hellfer-Sjöstrand formula is applied to obtain the last equality. Then the proof of \eqref{eq_vanish_bulk_torque_correlation_proof_6} is completed by recalling $P_{+}=\rho(\mathcal{H}_{+})$.

{\color{blue}Step 3:} Finally, we prove
\begin{equation} \label{eq_vanish_bulk_torque_correlation_proof_8}
\begin{aligned}
&Tr_{\mathcal{D}_{+}}\Big(\mathbbm{1}_{\Omega_{n,a_1^{+}}\cap \mathcal{D}_{+}}\Big(P_{+}[P_{+},\Lambda_2]SP_{+}+P_{+}^{\perp}[P_{+},\Lambda_2]SP_{+}^{\perp}\Big) \Big) \\
&=Tr_{\mathcal{D}_{+}}\Big(\mathbbm{1}_{\Omega_{n,a_1^{+}}\cap \mathcal{D}_{+}}P_{+}\big[[P_{+},S],[P_{+},\Lambda_2] \big] \Big).
\end{aligned}
\end{equation}
Then the proof of \eqref{eq_vanish_bulk_torque_correlation_proof_1} is completed by combining \eqref{eq_vanish_bulk_torque_correlation_proof_6}, \eqref{eq_vanish_bulk_torque_correlation_proof_8} and Proposition \ref{prop_vanish_bulk_torque_Kubo}.

The equality \eqref{eq_vanish_bulk_torque_correlation_proof_8} follows from standard algebraic manipulation of double commutators and the conditional cyclicity. In fact, using the orthogonality of $P_{+}$ and $P_{+}^{\perp}$, and $\mathbbm{1}=P_{+}+P_{+}^{\perp}$,
\begin{equation*}
\begin{aligned}
\text{Left side of \eqref{eq_vanish_bulk_torque_correlation_proof_8}} 
&=Tr_{\mathcal{D}_{+}}\Big(\mathbbm{1}_{\Omega_{n,a_1^{+}}\cap \mathcal{D}_{+}}\Big(-P_{+}[P_{+}^{\perp},\Lambda_2]SP_{+}+P_{+}^{\perp}[P_{+},\Lambda_2]SP_{+}^{\perp}\Big) \Big) \\
&=Tr_{\mathcal{D}_{+}}\Big(\mathbbm{1}_{\Omega_{n,a_1^{+}}\cap \mathcal{D}_{+}}\Big(P_{+}\Lambda_2 P_{+}^{\perp}SP_{+}-P_{+}^{\perp}\Lambda_2 P_{+}SP_{+}^{\perp}\Big) \Big) \\
&=Tr_{\mathcal{D}_{+}}\Big(\mathbbm{1}_{\Omega_{n,a_1^{+}}\cap \mathcal{D}_{+}}\Big(P_{+}\Lambda_2 P_{+}^{\perp}\cdot P_{+}^{\perp}SP_{+}-P_{+}^{\perp}\Lambda_2 P_{+}\cdot P_{+}SP_{+}^{\perp}\Big) \Big) \\
&=Tr_{\mathcal{D}_{+}}\Big(\mathbbm{1}_{\Omega_{n,a_1^{+}}\cap \mathcal{D}_{+}}\Big(-P_{+}\big[ P_{+}^{\perp},\Lambda_2\big] \cdot \big[P_{+}^{\perp},S\big]P_{+}+\big[ P_{+}^{\perp},\Lambda_2\big] P_{+}\cdot P_{+}\big[P_{+}^{\perp},S\big]\Big) \Big) .
\end{aligned}
\end{equation*}
By the conditional cyclicity in Proposition \ref{prop_cond_cyc_periodic}, we have
\begin{equation*}
Tr_{\mathcal{D}_{+}}\Big(\mathbbm{1}_{\Omega_{n,a_1^{+}}\cap \mathcal{D}_{+}}\big[ P_{+}^{\perp},\Lambda_2\big] P_{+}\cdot P_{+}\big[P_{+}^{\perp},S\big] \Big)
=Tr_{\mathcal{D}_{+}}\Big(\mathbbm{1}_{\Omega_{n,a_1^{+}}\cap \mathcal{D}_{+}}P_{+}\big[P_{+}^{\perp},S\big]\cdot \big[ P_{+}^{\perp},\Lambda_2\big] P_{+}  \Big) .
\end{equation*}
The fact that these operators are trace-class can be verified similarly as we did in Step 1. Hence
\begin{equation*}
\begin{aligned}
\text{Left side of \eqref{eq_vanish_bulk_torque_correlation_proof_8}}
&=Tr_{\mathcal{D}_{+}}\Big(\mathbbm{1}_{\Omega_{n,a_1^{+}}\cap \mathcal{D}_{+}}\Big(-P_{+}\big[ P_{+}^{\perp},\Lambda_2\big] \cdot \big[P_{+}^{\perp},S\big]P_{+}+P_{+}\big[P_{+}^{\perp},S\big]\cdot \big[ P_{+}^{\perp},\Lambda_2\big] P_{+}\Big) \Big) \\
&=Tr_{\mathcal{D}_{+}}\Big(\mathbbm{1}_{\Omega_{n,a_1^{+}}\cap \mathcal{D}_{+}}P_{+}\Big[\big[P_{+}^{\perp},S\big], \big[ P_{+}^{\perp},\Lambda_2\big]\Big] P_{+} \Big) \\
&=Tr_{\mathcal{D}_{+}}\Big(\mathbbm{1}_{\Omega_{n,a_1^{+}}\cap \mathcal{D}_{+}}P_{+}\Big[\big[P_{+},S\big], \big[ P_{+},\Lambda_2\big]\Big] P_{+} \Big) \\
&=Tr_{\mathcal{D}_{+}}\Big(\mathbbm{1}_{\Omega_{n,a_1^{+}}\cap \mathcal{D}_{+}}P_{+}\Big[\big[P_{+},S\big], \big[ P_{+},\Lambda_2\big]\Big] \Big)
\end{aligned}
\end{equation*}
where the conditional cyclicity is applied to obtain the last equality. This yields \eqref{eq_vanish_bulk_torque_correlation_proof_8}.
\end{proof}

\section{Bulk spin conductance: proof of Proposition \ref{prop_spin_conduct_justified}}
The proof of Proposition \ref{prop_spin_conduct_justified} is sketched in Section 1.2; here we complement the details.

\begin{proof}[Proof of Proposition \ref{prop_spin_conduct_justified}]
We only prove Proposition \ref{prop_spin_conduct_justified} for the `$+$' case, and the proof of the `$-$' case is similar

{\color{blue}Step 0:} Let $\Lambda_1^{\star}(\bm{x})=H(x_1)$ be the Heaviside function of variable $x_1$. As in \eqref{eq_main_idea_proof_bulk_conduct_3}, 
by moving $\Lambda_1^{\star}$ outside the product cyclically, we decompose $\Sigma_{\pm}^{\Lambda_2,\Lambda_1^{\star} S}$ as 
$$
\Sigma_{\pm}^{\Lambda_2,\Lambda_1^{\star} S}= \text{I} + \text{II}
$$
where  
\begin{equation*} \label{eq_spin_conduct_justified_proof_1}
\begin{aligned}
\text{I}=&\int_{\mathbb{C}}\frac{\partial \tilde{\rho}}{\partial \overline{z}}\Big[
R_{+}(z) \big[\mathcal{H}_{+},\Lambda_2\big] R_{+}(z) \big[\mathcal{H}_{+}, S\big] R_{+}(z) \Lambda_1^{\star} \\
&\quad\quad\quad\quad - \Lambda_1^{\star} R_{+}(z) \big[\mathcal{H}_{+}, S\big] R_{+}(z) \big[\mathcal{H}_{+},\Lambda_2\big] R_{+}(z) \Big]dm(z), \\
\text{II}=&\int_{\mathbb{C}}\frac{\partial \tilde{\rho}}{\partial \overline{z}}\Big[\sum_{j=1}^{4}A_{+,j}(z) \Big]dm(z),
\end{aligned}
\end{equation*}
with
\begin{equation*}
\begin{aligned}
&A_{+,1}=R_{+}(z) \big[\mathcal{H}_{+},\Lambda_2\big] R_{+}(z) \big[\mathcal{H}_{+}, S\big]\Big[\Lambda^{\star}_1, R_{+}(z)\Big],\quad A_{+,2}=R_{+}(z) \big[\mathcal{H}_{+},\Lambda_2\big] R_{+}(z)S\big[\mathcal{H}_{+},\Lambda^{\star}_1 \big] R_{+}(z), \\
&A_{+,3}=-R_{+}(z) \big[\mathcal{H}_{+},\Lambda^{\star}_1 \big]S R_{+}(z) \big[\mathcal{H}_{+},\Lambda_2\big]R_{+}(z),\quad
A_{+,4}=\big[\Lambda_1^{\star},R_{+}(z)\big] \big[\mathcal{H}_{+}, S\big] R_{+}(z) \big[\mathcal{H}_{+},\Lambda_2\big] R_{+}(z).
\end{aligned}
\end{equation*}

{\color{blue}Step 1:} We justify that the operator $\text{II}$ is trace-class. Here we prove it only for $\int_{\mathbb{C}}\frac{\partial \tilde{\rho}}{\partial \overline{z}}A_{+,1}(z) dm(z)$, and the proof for the other operators in $\text{II}$ with $j\neq 1$ follows similar lines.

Note the following estimates hold for all $\bm{x},\bm{y}\in \mathcal{D}_{+}$ and $|\text{Im }z|\leq 1$ (analogous to \eqref{eq_vanish_bulk_torque_correlation_proof_3}) by the Lemmas in Section 2.1
\begin{equation} \label{eq_spin_conduct_justified_proof_2}
\left\{
\begin{aligned}
&|R_{+}(z)(\bm{x},\bm{y})|\leq \frac{C_1}{|\text{Im }z|}e^{-D_1|\text{Im }z|\|\bm{x}-\bm{y}\|_{1}}, \\
&|[\mathcal{H}_{+},S](\bm{x},\bm{y})| \leq C_2e^{-D_2\|\bm{x}-\bm{y}\|_{1}}, \\
&|[\mathcal{H}_{+},\Lambda_2](\bm{x},\bm{y})| \leq C_3e^{-D_3(\|\bm{x}-\bm{y}\|_{1}+|x_2|+|y_2|)}, \\
&|[R_{+}(z),\Lambda_1^{\star}](\bm{x},\bm{y})|\leq \frac{C_4}{|\text{Im }z|}e^{-D_4|\text{Im }z|(\|\bm{x}-\bm{y}\|_{1}+|x_1|+|y_1|)},
\end{aligned}
\right.
\end{equation}
where $C_i,D_i$ depend only on $\Lambda_1^{\star},\Lambda_2$. Hence by Lemma \ref{lem_product_localized}
\begin{equation*}
\begin{aligned}
\Big|A_{+,1}(z)(\bm{x},\bm{y}) \Big| \leq \frac{C}{|\text{Im }z|^{11}}e^{-\frac{D|\text{Im }z|}{4}(\|\bm{x}-\bm{y}\|_{1}+|x_1|+|y_1|+|x_2|+|y_2|)},
\end{aligned}
\end{equation*}
where $D=\min_{1\leq i\leq 4} D_i$ and $C$ depend only on $\Lambda_1^{\star},\Lambda_2$. Consequently, we have $\big\|A_{+,1}(z)\big\|_{\mathscr{T}_{\mathcal{D}_{+}}}\leq \frac{C}{|\text{Im }z|^{13}}$ by Lemma \ref{lem_trace_norm_kernel}, and hence
\begin{equation*}
\begin{aligned}
\Big\|\int_{\mathbb{C}}\frac{\partial \tilde{\rho}}{\partial \overline{z}}A_{+,1}(z) dm(z)\Big\|_{\mathscr{T}_{\mathcal{D}_{+}}} 
\leq \int_{\mathbb{C}}\Big|\frac{\partial \tilde{\rho}}{\partial \overline{z}}\Big|\big\|A_{+,1}(z) \big\|_{\mathscr{T}_{\mathcal{D}_{+}}}dm(z)<\infty
\end{aligned}
\end{equation*}
by the almost-analyticity of $\tilde{\rho}$. 
Thus, $\int_{\mathbb{C}}\frac{\partial \tilde{\rho}}{\partial \overline{z}}A_{+,1}(z) dm(z)$ is trace-class.

{\color{blue}Step 2:} We show that 
\begin{equation} \label{eq_spin_conduct_justified_proof_3}
\begin{aligned}
Tr_{ka_1^{+},\mathcal{D}_{+}}^{pv,1}(I)=0 .
\end{aligned}
\end{equation}
This, combined with the results in Step 1, justifies that $Tr_{ka_1^{+},\mathcal{D}_{+}}^{pv,1}(\Sigma_{\pm}^{\Lambda_2,\Lambda_1^{\star} S})$ is well-defined. Again, we only prove \eqref{eq_spin_conduct_justified_proof_3} for $k=1$.

By definition of the principal-value trace, we have
\begin{equation*}
\begin{aligned}
Tr_{a_1^{+},\mathcal{D}_{+}}^{pv,1}(I)
=&\lim_{N\to\infty}\sum_{n=-N}^{N-1} 
Tr_{\mathcal{D}_{+}}\Big( \mathbbm{1}_{\Omega_{n,a_{1}^{+}}\cap \mathcal{D}_{+}}
\Big\{\int_{\mathbb{C}}\frac{\partial \tilde{\rho}}{\partial \overline{z}}\Big[
R_{+}(z) \big[\mathcal{H}_{+},\Lambda_2\big] R_{+}(z) \big[\mathcal{H}_{+}, S\big] R_{+}(z) \Lambda_1^{\star} \\
&\quad\quad\quad\quad\quad\quad\quad\quad\quad - \Lambda_1^{\star} R_{+}(z) \big[\mathcal{H}_{+}, S\big] R_{+}(z) \big[\mathcal{H}_{+},\Lambda_2\big] R_{+}(z) \Big]dm(z)\Big\} \mathbbm{1}_{\Omega_{n,a_{1}^{+}}\cap \mathcal{D}_{+}} \Big) .
\end{aligned}
\end{equation*}
Since $\mathbbm{1}_{\Omega_{n,a_{1}^{+}}\cap \mathcal{D}_{+}}\Lambda_1^{\star}=\Lambda_1^{\star}\mathbbm{1}_{\Omega_{n,a_{1}^{+}}\cap \mathcal{D}_{+}}=\mathbbm{1}_{\Omega_{n,a_{1}^{+}}\cap \mathcal{D}_{+}}\mathbbm{1}_{\{n\geq 0\}}(n)$, 
\begin{equation*}
\begin{aligned}
Tr_{a_1^{+},\mathcal{D}_{+}}^{pv,1}(I)
=&\lim_{N\to\infty}\sum_{n=0}^{N-1} 
Tr_{\mathcal{D}_{+}}\Big( \mathbbm{1}_{\Omega_{n,a_{1}^{+}}\cap \mathcal{D}_{+}}
\Big\{\int_{\mathbb{C}}\frac{\partial \tilde{\rho}}{\partial \overline{z}}\Big[
R_{+}(z) \big[\mathcal{H}_{+},\Lambda_2\big] R_{+}(z) \big[\mathcal{H}_{+}, S\big] R_{+}(z) \\
&\quad\quad\quad\quad\quad\quad\quad\quad\quad -  R_{+}(z) \big[\mathcal{H}_{+}, S\big] R_{+}(z) \big[\mathcal{H}_{+},\Lambda_2\big] R_{+}(z) \Big]dm(z)\Big\} \mathbbm{1}_{\Omega_{n,a_{1}^{+}}\cap \mathcal{D}_{+}} \Big) \\
&=\lim_{N\to\infty}\sum_{n=0}^{N-1} 
Tr_{\mathcal{D}_{+}}\Big( \mathbbm{1}_{\Omega_{n,a_{1}^{+}}\cap \mathcal{D}_{+}}
\Sigma_{+}^{\Lambda_2,S}
\mathbbm{1}_{\Omega_{n,a_{1}^{+}}\cap \mathcal{D}_{+}} \Big)
=\lim_{N\to\infty}\sum_{n=0}^{N-1} 
Tr_{\mathcal{D}_{+}}\Big( \mathbbm{1}_{\Omega_{n,a_{1}^{+}}\cap \mathcal{D}_{+}}
\Sigma_{+}^{\Lambda_2,S} \Big) .
\end{aligned}
\end{equation*}
Hence, \eqref{eq_spin_conduct_justified_proof_3} follows by observing that each term in this summation vanishes by Proposition \ref{prop_vanish_bulk_torque_correlation}.

{\color{blue}Step 3:} We consider a general switch function $\Lambda_1$ and show that \eqref{eq_spin_conduc_1} is independent of $\Lambda_1$. It is sufficient to prove
\begin{equation} \label{eq_spin_conduct_justified_proof_4}
Tr_{\mathcal{D}_{+}}(\Sigma_{+}^{\Lambda_2,\tilde{\Lambda}_1 S})=0
\end{equation}
with
\begin{equation*}
\begin{aligned}
\Sigma_{+}^{\Lambda_2,\tilde{\Lambda}_1 S}
&:=\int_{\mathbb{C}}dm(z)\frac{\partial \tilde{\rho}}{\partial \overline{z}}\Big(
R_{+}(z) \big[\mathcal{H}_{+},\Lambda_2\big] R_{+}(z) \big[\mathcal{H}_{+},\tilde{\Lambda}_1 S\big] R_{+}(z) \\
&\quad\quad\quad\quad\quad\quad\quad - R_{+}(z) \big[\mathcal{H}_{+},\tilde{\Lambda}_1 S\big] R_{+}(z) \big[\mathcal{H}_{+},\Lambda_2\big] R_{+}(z) \Big) \\
&=-\int_{\mathbb{C}}dm(z)\frac{\partial \tilde{\rho}}{\partial \overline{z}}\Big(
R_{+}(z) \big[\mathcal{H}_{+},\Lambda_2\big] \big[R_{+}(z),\tilde{\Lambda}_1 S\big]  -  \big[R_{+}(z),\tilde{\Lambda}_1 S\big]  \big[\mathcal{H}_{+},\Lambda_2\big] R_{+}(z) \Big)
\end{aligned}
\end{equation*}
for any $\tilde{\Lambda}_1=\Lambda_1-\Lambda_1^{\prime}$ being the difference of two switch functions in $x_1$.

We prove \eqref{eq_spin_conduct_justified_proof_4} by a similar algebraic manipulation as in Step 3 of Proposition \ref{prop_vanish_bulk_torque_correlation}. Note that $\Sigma_{\pm}^{\Lambda_2,\tilde{\Lambda}_1 S}$ is trace-class because $\tilde{\Lambda}_1$ is localized in the $x_1$-direction (in fact, it is compactly supported in $x_1$ as the difference of two switch functions), and $\big[\mathcal{H}_{+},\Lambda_2\big]$ is localized in the $x_2$-direction. One can rigorously justify this point by following similar arguments as in Step 1, which is omitted here. Hence, by the cyclicity
\begin{equation*}
\begin{aligned}
Tr_{\mathcal{D}_{+}}(\Sigma_{+}^{\Lambda_2,\tilde{\Lambda}_1 S})
&=Tr_{\mathcal{D}_{+}}\big((P_{+}^2+(P_{+}^{\perp})^2)\Sigma_{+}^{\Lambda_2,\tilde{\Lambda}_1 S}\big)
=Tr_{\mathcal{D}_{+}}(P_{+}\Sigma_{+}^{\Lambda_2,\tilde{\Lambda}_1 S}P_{+})+Tr_{\mathcal{D}_{+}}(P_{+}^{\perp}\Sigma_{+}^{\Lambda_2,\tilde{\Lambda}_1 S}P_{+}^{\perp}) .
\end{aligned}
\end{equation*}
Note that by Proposition \ref{prop_off_diag}, for $P\in\{P_{+},P_{+}^{\perp}\}$, 
\begin{equation*}
P \Big( \int_{\mathbb{C}}\frac{\partial \tilde{\rho}}{\partial \overline{z}}R_{+}(z) \big[\mathcal{H}_{+},\Lambda_2\big] \tilde{\Lambda}_1 S R_{+}(z) dm(z)\Big) P
=P \Big( \int_{\mathbb{C}}\frac{\partial \tilde{\rho}}{\partial \overline{z}}  R_{+}(z) \tilde{\Lambda}_1 S \big[\mathcal{H}_{+},\Lambda_2\big] R_{+}(z) dm(z)\Big) P =0.
\end{equation*}
It follows that
\begin{equation*}
\begin{aligned}
Tr_{\mathcal{D}_{+}}(\Sigma_{+}^{\Lambda_2,\tilde{\Lambda}_1 S})
&=Tr_{\mathcal{D}_{+}}(P_{+}\Sigma_{+}^{\Lambda_2,\tilde{\Lambda}_1 S}P_{+})+Tr_{\mathcal{D}_{+}}(P_{+}^{\perp}\Sigma_{+}^{\Lambda_2,\tilde{\Lambda}_1 S}P_{+}^{\perp}) \\
&=-Tr_{\mathcal{D}_{+}}\Big(\sum_{P\in\{P_{+},P_{+}^{\perp}\}}
P  \Big( \int_{\mathbb{C}}\frac{\partial \tilde{\rho}}{\partial \overline{z}}
R_{+}(z) \big[\mathcal{H}_{+},\Lambda_2\big] R_{+}(z)\tilde{\Lambda}_1 S dm(z)  \Big)  P
\Big) \\
&\quad -Tr_{\mathcal{D}_{+}}\Big(\sum_{P\in\{P_{+},P_{+}^{\perp}\}}
P  \Big( \int_{\mathbb{C}}\frac{\partial \tilde{\rho}}{\partial \overline{z}} \tilde{\Lambda}_1 S R_{+}(z) \big[\mathcal{H}_{+},\Lambda_2\big] R_{+}(z) dm(z)  \Big)  P
\Big) \\
&=Tr_{\mathcal{D}_{+}}\Big(\sum_{P\in\{P_{+},P_{+}^{\perp}\}}
P  \Big( \int_{\mathbb{C}}\frac{\partial \tilde{\rho}}{\partial \overline{z}} \big[R_{+}(z),\Lambda_2\big] \tilde{\Lambda}_1 S dm(z)  \Big)  P \Big) \\
&\quad +Tr_{\mathcal{D}_{+}}\Big(\sum_{P\in\{P_{+},P_{+}^{\perp}\}}
P  \Big( \int_{\mathbb{C}}\frac{\partial \tilde{\rho}}{\partial \overline{z}} \tilde{\Lambda}_1 S \big[R_{+}(z),\Lambda_2\big] dm(z)  \Big)  P \Big) \\
&=i\pi Tr_{\mathcal{D}_{+}} \Big( \sum_{P\in\{P_{+},P_{+}^{\perp}\}}P
\big(\big[P_{+},\Lambda_2 \big]\tilde{\Lambda}_1 S +\tilde{\Lambda}_1 S\big[P_{+},\Lambda_2 \big] \big)
P\Big), 
\end{aligned}
\end{equation*}
where the Hellfer-Sjöstrand formula and $\rho(\mathcal{H}_{+})=P_{+}$ are applied to derive the last equality. Again, since $\big[P_{+},\Lambda_2 \big]\tilde{\Lambda}_1 S, \tilde{\Lambda}_1 S \big[P_{+},\Lambda_2 \big]$ are trace-class (since $\tilde{\Lambda}_1$ is localized in the $x_1$-direction), the cyclicity applies
\begin{equation*}
\begin{aligned}
Tr_{\mathcal{D}_{+}}(\Sigma_{+}^{\Lambda_2,\tilde{\Lambda}_1 S})
&=i\pi Tr_{\mathcal{D}_{+}} \Big( \sum_{P\in\{P_{+},P_{+}^{\perp}\}}P^2
\big(\big[P_{+},\Lambda_2 \big]\tilde{\Lambda}_1 S +\tilde{\Lambda}_1 S\big[P_{+},\Lambda_2 \big] \big) \Big) \\
&=i\pi Tr_{\mathcal{D}_{+}} \Big( \big[P_{+},\Lambda_2 \big]\tilde{\Lambda}_1 S +\tilde{\Lambda}_1 S\big[P_{+},\Lambda_2 \big] \Big)
=2i\pi Tr_{\mathcal{D}_{+}} \Big( \big[P_{+},\Lambda_2 \big]\tilde{\Lambda}_1 S \Big) .
\end{aligned}
\end{equation*}
The last trace vanishes as 
\begin{equation*}
Tr_{\mathcal{D}_{+}} \Big( \big[P_{+},\Lambda_2 \big]\tilde{\Lambda}_1 S \Big) = \int_{\text{supp }(\tilde{\Lambda}_1)}\text{tr}\big( P_{+}(\bm{x},\bm{x})\cdot \frac{1}{2}\sigma_{z} \big)\big(\Lambda_2(x_2)-\Lambda_2(x_2) \big)\tilde{\Lambda}_1(\bm{x}) d\bm{x} =0.
\end{equation*}
This concludes the proof of \eqref{eq_spin_conduct_justified_proof_4}.

{\color{blue}Step 4:} We conclude the proof by proving that the spin conductance \eqref{eq_spin_conduc_1} is independent of the density function $\rho$. For two different density functions $\rho_1$ and $\rho_2$ satisfying \eqref{eq_density_function}, let $\rho_{\delta}=\rho_1-\rho_1$ and let $\tilde{\rho}_{\delta}$ be an almost-analytical extension of $\rho_{\delta}$. Denote 
$$
A(z)=-R_{+}(z) \big[\mathcal{H}_{+},\Lambda_2\big] (z) \big[R_{+}(z),\Lambda_1 S\big]  +  \big[R_{+}(z),\Lambda_1 S\big]  \big[\mathcal{H}_{+},\Lambda_2\big] R_{+}(z).$$
The key point is that $A(z)$ is analytic within $\text{supp}(\tilde{\rho}_{\delta})$, a complex neighborhood of the spectral gap. By Stokes' theorem,
\begin{equation*}
\begin{aligned}
&\int_{\mathbb{C}}\frac{\partial \tilde{\rho_{\delta}}}{\partial \overline{z}}\Big[
R_{+}(z) \big[\mathcal{H}_{+},\Lambda_2\big] R_{+}(z) \big[\mathcal{H}_{+},\Lambda_1 S\big] R_{+}(z) \\
&\quad\quad\quad\quad\quad\quad\quad - R_{+}(z) \big[\mathcal{H}_{+},\Lambda_1 S\big] R_{+}(z) \big[\mathcal{H}_{+},\Lambda_2\big] R_{+}(z) \Big]dm(z) \\
&=\int_{\mathbb{C}}\frac{\partial \tilde{\rho_{\delta}}}{\partial \overline{z}}A(z)dm(z) 
=\oint_{\partial \text{supp}(\tilde{\rho}_{\delta})} \tilde{\rho_{\delta}}(z)A(z) dz =0.
\end{aligned}
\end{equation*}
This completes the proof.
\end{proof}

\section{Reduction to Kubo formula: proof of Proposition \ref{prop_kubo_formula}}

In this section, we prove that when the Hamiltonian $\mathcal{H}_{\pm}$ commutes with the spin operator, the bulk spin conductance in Definition \ref{def_spin_conducatance} recovers the well-known Kubo formula for spin conductance.

\begin{proof}[Proof of Proposition \ref{prop_kubo_formula}]
When $[\mathcal{H}_{\pm},S]=0$, the Leibniz rule gives
\begin{equation*}
[\mathcal{H}_{\pm},\Lambda_1 S]=\Lambda_1[\mathcal{H}_{\pm}, S]+[\mathcal{H}_{\pm},\Lambda_1 ]S=[\mathcal{H}_{\pm},\Lambda_1 ]S.
\end{equation*}
Hence
\begin{equation*}
\begin{aligned}
\Sigma_{\pm}^{\Lambda_2,\Lambda_1 S}
:=\int_{\mathbb{C}}dm(z)\frac{\partial \tilde{\rho}}{\partial \overline{z}}\Big[
&R_{\pm}(z) \big[\mathcal{H}_{\pm},\Lambda_2\big] R_{\pm}(z) \big[\mathcal{H}_{\pm},\Lambda_1 \big]S R_{\pm}(z) \\
&- R_{\pm}(z) \big[\mathcal{H}_{\pm},\Lambda_1 \big]S R_{\pm}(z) \big[\mathcal{H}_{\pm},\Lambda_2\big] R_{\pm}(z) \Big] .
\end{aligned}
\end{equation*}
It follows that $\Sigma_{\pm}^{\Lambda_2,\Lambda_1 S}$ is trace-class, as the commutators with $\Lambda_i$ for $i=1,2$ yield localization in the $x_i$-directions. Thus, the principal-value trace equals the conventional trace by Proposition \ref{prop_pv_usual_trace}:
\begin{equation*}
Tr_{ka_1^{\pm},\mathcal{D}_{\pm}}^{pv,1}(\Sigma_{\pm}^{\Lambda_2,\Lambda_1 S})
=Tr_{\mathcal{D}_{\pm}}(\Sigma_{\pm}^{\Lambda_2,\Lambda_1 S}) .
\end{equation*}
For the trace on the right side, the cyclicity gives
\begin{equation*}
\begin{aligned}
&Tr_{\mathcal{D}_{\pm}}(\Sigma_{\pm}^{\Lambda_2,\Lambda_1 S}) \\
&=Tr_{\mathcal{D}_{\pm}}(P_{\pm}\Sigma_{\pm}^{\Lambda_2,\Lambda_1 S }P_{\pm} + P_{\pm}^{\perp}\Sigma_{\pm}^{\Lambda_2,\Lambda_1 S }P_{\pm}^{\perp}) \\
&=-\sum_{P\in\{P_{\pm},P_{\pm}^{\perp}\}}Tr_{\mathcal{D}_{\pm}}\Big(P \Big( \int_{\mathbb{C}}\frac{\partial \tilde{\rho}}{\partial \overline{z}}\Big( R_{\pm}(z) \big[\mathcal{H}_{\pm},\Lambda_2\big] R_{\pm}(z) \Lambda_1 S + \Lambda_1 S R_{\pm}(z) \big[\mathcal{H}_{\pm},\Lambda_2\big] R_{\pm}(z) \Big)dm(z)
\Big)  P\Big) \\
&\overset{(i)}{=}\sum_{P\in\{P_{\pm},P_{\pm}^{\perp}\}}Tr_{\mathcal{D}_{\pm}}\Big(P \Big( \int_{\mathbb{C}}\frac{\partial \tilde{\rho}}{\partial \overline{z}}\Big( \big[R_{\pm}(z),\Lambda_2\big] \Lambda_1 S + \Lambda_1 S \big[R_{\pm}(z),\Lambda_2\big] \Big)dm(z)
\Big)  P\Big) \\
&\overset{(ii)}{=}i\pi\sum_{P\in\{P_{\pm},P_{\pm}^{\perp}\}} Tr_{\mathcal{D}_{\pm}}\Big(P \Big(
\big[P_{\pm},\Lambda_2 \big]\Lambda_1 S+\Lambda_1 S\big[P_{\pm},\Lambda_2 \big]
\Big)  P\Big),
\end{aligned}
\end{equation*}
where the equality $(i)$ follows from the identity $R_{\pm}(z) \big[\mathcal{H}_{\pm},\Lambda_2\big] R_{\pm}(z)=-\big[R_{\pm}(z),\Lambda_2\big]$, and $(ii)$ from the Hellfer-Sjöstrand formula. One proceeds by repeating the algebraic manipulations as in Step 3 of Proposition \ref{prop_vanish_bulk_torque_correlation}, making use of the cyclicity and orthogonality properties of the projections:
\begin{equation*}
\begin{aligned}
&Tr_{\mathcal{D}_{\pm}}(\Sigma_{\pm}^{\Lambda_2,\Lambda_1 S}) \\
&=i\pi\sum_{P\in\{P_{\pm},P_{\pm}^{\perp}\}} Tr_{\mathcal{D}_{\pm}}\Big(P \Big(
\big[P_{\pm},\Lambda_2 \big]\Lambda_1 S+\Lambda_1 S\big[P_{\pm},\Lambda_2 \big]
\Big)  P\Big) \\
&=i\pi Tr_{\mathcal{D}_{\pm}}\Big(P_{\pm} \Big(
\Lambda_2 P_{\pm}^{\perp} \Lambda_1 S-\Lambda_1 SP_{\pm}^{\perp}\Lambda_2
\Big)  P_{\pm}\Big)
+i\pi Tr_{\mathcal{D}_{\pm}}\Big(P_{\pm}^{\perp} \Big(
-\Lambda_2 P_{\pm} \Lambda_1 S+\Lambda_1 S P_{\pm}\Lambda_2
\Big)  P_{\pm}^{\perp}\Big) \\
&=i\pi Tr_{\mathcal{D}_{\pm}}\Big(P_{\pm} \Big(
-\big[ P_{\pm}^{\perp},\Lambda_2\big] \big[ P_{\pm}^{\perp},\Lambda_1 S\big]+\big[ P_{\pm}^{\perp},\Lambda_1 S\big]\big[ P_{\pm}^{\perp},\Lambda_2\big]
\Big)  P_{\pm}\Big) \\
&\quad +i\pi Tr_{\mathcal{D}_{\pm}}\Big(
\big[ P_{\pm}^{\perp},\Lambda_2\big]P_{\pm}\big[ P_{\pm}^{\perp},\Lambda_1 S\big] - \big[ P_{\pm}^{\perp},\Lambda_1 S\big]P_{\pm}\big[ P_{\pm}^{\perp},\Lambda_2\big] \Big) \\
&=i\pi Tr_{\mathcal{D}_{\pm}}\Big( -P_{\pm}\big[ P_{\pm},\Lambda_2\big] \big[ P_{\pm},\Lambda_1 S\big]P_{\pm}+P_{\pm}\big[ P_{\pm},\Lambda_1 S\big]\big[ P_{\pm},\Lambda_2\big]P_{\pm} \Big) \\
&\quad +i\pi Tr_{\mathcal{D}_{\pm}}\Big(
P_{\pm}\big[ P_{\pm},\Lambda_1 S\big] \big[ P_{\pm},\Lambda_2 \big]P_{\pm} - P_{\pm}\big[ P_{\pm},\Lambda_2\big] \big[ P_{\pm},\Lambda_1 S\big]P_{\pm} \Big) \\
&=2i\pi Tr_{\mathcal{D}_{\pm}}\Big(
P_{\pm}\Big[\big[ P_{\pm},\Lambda_1 S\big], \big[ P_{\pm},\Lambda_2 \big]\Big]P_{\pm}\Big) .
\end{aligned}
\end{equation*}
This concludes the proof of $\sigma_{\pm}^{\Lambda_2}=iTr_{\mathcal{D}_{\pm}}\Big(P_{\pm}\Big[ [P_{\pm},S\Lambda_1],[P_{\pm},\Lambda_2] \Big] \Big)$, with the right side being exactly the Kubo formula for the spin conductance \cite{elgart2005shortrange+functional,Avron1994charge,Marcelli2019spin_conductivity,marcelli2021new,Giovanna22charge_to_spin}. Since the spin operator commutes with $\mathcal{H}_{\pm}$, the ground state projection $P_{\pm}$ is decomposed according to the eigenspace of $S$. Denoting the component of $P_{\pm}$ lying in the spin-up/spin-down spaces (the $\frac{1}{2}/\frac{-1}{2}$ eigenspace of $S$) as $P_{\pm}^{\uparrow/\downarrow}$, respectively, we see
\begin{equation*}
\begin{aligned}
&iTr_{\mathcal{D}_{\pm}}\Big(P_{\pm}\Big[ [P_{\pm},S\Lambda_1],[P_{\pm},\Lambda_2] \Big] \Big) \\
&=\frac{i}{2}Tr_{\mathcal{D}_{\pm}}\Big(P_{\pm}^{\uparrow}\Big[ [P_{\pm}^{\uparrow},\Lambda_1],[P_{\pm}^{\uparrow},\Lambda_2] \Big] \Big) - \frac{i}{2}Tr_{\mathcal{D}_{\pm}}\Big(P_{\pm}^{\downarrow}\Big[ [P_{\pm}^{\downarrow},\Lambda_1],[P_{\pm}^{\downarrow},\Lambda_2] \Big] \Big) \\
&=\frac{1}{2} \mathcal{C}(P_{\pm}^{\uparrow})-\frac{1}{2} \mathcal{C}(P_{\pm}^{\downarrow}) .
\end{aligned}
\end{equation*}
This concludes the proof \eqref{eq_kubo_formula}. When $\mathcal{H}_{\pm}$ is TR symmetric, $\mathcal{C}(P_{\pm}^{\uparrow})=-\mathcal{C}(P_{\pm}^{\downarrow})$ \cite{avila2013shortrange+transfer}, and hence $\sigma_{\pm}^{\Lambda_2}=\mathcal{C}(P_{\pm}^{\uparrow})$. Therefore, $\sigma_{\pm}^{\Lambda_2} \text{ mod 2}$ recovers the well-known Fu-Kane-Mele $\mathbb{Z}_2$ index which equals exactly the parity of spin-up Chern number \cite{Vanderbilt2019Berry_phase}.
\end{proof}

\section{Interface spin-torque conductance: proof of Proposition \ref{prop_interface_spin_torque_justified}}

In this section, we prove Proposition \ref{prop_interface_spin_torque_justified}. The key point is, again, the vanishing mesoscopic average of spin-torque response in the bulk (Proposition \ref{prop_vanish_bulk_torque_correlation}).

\begin{proof}[Proof of Proposition \ref{prop_interface_spin_torque_justified}]

{\color{blue}Step 1:}  By the interface structure condition \eqref{eq_interface_structure}, for suffciently enough $N\in\mathbb{N}$, there exist $N_{+},N_{-}\in\mathbb{N}$ such that $N_{\pm}<N$ and $N_{\pm}a_{1}^{\pm}>L$. We write
\begin{equation} \label{eq_interface_spin_torque_justified_proof_1}
\begin{aligned}
\sum_{n=-N}^{N-1}Tr_{\mathcal{D}_{e}}(\mathbbm{1}_{\Omega_{n,a_1^{com}}\cap \mathcal{D}_{e}}\Sigma_{e}^{\Lambda_2,S})
&=\sum_{n=-N_{-}}^{N_{+}-1}Tr_{\mathcal{D}_{e}}(\mathbbm{1}_{\Omega_{n,a_1^{com}}\cap \mathcal{D}_{e}}\Sigma_{e}^{\Lambda_2,S}) \\
&\quad + \sum_{n=N_{+}}^{N-1}Tr_{\mathcal{D}_{e}}(\mathbbm{1}_{\Omega_{n,a_1^{com}}\cap \mathcal{D}_{e}}\Sigma_{e}^{\Lambda_2,S}) \\
&\quad + \sum_{n=-N}^{-N_{-}-1}Tr_{\mathcal{D}_{e}}(\mathbbm{1}_{\Omega_{n,a_1^{com}}\cap \mathcal{D}_{e}}\Sigma_{e}^{\Lambda_2,S}).
\end{aligned}
\end{equation}
Since the indicator $\mathbbm{1}_{\Omega_{n,a_1^{com}}\cap \mathcal{D}}$ is localized in the $x_1$-direction and $\Sigma_{e}^{\Lambda_2,S}$ is localized along $x_2$ due to the commutator $[\mathcal{H}_e,\Lambda_2]$ (a rigorous justification follows similar lines as in Step 1 of Proposition \ref{prop_vanish_bulk_torque_correlation}), 
the operators $\mathbbm{1}_{\Omega_{n,a_1^{com}}\cap \mathcal{D}_{e}}\Sigma_{e}^{\Lambda_2,S}$ are trace-class. It follows that  
the first sum on the right side of \eqref{eq_interface_spin_torque_justified_proof_1} is finite and and is independent of $N$:
\begin{equation} \label{eq_interface_spin_torque_justified_proof_2}
\sum_{n=-N_{-}}^{N_{+}-1}Tr_{\mathcal{D}_{e}}(\mathbbm{1}_{\Omega_{n,a_1^{com}}\cap \mathcal{D}_{e}}\Sigma_{e}^{\Lambda_2,S}) < \infty.
\end{equation}

{\color{blue}Step 2:} We prove that 
\begin{equation} \label{eq_interface_spin_torque_justified_proof_3}
\sum_{n=N_{+}}^{\infty}\Big| Tr_{\mathcal{D}_{e}}(\mathbbm{1}_{\Omega_{n,a_1^{com}}\cap \mathcal{D}_{e}}\Sigma_{e}^{\Lambda_2,S}) \Big| < \infty, 
\end{equation}
and in a similar manner
\begin{equation*}
\sum_{n=-\infty}^{-N_{-}-1}\Big| Tr_{\mathcal{D}_{e}}(\mathbbm{1}_{\Omega_{n,a_1^{com}}\cap \mathcal{D}_{e}}\Sigma_{e}^{\Lambda_2,S}) \Big|<\infty .
\end{equation*}
Then, together with \eqref{eq_interface_spin_torque_justified_proof_1} and \eqref{eq_interface_spin_torque_justified_proof_2}, we can conclude that \eqref{eq_interface_torque_1} is well-defined and hence complete the proof.

The key to the proof of \eqref{eq_interface_spin_torque_justified_proof_3} is that the terms in the summation decay exponentially as $n\to\infty$, since the spin-torque response vanishes mesoscopically in the bulk $\mathcal{D}_{+}$, as shown below. Indeed, by Proposition \ref{prop_vanish_bulk_torque_correlation}, for $n\geq N_{+}$,
\begin{equation*}
Tr_{\mathcal{D}_{e}}(\mathbbm{1}_{\Omega_{n,a_1^{com}}\cap \mathcal{D}_{e}}\Sigma_{+}^{\Lambda_2,S})
=Tr_{\mathcal{D}_{+}}(\mathbbm{1}_{\Omega_{n,a_1^{com}}\cap \mathcal{D}_{+}}\Sigma_{+}^{\Lambda_2,S})
=0.
\end{equation*}
Hence for $n\geq N_{+}$,
\begin{equation} \label{eq_interface_spin_torque_justified_proof_4}
Tr_{\mathcal{D}_{e}}(\mathbbm{1}_{\Omega_{n,a_1^{com}}\cap \mathcal{D}_{e}}\Sigma_{e}^{\Lambda_2,S})
=Tr_{\mathcal{D}_{+}}\Big(\mathbbm{1}_{\Omega_{n,a_1^{com}}\cap \mathcal{D}_{+}}\big(\Sigma_{e}^{\Lambda_2,S}-\Sigma_{+}^{\Lambda_2,S}\big)\Big) .
\end{equation}
By definition \eqref{eq_interface_torque_2} and the resolvent identity $R_{e}(z)-R_{+}(z)=R_{e}(z)(\mathcal{H}_{+}-\mathcal{H}_{e})R_{+}(z)$, we see that
$$
\Sigma_{e}^{\Lambda_2,S}-\Sigma_{+}^{\Lambda_2,S}=\int_{\mathbb{C}}\frac{\partial \tilde{\rho}}{\partial \overline{z}}A(z)dm(z)
$$ 
where $A(z)$ is a finite product of operators from the set
$$
\Big\{\mathcal{H}_{+}-\mathcal{H}_{e},[\mathcal{H}_{+}-\mathcal{H}_{e},\Lambda_2],[\mathcal{H}_{+}-\mathcal{H}_{e},S],R_{*}(z),[\mathcal{H}_{*},\Lambda_2],[\mathcal{H}_{*},S] :\, *\in\{e,+\} \Big\}.
$$
Note that $A(z)$ contains at least one operator involving $\mathcal{H}_{+}-\mathcal{H}_{e}$ and at least one commutator with $\Lambda_2$. This ensures that $A(z)(\bm{x},\bm{y})$ decays as $|x_2|,|y_2|\to +\infty$ and $x_1,y_1\to +\infty$. More precisely, note the following estimates by the Lemmas in Section 2.1:
\begin{equation*}
\left\{
\begin{aligned}
&|(\mathcal{H}_{+}-\mathcal{H}_{e})(\bm{x},\bm{y})|\leq C_1e^{-D_1(\|\bm{x}-\bm{y}\|_{1}+\text{dist}(\bm{x},\Omega_{-})+\text{dist}(\bm{y},\Omega_{-}))}, \\
&|[\mathcal{H}_{+}-\mathcal{H}_{e},\Lambda_2](\bm{x},\bm{y})|
=|(\mathcal{H}_{+}-\mathcal{H}_{e})(\bm{x},\bm{y})||\Lambda_2(\bm{x})-\Lambda_2(\bm{y})|
\leq C_2e^{-D_2(\|\bm{x}-\bm{y}\|_{1}+\text{dist}(\bm{x},\Omega_{-})+\text{dist}(\bm{y},\Omega_{-})+|x_2|+|y_2|)}, \\
&|[\mathcal{H}_{+}-\mathcal{H}_{e},S](\bm{x},\bm{y})|
=|(\mathcal{H}_{+}-\mathcal{H}_{e})(\bm{x},\bm{y})||S(\bm{x})-S(\bm{y})|
\leq C_3 e^{-D_3(\|\bm{x}-\bm{y}\|_{1}+\text{dist}(\bm{x},\Omega_{-})+\text{dist}(\bm{y},\Omega_{-}))}, \\
&|R_{*}(z)(\bm{x},\bm{y})|\leq \frac{C_4}{|\text{Im }z|}e^{-D_4|\text{Im }z|\|\bm{x}-\bm{y}\|_{1}}, \\
&|[\mathcal{H}_{*},S](\bm{x},\bm{y})|\leq C_5e^{-D_5\|\bm{x}-\bm{y}\|_{1}}, \\
&|[\mathcal{H}_{*},\Lambda_2](\bm{x},\bm{y})| \leq C_6 e^{-D_3(\|\bm{x}-\bm{y}\|_{1}+|x_2|+|y_2|)},
\end{aligned}
\right.
\end{equation*}
where $*\in\{+,e\}$, $C_i,D_i$ depend only on $\Lambda_2$. By Lemma \ref{lem_product_localized}, the following estimate of $A(z)(\bm{x},\bm{y})$ holds for all $|\text{Im }z|\leq 1$
\begin{equation} \label{eq_interface_spin_torque_justified_proof_5}
\big| A(z)(\bm{x},\bm{y}) \big| \leq \frac{C}{|\text{Im }z|^{p}}e^{-\frac{D|\text{Im }z|}{4}(\|\bm{x}-\bm{y}\|_{1}+\text{dist}(\bm{x},\Omega_{-})+\text{dist}(\bm{y},\Omega_{-})+|x_2|+|y_2|)}
\end{equation}
for some $p\in\mathbb{N}$, $D=\min_{1\leq i\leq 6} D_i$ and $C>0$ depending only on $\Lambda_2$. This implies 
$$
\|\mathbbm{1}_{\Omega_{n,a_1^{com}}\cap \mathcal{D}_{+}}A(z)\|_{\mathscr{T}_{\mathcal{D}_{+}}}\leq \frac{C}{|\text{Im }z|^{p+2}}
$$
by Lemma \ref{lem_trace_norm_kernel} and \ref{lem_product_localized} (analogous to \eqref{eq_vanish_bulk_torque_correlation_proof_5} in Section 3). Hence, by the almost-analyticity of $\tilde{\rho}$
\begin{equation*}
\Big\|\mathbbm{1}_{\Omega_{n,a_1^{com}}\cap \mathcal{D}_{+}}\big(\Sigma_{e}^{\Lambda_2,S}-\Sigma_{+}^{\Lambda_2,S}\big)\mathbbm{1}_{\Omega_{n,a_1^{com}}\cap \mathcal{D}_{+}}\Big\|_{\mathscr{T}_{\mathcal{D}_{+}}}
\leq \int_{\mathbb{C}}\Big|\frac{\partial \tilde{\rho}}{\partial \overline{z}}\Big| \|\mathbbm{1}_{\Omega_{n,a_1^{com}}\cap \mathcal{D}_{+}}A(z)\mathbbm{1}_{\Omega_{n,a_1^{com}}\cap \mathcal{D}_{+}}\|_{\mathscr{T}_{\mathcal{D}_{+}}} <\infty .
\end{equation*}
This implies we can interchange the trace and complex integral when estimating \eqref{eq_interface_spin_torque_justified_proof_4}
\begin{equation*}
\begin{aligned}
&\Big|Tr_{\mathcal{D}_{+}}\Big(\mathbbm{1}_{\Omega_{n,a_1^{com}}\cap \mathcal{D}_{+}}\big(\Sigma_{e}^{\Lambda_2,S}-\Sigma_{+}^{\Lambda_2,S}\big)\mathbbm{1}_{\Omega_{n,a_1^{com}}\cap \mathcal{D}_{+}}\Big)\Big| \\
&=\Big|\int_{\Omega_{n,a_1^{com}}\cap \mathcal{D}_{+}} d\bm{x}
\int_{\mathbb{C}} \frac{\partial \tilde{\rho}}{\partial \overline{z}}A(z)(\bm{x},\bm{x})dm(z) \Big| 
=\Big|\int_{\mathbb{C}} \frac{\partial \tilde{\rho}}{\partial \overline{z}} dm(z)
\int_{\Omega_{n,a_1^{com}}\cap \mathcal{D}_{+}} 
A(z)(\bm{x},\bm{x}) d\bm{x} \Big| \\
&\leq \int_{\mathbb{C}} \Big|\frac{\partial \tilde{\rho}}{\partial \overline{z}} \Big| dm(z)
\int_{\Omega_{n,a_1^{com}}\cap \mathcal{D}_{+}} 
\big|A(z)(\bm{x},\bm{x})\big| d\bm{x} .
\end{aligned}
\end{equation*}
Applying \eqref{eq_interface_spin_torque_justified_proof_5}, we see
\begin{equation*}
\begin{aligned}
&\Big|Tr_{\mathcal{D}_{+}}\Big(\mathbbm{1}_{\Omega_{n,a_1^{com}}\cap \mathcal{D}_{+}}\big(\Sigma_{e}^{\Lambda_2,S}-\Sigma_{+}^{\Lambda_2,S}\big)\mathbbm{1}_{\Omega_{n,a_1^{com}}\cap \mathcal{D}_{+}}\Big)\Big| \\
&\leq C\int_{\mathbb{C}} \Big|\frac{\partial \tilde{\rho}}{\partial \overline{z}} \Big|\frac{1}{|\text{Im }z|^{p}} dm(z)
\int_{\Omega_{n,a_1^{com}}\cap \mathcal{D}_{+}} e^{-\frac{D|\text{Im }z|}{2}(\text{dist}(\bm{x},\Omega_{-})+|x_2|)} d\bm{x} \\
&\leq C\int_{\mathbb{C}} \Big|\frac{\partial \tilde{\rho}}{\partial \overline{z}} \Big|\frac{1}{|\text{Im }z|^{p+1}}e^{-\frac{D|\text{Im }z|}{2}\text{dist}(\Omega_{n,a_1^{com}}\cap \mathcal{D}_{+},\Omega_{-})} dm(z) \\
&\leq C\int_{\mathbb{C}} \Big|\frac{\partial \tilde{\rho}}{\partial \overline{z}} \Big|\frac{1}{|\text{Im }z|^{p+1}} e^{-nD|\text{Im }z|}  dm(z) .
\end{aligned}
\end{equation*}
Here the constants $C,D$ are modified accordingly in each step but remain independent of $n$ and $\text{Im }z$. Sum over $n$ and using \eqref{eq_interface_spin_torque_justified_proof_4} and the almost-analyticity of $\tilde{\rho}$, we conclude that
\begin{equation*}
\begin{aligned}
\sum_{n=N_{+}}^{\infty}\Big| Tr_{\mathcal{D}_{e}}(\mathbbm{1}_{\Omega_{n,a_1^{com}}\cap \mathcal{D}_{e}}\Sigma_{e}^{\Lambda_2,S}\mathbbm{1}_{\Omega_{n,a_1^{com}}\cap \mathcal{D}_{e}}) \Big|
&\leq C\int_{\mathbb{C}} \Big|\frac{\partial \tilde{\rho}}{\partial \overline{z}} \Big|\frac{1}{|\text{Im }z|^{p+1}} \Big(\sum_{n=N_{+}}^{\infty}e^{-nD|\text{Im }z|}\Big)  dm(z) \\
&\leq C\int_{\mathbb{C}} \Big|\frac{\partial \tilde{\rho}}{\partial \overline{z}} \Big|\frac{1}{|\text{Im }z|^{p+1}(1-e^{-D|\text{Im z}|})}  dm(z) \\
&\leq C\int_{-1}^{1} \frac{|\text{Im }z|^2}{1-e^{-D|\text{Im z}|}}  d|\text{Im }z| <\infty .
\end{aligned}
\end{equation*}
Hence \eqref{eq_interface_spin_torque_justified_proof_3} is proved.
\end{proof}

\section{Bulk-interface correspondence: proof of Theorem \ref{thm_bic}}

Once the bulk and interface spin conductivities are defined, the BIC \eqref{eq_bic} follows from the physical principal: the conservation law of spin transport. Before showing the proof, we illustrate the physical idea as follows. Recall the thought experiment in Figure \ref{fig_thought experiment}. For a large strip $\Omega_{R}=[-R,R]\times \mathbb{R}$ centered at the origin ($R\gg L$), the outward spin current flowing across its right boundary is characterized by 
$$
[\mathcal{H}_{e},S\mathbbm{1}_{\Omega_{R}^{c}\cap \Omega_{+}}]\simeq [\mathcal{H}_{+},S\mathbbm{1}_{\Omega_{R}^{c}\cap \Omega_{+}}] .
$$
As seen from Definition \ref{def_spin_conducatance}, this spin flow is described exactly by the bulk spin conductivity $\sigma_{+}^{\Lambda_2}$ (by setting $\Lambda_1=\mathbbm{1}_{\Omega_{R}^{c}\cap \Omega_{+}}$ in \eqref{eq_spin_conduc_2}). Similarly, the inward spin current flowing across the left boundary is given by
$$
[\mathcal{H}_{e},S\mathbbm{1}_{\Omega_{R}^{c}\cap \Omega_{-}}]
\simeq [\mathcal{H}_{-},S\mathbbm{1}_{\Omega_{R}^{c}\cap \Omega_{-}}]
\overset{\Lambda_1=1-\mathbbm{1}_{\Omega_{R}^{c}\cap \Omega_{-}}}{=} [\mathcal{H}_{-},S]-[\mathcal{H}_{-},\Lambda_2] \sim 0-\sigma_{-}^{\Lambda_2},
$$
where the fact that the response of the spin-torque $[\mathcal{H}_{-},S]$ vanishes (mesoscopically) in the left bulk is used in the last step. The longitudinal current is characterized by $[\mathcal{H}_{e},S\mathbbm{1}_{\Omega_{R}}]$, while the spin-torque response inside the box is described by $[\mathcal{H}_{e},S]\mathbbm{1}_{\Omega_{R}}\simeq [\mathcal{H}_{e},S]$. In the operator language, the conservation law is simply the following identity
\begin{equation*}
\underbrace{[\mathcal{H}_{e},S]}_{\sim \sigma^{torque}_{e}}=[\mathcal{H}_{e},S(\mathbbm{1}_{\Omega_{R}^{c}\cap \Omega_{-}}+\mathbbm{1}_{\Omega_{R}^{c}\cap \Omega_{-}}+\mathbbm{1}_{\Omega_{R}})]
\simeq \underbrace{[\mathcal{H}_{+},S\mathbbm{1}_{\Omega_{R}^{c}\cap \Omega_{+}}]}_{\sim \sigma_{+}}
- \underbrace{[\mathcal{H}_{-},S(1-\mathbbm{1}_{\Omega_{R}^{c}\cap \Omega_{+}})]}_{\sim \sigma_{-}}
+\underbrace{[\mathcal{H}_{e},S\mathbbm{1}_{\Omega_{R}}]}_{\sim \sigma_{e}^{drift}}.
\end{equation*}
This informal argument is made rigorous in the sequel.

\begin{proof}[Proof of Theorem \ref{thm_bic}]

{\color{blue}Step 1:} As outlined above, we start from the interface spin-torque conductance $\sigma_{e}^{torque,\Lambda_2,\rho}$. Fix $R>L$. By its definition \eqref{eq_interface_torque_1}, we write
\begin{equation} \label{eq_bic_proof_1}
2\pi\sigma_{e}^{torque,\Lambda_2,\rho}= Tr_{a_1^{com},\mathcal{D}_{e}}^{pv,1}(\Sigma_{e}^{\Lambda_2,S\mathbbm{1}_{\Omega_{R}^{c}\cap \Omega_{+}\cap \mathcal{D}_{e}}})
+Tr_{a_1^{com},\mathcal{D}_{e}}^{pv,1}(\Sigma_{e}^{\Lambda_2,S\mathbbm{1}_{\Omega_{R}^{c}\cap \Omega_{-}\cap \mathcal{D}_{e}}}) +Tr_{a_1^{com},\mathcal{D}_{e}}^{pv,1}(\Sigma_{e}^{\Lambda_2,S\mathbbm{1}_{\Omega_{R}\cap \mathcal{D}_{e}}})
\end{equation}
with
\begin{equation*}
\begin{aligned}
\Sigma_{e}^{\Lambda_2, S\mathbbm{1}_{U}}
&:=\int_{\mathbb{C}}dm(z)\frac{\partial \tilde{\rho}}{\partial \overline{z}}\Big[
R_{e}(z) \big[\mathcal{H}_{e},\Lambda_2\big] R_{e}(z) \big[\mathcal{H}_{e}, S\mathbbm{1}_{U}\big] R_{e}(z) \\
&\quad\quad\quad\quad\quad\quad\quad- R_{e}(z) \big[\mathcal{H}_{e}, S\mathbbm{1}_{U}\big] R_{e}(z) \big[\mathcal{H}_{e},\Lambda_2\big] R_{e}(z)
\Big] \\
&=:\int_{\mathbb{C}}\frac{\partial \tilde{\rho}}{\partial \overline{z}} \Sigma_{e}^{\Lambda_2, S\mathbbm{1}_{U}}(z) dm(z)
\end{aligned}
\end{equation*}
for $U\in\{\Omega_{R}\cap \mathcal{D}_{e},\Omega_{R}^{c}\cap\Omega_{+}\cap \mathcal{D}_{e},\Omega_{R}^{c}\cap\Omega_{-}\cap \mathcal{D}_{e} \}$. We show that the three terms in \eqref{eq_bic_proof_1} converge to $\sigma_{+}^{\Lambda_2}$, $-\sigma_{-}^{\Lambda_2}$, and $-\sigma_{e}^{drift,\Lambda_2/\rho}$ as $R\to\infty$ in Step 2, 3 and 4, respectively.

{\color{blue}Step 2:} We decompose $Tr_{a_1^{com},\mathcal{D}_{e}}^{pv,1}(\Sigma_{e}^{\Lambda_2,S\mathbbm{1}_{\Omega_{R}^{c}\cap \Omega_{+}\cap \mathcal{D}_{e}}})$ as
\begin{equation*}
\begin{aligned}
Tr_{a_1^{com},\mathcal{D}_{e}}^{pv,1}(\Sigma_{e}^{\Lambda_2,S\mathbbm{1}_{\Omega_{R}^{c}\cap \Omega_{+}\cap \mathcal{D}_{e}}})
&=Tr_{a_1^{com},\mathcal{D}_{e}}^{pv,1}(\mathbbm{1}_{\{x_1\leq L\}}\Sigma_{e}^{\Lambda_2,S\mathbbm{1}_{\Omega_{R}^{c}\cap \Omega_{+}\cap \mathcal{D}_{e}}})
+Tr_{a_1^{com},\mathcal{D}_{e}}^{pv,1}(\mathbbm{1}_{\{x_1>L\}}\Sigma_{e}^{\Lambda_2,S\mathbbm{1}_{\Omega_{R}^{c}\cap \Omega_{+}\cap \mathcal{D}_{e}}}) \\
&=Tr_{a_1^{com},\mathcal{D}_{e}}^{pv,1}(\mathbbm{1}_{\{x_1\leq L\}}\Sigma_{e}^{\Lambda_2,S\mathbbm{1}_{\Omega_{R}^{c}\cap \Omega_{+}\cap \mathcal{D}_{e}}})
+Tr_{a_1^{com},\mathcal{D}_{+}}^{pv,1}(\mathbbm{1}_{\{x_1>L\}}\Sigma_{e}^{\Lambda_2,S\mathbbm{1}_{\Omega_{R}^{c}\cap \Omega_{+}\cap \mathcal{D}_{e}}}),
\end{aligned}
\end{equation*}
where the last equality follows from \eqref{eq_interface_structure}. We prove the following two equalities separately, in two steps that follow:
\begin{equation} \label{eq_bic_proof_2}
\lim_{R\to\infty} Tr_{a_1^{com},\mathcal{D}_{e}}^{pv,1}(\mathbbm{1}_{\{x_1\leq L\}}\Sigma_{e}^{\Lambda_2,S\mathbbm{1}_{\Omega_{R}^{c}\cap \Omega_{+}\cap \mathcal{D}_{e}}})
=0,
\end{equation}
and
\begin{equation} \label{eq_bic_proof_3}
\lim_{R\to\infty} Tr_{a_1^{com},\mathcal{D}_{+}}^{pv,1}(\mathbbm{1}_{\{x_1> L\}}\Sigma_{e}^{\Lambda_2,S\mathbbm{1}_{\Omega_{R}^{c}\cap \Omega_{+}\cap \mathcal{D}_{e}}})
=2\pi\sigma_{+}^{\Lambda_2}.
\end{equation}
Then we conclude that
\begin{equation*}
Tr_{a_1^{com},\mathcal{D}_{e}}^{pv,1}(\Sigma_{e}^{\Lambda_2,S\mathbbm{1}_{\Omega_{R}^{c}\cap \Omega_{+}\cap \mathcal{D}_{e}}})=2\pi\sigma_{+}^{\Lambda_2} .
\end{equation*}

{\color{blue}Step 2.1:} We prove \eqref{eq_bic_proof_2}. The physical intuition for \eqref{eq_bic_proof_2} is illustrated as follows: since $\Sigma_{e}^{\Lambda_2,S\mathbbm{1}_{\Omega_{R}^{c}\cap \Omega_{+}\cap \mathcal{D}_{e}}}$ measures the spin current across the boundary $\partial (\Omega_{R}^{c}\cap \Omega_{+})$ of the right half plane $\Omega_{R}^{c}\cap \Omega_{+}$, its expectation in the left half plane $\{x_1\leq L\}$ necessarily vanishes as $R\to\infty$ because $\text{dist}(\Omega_{R}^{c}\cap \Omega_{+},\{x_1\leq L\})=R-L\to\infty$.

First we show the $\Sigma_{e}^{\Lambda_2,S\mathbbm{1}_{\Omega_{R}^{c}\cap \Omega_{+}\cap \mathcal{D}_{e}}}\mathbbm{1}_{\{x_1\leq L\}}$ is trace-class; hence the principal-value trace in \eqref{eq_bic_proof_2} can be replaced by the conventional trace. This follows from the fact that $[\mathcal{H}_{e},\Lambda_2]$ is localized in the $x_2$-direction, and more importantly, $\mathbbm{1}_{\Omega_{R}^{c}\cap \Omega_{+}}$ (decays as $x_1\to -\infty$) and $\mathbbm{1}_{\{x_1\leq L\}}$ (decays as $x_1\to \infty$) together provide the localization in $x_1$. In fact, note the following estimates (analogous to the ones in the previous sections): 
\begin{equation*} 
\left\{
\begin{aligned}
&|R_{e}(z)(\bm{x},\bm{y})|\leq \frac{C_1}{|\text{Im }z|}e^{-D_1|\text{Im }z|\|\bm{x}-\bm{y}\|_{1}}, \\
&|S(\bm{x},\bm{y})|\leq C_2e^{-D_2\|\bm{x}-\bm{y}\|_{1}}, \\
&|\mathcal{H}_{e}(\bm{x},\bm{y})| \leq C_3e^{-D_3\|\bm{x}-\bm{y}\|_{1}}, \\
&|[\mathcal{H}_{e},\Lambda_2](\bm{x},\bm{y})| \leq C_4e^{-D_4(\|\bm{x}-\bm{y}\|_{1}+|x_2|+|y_2|)}, \\
&|\mathbbm{1}_{\Omega_{R}^{c}\cap \Omega_{+}\cap\mathcal{D}_{e}}(\bm{x},\bm{y})|
\leq C_{5}e^{-D_5(\|\bm{x}-\bm{y}\|_{1}+\text{dist}(\bm{x},\Omega_{R}^{c}\cap \Omega_{+}\cap\mathcal{D}_{e})+\text{dist}(\bm{y},\Omega_{R}^{c}\cap \Omega_{+}\cap\mathcal{D}_{e}))}, \\
&|\mathbbm{1}_{\{x_1\leq L\}}(\bm{x},\bm{y})|
\leq C_{6}e^{-D_6(\|\bm{x}-\bm{y}\|_{1}+\text{dist}(\bm{x},\{x_1\leq L\})+\text{dist}(\bm{y},\{x_1\leq L\})}. 
\end{aligned}
\right.
\end{equation*}
By Lemma \ref{lem_product_localized}, there exists $C,D,p>0$ depending only on $\Lambda_2$ such that
\begin{equation*}
\begin{aligned}
&\Big|\Big(\mathbbm{1}_{\{x_1\leq L\}}\Sigma_{e}^{\Lambda_2,S\mathbbm{1}_{\Omega_{R}^{c}\cap \Omega_{+}\cap\mathcal{D}_{e}}}(z)\Big)(\bm{x},\bm{y})\Big| \\
&\leq \frac{C}{|\text{Im }z|^{p}}e^{-2D|\text{Im }z|\big(\|\bm{x}-\bm{y}\|_{1}+|x_2|+|y_2|+\text{dist}(\bm{x},\Omega_{R}^{c}\cap \Omega_{+}\cap\mathcal{D}_{e})+\text{dist}(\bm{y},\Omega_{R}^{c}\cap \Omega_{+}\cap\mathcal{D}_{e})+\text{dist}(\bm{x},\{x_1\leq L\})+\text{dist}(\bm{y},\{x_1\leq L\}) \big)}.
\end{aligned}
\end{equation*}
By the triangle inequality
\begin{equation*}
\text{dist}(\bm{x},\Omega_{R}^{c}\cap \Omega_{+}\cap\mathcal{D}_{e})+\text{dist}(\bm{x},\{x_1\leq L\})\geq \text{dist}(\{x_1\leq L\},\Omega_{R}^{c}\cap \Omega_{+})=R-L,
\end{equation*}
\begin{equation*}
\text{dist}(\bm{x},\Omega_{R}^{c}\cap \Omega_{+}\cap\mathcal{D}_{e})+\text{dist}(\bm{x},\{x_1\leq L\})
\geq \text{dist}(\bm{x},\Omega_{+})+\text{dist}(\bm{x},\Omega_{-})-L\geq |x_1|-L.
\end{equation*}
Therefore,
\begin{equation} \label{eq_bic_proof_4}
\Big|\Big(\mathbbm{1}_{\{x_1\leq L\}}\Sigma_{e}^{\Lambda_2,S\mathbbm{1}_{\Omega_{R}^{c}\cap \Omega_{+}\cap\mathcal{D}_{e} }}(z)\Big)(\bm{x},\bm{y})\Big|
\leq \frac{Ce^{-2D|\text{Im }z|R}}{|\text{Im }z|^{p}}e^{-D|\text{Im }z|\big(\|\bm{x}-\bm{y}\|_{1}+|x_2|+|y_2|+|x_1|+|y_1| \big)},
\end{equation}
where $C,D$ depend only on $\Lambda_2$ and $L$. This, argued analogously to Step 2 of Proposition \ref{prop_interface_spin_torque_justified}, implies 
that 
$$
\mathbbm{1}_{\{x_1\leq L\}}\Sigma_{e}^{\Lambda_2, S\mathbbm{1}_{\Omega_{R}^{c}\cap \Omega_{+}\cap\mathcal{D}_{e}}}=\int_{\mathbb{C}}\frac{\partial \tilde{\rho}}{\partial \overline{z}}\big(\mathbbm{1}_{\{x_1\leq L\}} \Sigma_{e}^{\Lambda_2, S\mathbbm{1}_{\Omega_{R}^{c}\cap \Omega_{+}\cap\mathcal{D}_{e}}}(z)\big) dm(z)
$$ 
is trace-class, and we can thus interchange the complex integral and trace to obtain
\begin{equation*}
\begin{aligned}
Tr_{a_1^{com},\mathcal{D}_{e}}^{pv,1}(\mathbbm{1}_{\{x_1\leq L\}}\Sigma_{e}^{\Lambda_2,S\mathbbm{1}_{\Omega_{R}^{c}\cap \Omega_{+}\cap \mathcal{D}_{e}}})
&=Tr_{\mathcal{D}_e}(\mathbbm{1}_{\{x_1\leq L\}}\Sigma_{e}^{\Lambda_2, S\mathbbm{1}_{\Omega_{R}^{c}\cap \Omega_{+}\cap\mathcal{D}_{e}}}) \\
&=\int_{\mathbb{C}}\frac{\partial \tilde{\rho}}{\partial \overline{z}} 
Tr_{\mathcal{D}_{e}}\big(\mathbbm{1}_{\{x_1\leq L\}} \Sigma_{e}^{\Lambda_2, S\mathbbm{1}_{\Omega_{R}^{c}\cap \Omega_{+}\cap\mathcal{D}_{e}}}(z)\big)
dm(z).  
\end{aligned}
\end{equation*}
By \eqref{eq_bic_proof_4}, we have
\begin{equation*}
\begin{aligned}
\Big| Tr_{a_1^{com},\mathcal{D}_{e}}^{pv,1}(\mathbbm{1}_{\{x_1\leq L\}}\Sigma_{e}^{\Lambda_2,S\mathbbm{1}_{\Omega_{R}^{c}\cap \Omega_{+}\cap \mathcal{D}_{e}}}) \Big|
&=C\int_{\mathbb{C}}\Big|\frac{\partial \tilde{\rho}}{\partial \overline{z}} \Big|\frac{e^{-2D|\text{Im }z|R}}{|\text{Im }z|^{p}}
dm(z). 
\end{aligned}
\end{equation*}
Then \eqref{eq_bic_proof_2} follows by the almost-analyticity \eqref{eq_almost_analytic} and the dominated convergence theorem.

{\color{blue}Step 2.2:} We prove \eqref{eq_bic_proof_3}. This is achieved by summing the following three identities:
\begin{equation} \label{eq_bic_proof_5}
\lim_{R\to\infty} Tr_{a_1^{com},\mathcal{D}_{+}}^{pv,1}\Big(\mathbbm{1}_{\{x_1> L\}}\big(\Sigma_{e}^{\Lambda_2,S\mathbbm{1}_{\Omega_{R}^{c}\cap \Omega_{+}\cap \mathcal{D}_{e}}}-\Sigma_{+}^{\Lambda_2,S\mathbbm{1}_{\Omega_{R}^{c}\cap \Omega_{+}\cap \mathcal{D}_{e}}} \big) \Big)
=0.
\end{equation}
\begin{equation} \label{eq_bic_proof_6}
-\lim_{R\to\infty} Tr_{a_1^{com},\mathcal{D}_{+}}^{pv,1}(\mathbbm{1}_{\{x_1\leq L\}}\Sigma_{+}^{\Lambda_2,S\mathbbm{1}_{\Omega_{R}^{c}\cap \Omega_{+}\cap \mathcal{D}_{e}}})
=0.
\end{equation}
\begin{equation} \label{eq_bic_proof_7}
Tr_{a_1^{com},\mathcal{D}_{+}}^{pv,1}(\Sigma_{+}^{\Lambda_2,S\mathbbm{1}_{\Omega_{R}^{c}\cap \Omega_{+}\cap \mathcal{D}_{e}}})
=2\pi\sigma_{+}^{\Lambda_2} \quad (\forall R>L).
\end{equation}
The identity \eqref{eq_bic_proof_6} follows from similar lines as in Step 2.1, and \eqref{eq_bic_proof_7} follows directly from Proposition \ref{prop_spin_conduct_justified}. It remains to prove \eqref{eq_bic_proof_5}. The proof is also similar to Step 2.1, with the only difference being that the decay in the $x_1$-direction is now provided by $\mathcal{H}_{e}-\mathcal{H}_{+}$ (Lemma \ref{lem_He-Hpm}), rather than by the indicator function $\mathbbm{1}_{\{x_1\leq L\}}$, as we demonstrate below.

Note that, similar to Step 2 of Proposition \ref{prop_interface_spin_torque_justified}, the resolvent identity yields that the difference 
$$
\Sigma_{e}^{\Lambda_2,S\mathbbm{1}_{\Omega_{R}^{c}\cap \Omega_{+}\cap \mathcal{D}_{e}}}(z)-\Sigma_{+}^{\Lambda_2,S\mathbbm{1}_{\Omega_{R}^{c}\cap \Omega_{+}\cap \mathcal{D}_{e}}}(z)
$$ 
can be expressed as a finite product of operators from the set
$$
\Big\{\mathcal{H}_{+}-\mathcal{H}_{e},[\mathcal{H}_{+}-\mathcal{H}_{e},\Lambda_2],\mathbbm{1}_{\Omega_{R}^{c}\cap \Omega_{+}\cap \mathcal{D}_{e}},R_{*}(z),[\mathcal{H}_{*},\Lambda_2] \quad (*\in\{e,+\}) \Big\},
$$
with the operator $\mathbbm{1}_{\Omega_{R}^{c}\cap \Omega_{+}\cap\mathcal{D}_{e}}$ and at lease another one involving $\mathcal{H}_{+}-\mathcal{H}_{e}$. These two operators provide the localization in $x_1$ by Lemma \ref{lem_He-Hpm}:
\begin{equation*}
\left\{
\begin{aligned}
&|(\mathcal{H}_{+}-\mathcal{H}_{e})(\bm{x},\bm{y})|\leq C_1e^{-D_1(\|\bm{x}-\bm{y}\|_{1}+\text{dist}(\bm{x},\Omega_{-})+\text{dist}(\bm{y},\Omega_{-}))}, \\
&|\mathbbm{1}_{\Omega_{R}^{c}\cap \Omega_{+}\cap\mathcal{D}_{e}}(\bm{x},\bm{y})|
\leq C_{2}e^{-D_2(\|\bm{x}-\bm{y}\|_{1}+\text{dist}(\bm{x},\Omega_{R}^{c}\cap \Omega_{+}\cap\mathcal{D}_{e})+\text{dist}(\bm{y},\Omega_{R}^{c}\cap \Omega_{+}\cap\mathcal{D}_{e}))}.
\end{aligned}
\right.
\end{equation*}
 Hence, by arguing in the same way as in Step 2.1, we derive the following bound: 
\begin{equation*}
\begin{aligned}
&\Big|\Big(\mathbbm{1}_{\{x_1> L\}}\big(\Sigma_{e}^{\Lambda_2,S\mathbbm{1}_{\Omega_{R}^{c}\cap \Omega_{+}\cap \mathcal{D}_{e}}}(z)-\Sigma_{+}^{\Lambda_2,S\mathbbm{1}_{\Omega_{R}^{c}\cap \Omega_{+}\cap \mathcal{D}_{e}}}(z)\big)\Big)(\bm{x},\bm{y}) \Big| \\
&\leq \Big|\Big(\Sigma_{e}^{\Lambda_2,S\mathbbm{1}_{\Omega_{R}^{c}\cap \Omega_{+}\cap \mathcal{D}_{e}}}(z)-\Sigma_{+}^{\Lambda_2,S\mathbbm{1}_{\Omega_{R}^{c}\cap \Omega_{+}\cap \mathcal{D}_{e}}}(z)\Big)(\bm{x},\bm{y}) \Big| \\
&\leq \frac{C}{|\text{Im }z|^{p}}e^{-2D|\text{Im }z|\big(\|\bm{x}-\bm{y}\|_{1}+|x_2|+|y_2|+\text{dist}(\bm{x},\Omega_{-})+\text{dist}(\bm{y},\Omega_{-})+\text{dist}(\bm{x},\Omega_{R}^{c}\cap \Omega_{+}\cap\mathcal{D}_{e})+\text{dist}(\bm{y},\Omega_{R}^{c}\cap \Omega_{+}\cap\mathcal{D}_{e}) \big)} \\
&\leq \frac{Ce^{-D|\text{Im }z|R}}{|\text{Im }z|^{p}}e^{-D|\text{Im }z|\big(\|\bm{x}-\bm{y}\|_{1}+|x_2|+|y_2|+|x_1|+|y_1| \big)}
\end{aligned}
\end{equation*}
for some $C,D,p>0$ that are independent of $\text{Im }z$ and $R$. With this estimate, one can prove \eqref{eq_bic_proof_5} following similar calculations as in Step 2.1. The details are skipped.

{\color{blue}Step 3:} In this step we prove
\begin{equation} \label{eq_bic_proof_8}
\lim_{R\to\infty} Tr_{a_1^{com},\mathcal{D}_{e}}^{pv,1}(\Sigma_{e}^{\Lambda_2,S\mathbbm{1}_{\Omega_{R}^{c}\cap \Omega_{-}\cap \mathcal{D}_{e}}})=-2\pi \sigma_{-}^{\Lambda_2}.
\end{equation}
The proof resembles that of Step 2, with additional effort on dealing with the spin-torque response in the left bulk region $\mathcal{D}_{-}$, as indicated at the beginning of this section. On the technical side, this involves tuning an additional parameter, specifically, $N>0$, as detailed in the paragraphs below. Denote $\Lambda_1^{R}=1-\mathbbm{1}_{\Omega_{R}^{c}\cap \Omega_{-}\cap \mathcal{D}_{e}}$. Choosing $N>0$ sufficiently large so that $Na_{1}^{-}>L$, we have the following decomposition: 
\begin{equation*}
\begin{aligned}
Tr_{a_1^{com},\mathcal{D}_{e}}^{pv,1}(\Sigma_{e}^{\Lambda_2,S\mathbbm{1}_{\Omega_{R}^{c}\cap \Omega_{-}\cap \mathcal{D}_{e}}})
&=Tr_{a_1^{com},\mathcal{D}_{e}}^{pv,1}(\mathbbm{1}_{\{x_1>-Na_{1}^{-}\}}\Sigma_{e}^{\Lambda_2,S\mathbbm{1}_{\Omega_{R}^{c}\cap \Omega_{-}\cap \mathcal{D}_{e}}} ) \\
&\quad +Tr_{a_1^{com},\mathcal{D}_{-}}^{pv,1}(\mathbbm{1}_{\{x_1\leq -Na_{1}^{-}\}}\Sigma_{e}^{\Lambda_2,S\mathbbm{1}_{\Omega_{R}^{c}\cap \Omega_{-}\cap \mathcal{D}_{e}}} ) \\
&=Tr_{a_1^{com},\mathcal{D}_{e}}^{pv,1}(\mathbbm{1}_{\{x_1>-Na_{1}^{-}\}}\Sigma_{e}^{\Lambda_2,S\mathbbm{1}_{\Omega_{R}^{c}\cap \Omega_{-}\cap \mathcal{D}_{e}}} ) \\
&\quad - Tr_{a_1^{com},\mathcal{D}_{-}}^{pv,1}(\mathbbm{1}_{\{x_1\leq -Na_{1}^{-}\}}\Sigma_{e}^{\Lambda_2,S\Lambda_1^{R}} ) 
+ Tr_{a_1^{com},\mathcal{D}_{-}}^{pv,1}(\mathbbm{1}_{\{x_1\leq -Na_{1}^{-}\}}\Sigma_{e}^{\Lambda_2,S} ) \\
&=Tr_{a_1^{com},\mathcal{D}_{e}}^{pv,1}(\mathbbm{1}_{\{x_1>-Na_{1}^{-}\}}\Sigma_{e}^{\Lambda_2,S\mathbbm{1}_{\Omega_{R}^{c}\cap \Omega_{-}\cap \mathcal{D}_{e}}} ) - Tr_{a_1^{com},\mathcal{D}_{-}}^{pv,1}(\mathbbm{1}_{\{x_1\leq -Na_{1}^{-}\}}\Sigma_{e}^{\Lambda_2,S\Lambda_1^{R}} ) \\ 
&\quad + Tr_{a_1^{com},\mathcal{D}_{-}}^{pv,1}(\mathbbm{1}_{\{x_1\leq -Na_{1}^{-}\}}(\Sigma_{e}^{\Lambda_2,S}-\Sigma_{-}^{\Lambda_2,S}) ) + Tr_{a_1^{com},\mathcal{D}_{-}}^{pv,1}(\mathbbm{1}_{\{x_1\leq -Na_{1}^{-}\}}\Sigma_{-}^{\Lambda_2,S}).
\end{aligned}
\end{equation*}
We claim that the following identities hold:
\begin{equation} \label{eq_bic_proof_9}
\lim_{R\to \infty}Tr_{a_1^{com},\mathcal{D}_{e}}^{pv,1}(\mathbbm{1}_{\{x_1>-Na_{1}^{-}\}}\Sigma_{e}^{\Lambda_2,S\mathbbm{1}_{\Omega_{R}^{c}\cap \Omega_{-}\cap \mathcal{D}_{e}}} )=0 \quad \text{for fixed $N$},
\end{equation}
\begin{equation} \label{eq_bic_proof_10}
\lim_{R\to \infty}Tr_{a_1^{com},\mathcal{D}_{-}}^{pv,1}(\mathbbm{1}_{\{x_1\leq -Na_{1}^{-}\}}\Sigma_{e}^{\Lambda_2,S\Lambda_1^{R}} ) = 2\pi \sigma_{-}^{\Lambda_2} \quad \text{for fixed $N$},
\end{equation}
\begin{equation} \label{eq_bic_proof_11}
\lim_{N\to \infty} Tr_{a_1^{com},\mathcal{D}_{-}}^{pv,1}(\mathbbm{1}_{\{x_1\leq -Na_{1}^{-}\}}(\Sigma_{e}^{\Lambda_2,S}-\Sigma_{-}^{\Lambda_2,S}) ) =0 ,
\end{equation}
\begin{equation} \label{eq_bic_proof_12} Tr_{a_1^{com},\mathcal{D}_{-}}^{pv,1}(\mathbbm{1}_{\{x_1\leq -Na_{1}^{-}\}}\Sigma_{-}^{\Lambda_2,S} ) =0 \quad \text{for any $N$}.
\end{equation}
Then the proof of \eqref{eq_bic_proof_8} follows as
\begin{equation*}
\begin{aligned}
&\lim_{R\to\infty} Tr_{a_1^{com},\mathcal{D}_{e}}^{pv,1}(\Sigma_{e}^{\Lambda_2,S\mathbbm{1}_{\Omega_{R}^{c}\cap \Omega_{-}\cap \mathcal{D}_{e}}}) \\
&=\lim_{R\to\infty}\Big[Tr_{a_1^{com},\mathcal{D}_{e}}^{pv,1}(\mathbbm{1}_{\{x_1>-Na_{1}^{-}\}}\Sigma_{e}^{\Lambda_2,S\mathbbm{1}_{\Omega_{R}^{c}\cap \Omega_{-}\cap \mathcal{D}_{e}}} ) - Tr_{a_1^{com},\mathcal{D}_{-}}^{pv,1}(\mathbbm{1}_{\{x_1\leq -Na_{1}^{-}\}}\Sigma_{e}^{\Lambda_2,S\Lambda_1^{R}} ) \\ 
&\quad + Tr_{a_1^{com},\mathcal{D}_{-}}^{pv,1}(\mathbbm{1}_{\{x_1\leq -Na_{1}^{-}\}}(\Sigma_{e}^{\Lambda_2,S}-\Sigma_{-}^{\Lambda_2,S}) ) + Tr_{a_1^{com},\mathcal{D}_{-}}^{pv,1}(\mathbbm{1}_{\{x_1\leq -Na_{1}^{-}\}}\Sigma_{-}^{\Lambda_2,S} )\Big] \\
&=-2\pi \sigma_{-}^{\Lambda_2}+Tr_{a_1^{com},\mathcal{D}_{-}}^{pv,1}(\mathbbm{1}_{\{x_1\leq -Na_{1}^{-}\}}(\Sigma_{e}^{\Lambda_2,S}-\Sigma_{-}^{\Lambda_2,S}) ) + Tr_{a_1^{com},\mathcal{D}_{-}}^{pv,1}(\mathbbm{1}_{\{x_1\leq -Na_{1}^{-}\}}\Sigma_{-}^{\Lambda_2,S} ) \\
&\overset{N\to\infty }{\to} -2\pi \sigma_{-}^{\Lambda_2} .
\end{aligned}
\end{equation*}
The proofs of \eqref{eq_bic_proof_9} and \eqref{eq_bic_proof_10} proceed exactly as in Steps 2.1 and 2.2, respectively. Moreover, \eqref{eq_bic_proof_12} follows from the definition of principal-value trace and Proposition \ref{prop_vanish_bulk_torque_correlation}:
\begin{equation*}
\begin{aligned}
Tr_{a_1^{com},\mathcal{D}_{-}}^{pv,1}(\mathbbm{1}_{\{x_1\leq -Na_{1}^{-}\}}\Sigma_{-}^{\Lambda_2,S} )
&=\lim_{M\to\infty}\sum_{n=-M}^{M-1}Tr_{\mathcal{D}_{-}}(\mathbbm{1}_{\Omega_{n,a_1^{com}}\cap \mathcal{D}_{-}}\mathbbm{1}_{\{x_1\leq -Na_{1}^{-}\}}\Sigma_{-}^{\Lambda_2,S}) \\
&=\lim_{M\to\infty}\sum_{n=-M}^{-N}Tr_{\mathcal{D}_{-}}(\mathbbm{1}_{\Omega_{n,a_1^{com}}\cap \mathcal{D}_{-}}\Sigma_{-}^{\Lambda_2,S})=0
\end{aligned}
\end{equation*}
as each term in the sum equals zero (vanishing mesoscopic average of spin-torque response in the left bulk). Finally, we point out the identity \eqref{eq_bic_proof_11} is a consequence of the fact that the difference of Hamiltonians $\mathcal{H}_{e}-\mathcal{H}_{-}$ evaluated inside the left half plane $\{x_1\leq -Na_{1}^{-}\}$ vanishes asymptotically as $N\to\infty$. A rigorous justification is similar to the proof of \eqref{eq_bic_proof_5}. We only sketch the main steps as we did in Step 2.2.

By the resolvent identity, $\big(\Sigma_{e}^{\Lambda_2,S}(z)-\Sigma_{-}^{\Lambda_2,S}(z)\big)\mathbbm{1}_{\{x_1\leq -Na_{1}^{-}\}}$ is the product of finitely many operators from 
$$
\Big\{\mathcal{H}_{-}-\mathcal{H}_{e},[\mathcal{H}_{-}-\mathcal{H}_{e},\Lambda_2],\mathbbm{1}_{\{x_1\leq -Na_{1}^{-}\}},R_{*}(z),[\mathcal{H}_{*},\Lambda_2] \quad (*\in\{e,-\}) \Big\}
$$
and necessarily contains $\mathcal{H}_{-}-\mathcal{H}_{e}$ and $\mathbbm{1}_{\{x_1\leq -Na_{1}^{-}\}}$. These two operators provide the localization in $x_1$:
\begin{equation*}
\left\{
\begin{aligned}
&|(\mathcal{H}_{-}-\mathcal{H}_{e})(\bm{x},\bm{y})|\leq C_1e^{-D_1(\|\bm{x}-\bm{y}\|_{1}+\text{dist}(\bm{x},\Omega_{+})+\text{dist}(\bm{y},\Omega_{+}))}, \\
&|\mathbbm{1}_{\{x_1\leq -Na_{1}^{-}\}}(\bm{x},\bm{y})|
\leq C_{2}e^{-D_2(\|\bm{x}-\bm{y}\|_{1}+\text{dist}(\bm{x},\{x_1\leq -Na_{1}^{-}\})+\text{dist}(\bm{y},\{x_1\leq -Na_{1}^{-}\}))}
\end{aligned}
\right.
\end{equation*}
by Lemma \ref{lem_He-Hpm}. Hence, incorporating the other operators in $\mathbbm{1}_{\{x_1\leq -Na_{1}^{-}\}}\big(\Sigma_{e}^{\Lambda_2,S}(z)-\Sigma_{-}^{\Lambda_2,S}(z)\big)$ and arguing in the same way as in Step 2.1, we have the following bound
\begin{equation*}
\begin{aligned}
&\Big|\Big(\mathbbm{1}_{\{x_1\leq -Na_{1}^{-}\}}\big(\Sigma_{e}^{\Lambda_2,S}(z)-\Sigma_{-}^{\Lambda_2,S}(z)\big)\Big)(\bm{x},\bm{y}) \Big| \\
&\leq \frac{C}{|\text{Im }z|^{p}}e^{-2D|\text{Im }z|\big(\|\bm{x}-\bm{y}\|_{1}+|x_2|+|y_2|+\text{dist}(\bm{x},\Omega_{+})+\text{dist}(\bm{y},\Omega_{+})+\text{dist}(\bm{x},\{x_1\leq -Na_{1}^{-}\})+\text{dist}(\bm{y},\{x_1\leq -Na_{1}^{-}\}) \big)} \\
&\leq \frac{Ce^{-D|\text{Im }z|N}}{|\text{Im }z|^{p}}e^{-D|\text{Im }z|\big(\|\bm{x}-\bm{y}\|_{1}+|x_2|+|y_2|+|x_1|+|y_1| \big)}
\end{aligned}
\end{equation*}
for some $C,D,p>0$ that are independent of $\text{Im }z$ and $N$. With this estimate, one can prove \eqref{eq_bic_proof_11} following similar as in calculations in Step 2.1. The details are skipped.

{\color{blue}Step 4:} In this step we prove
\begin{equation} \label{eq_bic_proof_13}
\lim_{R\to\infty} Tr_{a_1^{com},\mathcal{D}_{e}}^{pv,1}(\Sigma_{e}^{\Lambda_2,S\mathbbm{1}_{\Omega_{R}\cap \mathcal{D}_{e}}})=-2\pi\sigma_{e}^{drift,\Lambda_2/\rho}.
\end{equation}
Recall that
\begin{equation*}
\begin{aligned}
\Sigma_{e}^{\Lambda_2, S\mathbbm{1}_{\Omega_{R}\cap \mathcal{D}_{e}}}
&=\int_{\mathbb{C}}dm(z)\frac{\partial \tilde{\rho}}{\partial \overline{z}}\Big[
R_{e}(z) \big[\mathcal{H}_{e},\Lambda_2\big] R_{e}(z) \big[\mathcal{H}_{e}, S\mathbbm{1}_{\Omega_{R}\cap \mathcal{D}_{e}}\big] R_{e}(z) \\
&\quad\quad\quad\quad\quad\quad\quad- R_{e}(z) \big[\mathcal{H}_{e}, S\mathbbm{1}_{\Omega_{R}\cap \mathcal{D}_{e}}\big] R_{e}(z) \big[\mathcal{H}_{e},\Lambda_2\big] R_{e}(z)
\Big] \\
&=\int_{\mathbb{C}}\frac{\partial \tilde{\rho}}{\partial \overline{z}} \Sigma_{e}^{\Lambda_2, S\mathbbm{1}_{\Omega_{R}\cap \mathcal{D}_{e}}}(z) dm(z).
\end{aligned}
\end{equation*}
we will avoid detailed technical arguments and instead present a streamlined outline of the proof, as each step closely parallels arguments already presented in the previous sections. First, we note that both $R_{e}(z) \big[\mathcal{H}_{e},\Lambda_2\big] R_{e}(z) \big[\mathcal{H}_{e}, S\mathbbm{1}_{\Omega_{R}\cap \mathcal{D}_{e}}\big] R_{e}(z)$ and $R_{e}(z) \big[\mathcal{H}_{e},S\mathbbm{1}_{\Omega_{R}\cap \mathcal{D}_{e}}\big] R_{e}(z) \big[\mathcal{H}_{e}, \Lambda_2\big] R_{e}(z)$ are trace-class for $\text{Im }z\neq 0$ because the commutator $\big[\mathcal{H}_{e},\Lambda_2\big]$ is localized in $x_2$ and the indicator function $\mathbbm{1}_{\Omega_{R}\cap \mathcal{D}_{e}}$ is localized in $x_1$. Moreover, the trace-norm of $\Sigma_{e}^{\Lambda_2, S\mathbbm{1}_{\Omega_{R}\cap \mathcal{D}_{e}}}(z)$ diverges at worst with polynomial rate as $|\text{Im }z|\to 0$, analogous to Step 1 of Proposition \ref{prop_vanish_bulk_torque_correlation}. Thus, by the almost-analyticity of $\tilde{\rho}$, $\Sigma_{e}^{\Lambda_2, S\mathbbm{1}_{\Omega_{R}\cap \mathcal{D}_{e}}}$ is trace-class, and we can interchange the trace and complex integral to obtain
\begin{equation} \label{eq_bic_proof_14}
\begin{aligned}
Tr_{a_1^{com},\mathcal{D}_{e}}^{pv,1}(\Sigma_{e}^{\Lambda_2,S\mathbbm{1}_{\Omega_{R}\cap \mathcal{D}_{e}}})
&=Tr_{\mathcal{D}_{e}}(\Sigma_{e}^{\Lambda_2,S\mathbbm{1}_{\Omega_{R}\cap \mathcal{D}_{e}}})
=\int_{\mathbb{C}}\frac{\partial \tilde{\rho}}{\partial \overline{z}} Tr_{\mathcal{D}_{e}}\big(\Sigma_{e}^{\Lambda_2, S\mathbbm{1}_{\Omega_{R}\cap \mathcal{D}_{e}}}(z) \big) dm(z) \\
&=\int_{\mathbb{C}}\frac{\partial \tilde{\rho}}{\partial \overline{z}} Tr_{\mathcal{D}_{e}}\Big(  R_{e}(z) \big[\mathcal{H}_{e},\Lambda_2\big] R_{e}(z) \big[\mathcal{H}_{e}, S\mathbbm{1}_{\Omega_{R}\cap \mathcal{D}_{e}}\big] R_{e}(z) \Big) dm(z) \\
&\quad -\int_{\mathbb{C}}\frac{\partial \tilde{\rho}}{\partial \overline{z}} Tr_{\mathcal{D}_{e}}\Big( R_{e}(z) \big[\mathcal{H}_{e}, S\mathbbm{1}_{\Omega_{R}\cap \mathcal{D}_{e}}\big] R_{e}(z) \big[\mathcal{H}_{e},\Lambda_2\big] R_{e}(z) \Big) dm(z) .
\end{aligned}
\end{equation}
With some algebraic manipulation, one sees
\begin{equation*}
\begin{aligned}
Tr_{\mathcal{D}_{e}}\Big(  R_{e}(z) \big[\mathcal{H}_{e},\Lambda_2\big] R_{e}(z) \big[\mathcal{H}_{e}, S\mathbbm{1}_{\Omega_{R}\cap \mathcal{D}_{e}}\big] R_{e}(z) \Big)
&= -Tr_{\mathcal{D}_{e}}\Big(  R_{e}(z) \big[\mathcal{H}_{e},\Lambda_2\big] \big[R_{e}(z), S\mathbbm{1}_{\Omega_{R}\cap \mathcal{D}_{e}}\big] \Big) \\
&=- Tr_{\mathcal{D}_{e}}\Big(  R_{e}(z) \big[\mathcal{H}_{e},\Lambda_2\big] R_{e}(z) S\mathbbm{1}_{\Omega_{R}\cap \mathcal{D}_{e}} \Big) \\
&\quad + Tr_{\mathcal{D}_{e}}\Big(  R_{e}(z) \big[\mathcal{H}_{e},\Lambda_2\big]  S\mathbbm{1}_{\Omega_{R}\cap \mathcal{D}_{e}} R_{e}(z) \Big) \\
&=Tr_{\mathcal{D}_{e}}\Big( \big[R_{e}(z),\Lambda_2\big] S\mathbbm{1}_{\Omega_{R}\cap \mathcal{D}_{e}} \Big) \\
&\quad + Tr_{\mathcal{D}_{e}}\Big(  R_{e}^2(z) \big[\mathcal{H}_{e},\Lambda_2\big]  S\mathbbm{1}_{\Omega_{R}\cap \mathcal{D}_{e}} \Big), \\
\end{aligned}
\end{equation*}
where the cyclicity is applied to derive the last equality. The first trace on the right side vanishes, following a direct calculation in the position basis (analogous to Step 3 of Section 4; the key is that the spin operator $S$ commutes with position operators). Hence, with a similar manipulation on the second term of \eqref{eq_bic_proof_14}, we obtain
\begin{equation*}
\begin{aligned}
Tr_{a_1^{com},\mathcal{D}_{e}}^{pv,1}(\Sigma_{e}^{\Lambda_2,S\mathbbm{1}_{\Omega_{R}\cap \mathcal{D}_{e}}})
&=\int_{\mathbb{C}}\frac{\partial \tilde{\rho}}{\partial \overline{z}} \Big[Tr_{\mathcal{D}_{e}}\Big(  R_{e}^2(z) \big[\mathcal{H}_{e},\Lambda_2\big]  S\mathbbm{1}_{\Omega_{R}\cap \mathcal{D}_{e}} \Big) 
+ Tr_{\mathcal{D}_{e}}\Big( S\mathbbm{1}_{\Omega_{R}\cap \mathcal{D}_{e}} \big[\mathcal{H}_{e},\Lambda_2\big]  R_{e}^2(z)\Big) \Big]dm(z) \\
&=\int_{\mathbb{C}}\frac{\partial \tilde{\rho}}{\partial \overline{z}} \Big[Tr_{\mathcal{D}_{e}}\Big( \big[\mathcal{H}_{e},\Lambda_2\big]  S\mathbbm{1}_{\Omega_{R}\cap \mathcal{D}_{e}}R_{e}^2(z) \Big) 
+ Tr_{\mathcal{D}_{e}}\Big( S\mathbbm{1}_{\Omega_{R}\cap \mathcal{D}_{e}} \big[\mathcal{H}_{e},\Lambda_2\big]  R_{e}^2(z)\Big) \Big]dm(z) \\
&=\int_{\mathbb{C}}\frac{\partial \tilde{\rho}}{\partial \overline{z}} Tr_{\mathcal{D}_{e}}\Big( \big\{\big[\mathcal{H}_{e},\Lambda_2\big],  S\mathbbm{1}_{\Omega_{R}\cap \mathcal{D}_{e}}\big\}R_{e}^2(z) \Big) dm(z).
\end{aligned}
\end{equation*}
Noting that $\partial_{z}R_{e}(z)=-R_{e}^2(z)$, we have
\begin{equation*}
\begin{aligned}
Tr_{a_1^{com},\mathcal{D}_{e}}^{pv,1}(\Sigma_{e}^{\Lambda_2,S\mathbbm{1}_{\Omega_{R}\cap \mathcal{D}_{e}}})
&=-\int_{\mathbb{C}}\frac{\partial \tilde{\rho}}{\partial \overline{z}} Tr_{\mathcal{D}_{e}}\Big( \big\{\big[\mathcal{H}_{e},\Lambda_2\big] , S\mathbbm{1}_{\Omega_{R}\cap \mathcal{D}_{e}}\big\}\partial_{z} R_{e}(z) \Big) dm(z) \\
&=\int_{\mathbb{C}}\frac{\partial^2 \tilde{\rho}}{\partial z\partial \overline{z}} Tr_{\mathcal{D}_{e}}\Big( \big\{\big[\mathcal{H}_{e},\Lambda_2\big] , S\mathbbm{1}_{\Omega_{R}\cap \mathcal{D}_{e}}\big\} R_{e}(z) \Big) dm(z),
\end{aligned}
\end{equation*}
where the integral by parts with respect to $z$ is applied in the last equality. Hence, by the Hellfer-Sjöstrand formula, we conclude that
\begin{equation*}
Tr_{a_1^{com},\mathcal{D}_{e}}^{pv,1}(\Sigma_{e}^{\Lambda_2,S\mathbbm{1}_{\Omega_{R}\cap \mathcal{D}_{e}}})
=i\pi Tr_{\mathcal{D}_{e}}\Big( \big\{\big[\mathcal{H}_{e},\Lambda_2\big],  S\mathbbm{1}_{\Omega_{R}\cap \mathcal{D}_{e}}\big\} \rho^{\prime}(H_e) \Big) .
\end{equation*}
This, combined with the following identity
\begin{equation*}
\lim_{R\to\infty} Tr_{\mathcal{D}_{e}}\Big( \big\{\big[\mathcal{H}_{e},\Lambda_2\big] , S\mathbbm{1}_{\Omega_{R}^{c}\cap \Omega_{+}\cap \mathcal{D}_{e}} \big\}\rho^{\prime}(H_e) \Big)
=\lim_{R\to\infty} Tr_{\mathcal{D}_{e}}\Big( \big\{\big[\mathcal{H}_{e},\Lambda_2\big] , S\mathbbm{1}_{\Omega_{R}^{c}\cap \Omega_{-}\cap \mathcal{D}_{e}}\big\} \rho^{\prime}(H_e) \Big)
=0 ,
\end{equation*}
completes the proof of \eqref{eq_bic_proof_13}. Physically, the underlying idea behind this identity is that the longitudinal spin-drift transport vanishes when measured sufficiently deep within the bulk, since the bulk medium is insulating. Mathematically, this argument follows similarly to Step 2, by observing that the operator $\rho^{\prime}(H_e)$ can be rewritten as follows which is localized near the interface
$$
\rho^{\prime}(H_e)=\rho^{\prime}(H_e)-\rho^{\prime}(H_\pm),
$$ 
($\rho^{\prime}(H_\pm)=0$ since $\rho^{\prime}$ is supported within the spectral gap) and the indicator 
function $\mathbbm{1}_{\Omega_{R}^{c}\cap \Omega_{\pm}\cap \mathcal{D}_{e}}$ decays to zero as $R\to \infty$.
\end{proof}

\footnotesize
\bibliographystyle{plain}
\bibliography{ref}

\begin{thebibliography}{10}

\bibitem{AMZ2020bec_discrete_K}
Alexander Alldridge, Christopher Max, and Martin~R Zirnbauer.
\newblock Bulk-boundary correspondence for disordered free-fermion topological phases.
\newblock {\em Communications in Mathematical Physics}, 377(3):1761--1821, 2020.

\bibitem{avila2013shortrange+transfer}
Julio~Cesar Avila, Hermann Schulz-Baldes, and Carlos Villegas-Blas.
\newblock Topological invariants of edge states for periodic two-dimensional models.
\newblock {\em Mathematical Physics, Analysis and Geometry}, 16(2):137--170, 2013.

\bibitem{Avron1994charge}
Joseph~E. Avron, Ruedi Seiler, and Barry Simon.
\newblock Charge deficiency, charge transport and comparison of dimensions.
\newblock {\em Communications in Mathematical Physics}, 159(2):399--422, Jan 1994.

\bibitem{bal2019dirac+functional}
Guillaume Bal.
\newblock Continuous bulk and interface description of topological insulators.
\newblock {\em Journal of Mathematical Physics}, 60(8), 2019.

\bibitem{bal2022bec_dirac_functional_symbol}
Guillaume Bal.
\newblock Topological invariants for interface modes.
\newblock {\em Communications in Partial Differential Equations}, 47(8):1636--1679, 2022.

\bibitem{bernivig13topo_insulator}
B.~Andrei Bernevig and Taylor~L. Hughes.
\newblock {\em Topological Insulators and Topological Superconductors}.
\newblock Princeton University Press, 2013.

\bibitem{BHZ06}
B.~Andrei Bernevig, Taylor~L. Hughes, and Shou-Cheng Zhang.
\newblock Quantum spin hall effect and topological phase transition in hgte quantum wells.
\newblock {\em Science}, 314(5806):1757--1761, 2006.

\bibitem{bhowal2021orbital}
Sayantika Bhowal and Giovanni Vignale.
\newblock Orbital hall effect as an alternative to valley hall effect in gapped graphene.
\newblock {\em Physical Review B}, 103(19):195309, 2021.

\bibitem{BKR2017bec_discrete_K}
Chris Bourne, Johannes Kellendonk, and Adam Rennie.
\newblock The k-theoretic bulk--edge correspondence for topological insulators.
\newblock In {\em Annales Henri Poincar{\'e}}, volume~18, pages 1833--1866. Springer, 2017.

\bibitem{BR2018bec_continuous_K}
Chris Bourne and Adam Rennie.
\newblock Chern numbers, localisation and the bulk-edge correspondence for continuous models of topological phases.
\newblock {\em Mathematical Physics, Analysis and Geometry}, 21:1--62, 2018.

\bibitem{Bra2019bec_discrete_K}
Maxim Braverman.
\newblock Spectral flows of toeplitz operators and bulk-edge correspondence.
\newblock {\em Letters in Mathematical Physics}, 109:2271--2289, 2019.

\bibitem{Vanderbilt2019Berry_phase}
Laurent Chaput.
\newblock Berry phases in electronic structure theory. electric polarization, orbital magnetization and topological insulators. by david vanderbilt. cambridge university press, 2018. hardback, pp. x+384. price gbp 59.99. isbn 9781107157651.
\newblock {\em Acta crystallographica. Section A, Foundations and advances}, 75 Pt 6:913--914, 2019.

\bibitem{CG2005bec_schrodinger_functional}
Jean-Michel Combes and Fran{\c{c}}ois Germinet.
\newblock Edge and impurity effects on quantization of hall currents.
\newblock {\em Communications in mathematical physics}, 256:159--180, 2005.

\bibitem{cornean2021landau+functional}
Horia~D Cornean, Massimo Moscolari, and Stefan Teufel.
\newblock General bulk-edge correspondence at positive temperature.
\newblock {\em arXiv preprint arXiv:2107.13456}, 2021.

\bibitem{de2016spectral_flow}
Giuseppe De~Nittis and Hermann Schulz-Baldes.
\newblock Spectral flows associated to flux tubes.
\newblock In {\em Annales Henri Poincar{\'e}}, volume~17, pages 1--35. Springer, 2016.

\bibitem{DGR2011bec_schrodinger_functional}
Nicolas Dombrowski, Fran{\c{c}}ois Germinet, and Georgi Raikov.
\newblock Quantization of edge currents along magnetic barriers and magnetic guides.
\newblock In {\em Annales Henri Poincar{\'e}}, volume~12, pages 1169--1197. Springer, 2011.

\bibitem{drouot2021microlocal}
Alexis Drouot.
\newblock Microlocal analysis of the bulk-edge correspondence.
\newblock {\em Communications in Mathematical Physics}, 383:2069--2112, 2021.

\bibitem{drouot2024bec_curvedinterfaces}
Alexis Drouot and Xiaowen Zhu.
\newblock The bulk-edge correspondence for curved interfaces, 2024.

\bibitem{EG2002bec_discrete_functional}
Peter Elbau and Gian-Michele Graf.
\newblock Equality of bulk and edge hall conductance revisited.
\newblock {\em Communications in mathematical physics}, 229:415--432, 2002.

\bibitem{elgart2005shortrange+functional}
Alexander Elgart, Gian~M Graf, and Jeffrey~H Schenker.
\newblock Equality of the bulk and edge hall conductances in a mobility gap.
\newblock {\em Communications in mathematical physics}, 259:185--221, 2005.

\bibitem{ezawa2014symmetry}
Motohiko Ezawa.
\newblock Symmetry protected topological charge in symmetry broken phase: Spin-chern, spin-valley-chern and mirror-chern numbers.
\newblock {\em Physics Letters A}, 378(16-17):1180--1184, 2014.

\bibitem{gontier2023edge}
David Gontier.
\newblock Edge states for second order elliptic operators in a channel.
\newblock {\em Journal of Spectral Theory}, 12(3):1155--1202, 2023.

\bibitem{graf2013shortrange+scattering}
Gian~Michele Graf and Marcello Porta.
\newblock Bulk-edge correspondence for two-dimensional topological insulators.
\newblock {\em Communications in Mathematical Physics}, 324:851--895, 2013.

\bibitem{graf2018shortrange+transfer}
Gian~Michele Graf and Jacob Shapiro.
\newblock The bulk-edge correspondence for disordered chiral chains.
\newblock {\em Communications in Mathematical Physics}, 363:829--846, 2018.

\bibitem{jo2024spintronics}
Daegeun Jo, Dongwook Go, Gyung-Min Choi, and Hyun-Woo Lee.
\newblock Spintronics meets orbitronics: Emergence of orbital angular momentum in solids.
\newblock {\em npj Spintronics}, 2(1):19, 2024.

\bibitem{KaneMele05Z_2}
C.~L. Kane and E.~J. Mele.
\newblock ${Z}_{2}$ topological order and the quantum spin hall effect.
\newblock {\em Phys. Rev. Lett.}, 95:146802, Sep 2005.

\bibitem{Kellendonk02landau+ktheory}
J.~Kellendonk, T.~Richter, and H.~Schulz-Baldes.
\newblock Edge current channels and chern numbers in the integer quantum hall effect.
\newblock {\em Reviews in Mathematical Physics}, 14(01):87--119, 2002.

\bibitem{Kellendonk2004landau+ktheory}
Johannes Kellendonk and Hermann Schulz-Baldes.
\newblock Boundary maps for c*-crossed products with with an application to the quantum hall effect.
\newblock {\em Communications in Mathematical Physics}, 249:611--637, 2004.

\bibitem{Klitzing80qhe}
K.~v. Klitzing, G.~Dorda, and M.~Pepper.
\newblock New method for high-accuracy determination of the fine-structure constant based on quantized hall resistance.
\newblock {\em Phys. Rev. Lett.}, 45:494--497, Aug 1980.

\bibitem{Kohn96density}
W.~Kohn.
\newblock Density functional and density matrix method scaling linearly with the number of atoms.
\newblock {\em Phys. Rev. Lett.}, 76:3168--3171, Apr 1996.

\bibitem{kubota2017bec_discrete_K}
Yosuke Kubota.
\newblock Controlled topological phases and bulk-edge correspondence.
\newblock {\em Communications in Mathematical Physics}, 349(2):493--525, 2017.

\bibitem{lee2020valley}
Kyu~Won Lee and Cheol~Eui Lee.
\newblock Quantum valley hall effect in wide-gap semiconductor sic monolayer.
\newblock {\em Scientific reports}, 10(1):5044, 2020.

\bibitem{li2024interface}
Wei Li, Junshan Lin, Jiayu Qiu, and Hai Zhang.
\newblock Interface modes in honeycomb topological photonic structures with broken reflection symmetry.
\newblock {\em arXiv preprint arXiv:2405.03238}, 2024.

\bibitem{lin2022transfer}
Junshan Lin and Hai Zhang.
\newblock Mathematical theory for topological photonic materials in one dimension.
\newblock {\em Journal of Physics A: Mathematical and Theoretical}, 55(49):495203, 2022.

\bibitem{ludewig2020shortrange+coarse}
Matthias Ludewig and Guo~Chuan Thiang.
\newblock Cobordism invariance of topological edge-following states.
\newblock {\em arXiv preprint arXiv:2001.08339}, 2020.

\bibitem{Giovanna22charge_to_spin}
Giovanna Marcelli and Domenico Monaco.
\newblock From charge to spin: Analogies and differences in quantum transport coefficients.
\newblock {\em Journal of Mathematical Physics}, 63(7):072102, 07 2022.

\bibitem{Marcelli2019spin_conductivity}
Giovanna Marcelli, Gianluca Panati, and Cl{\'e}ment Tauber.
\newblock Spin conductance and spin conductivity in topological insulators: Analysis of kubo-like terms.
\newblock {\em Annales Henri Poincar{\'e}}, 20(6):2071--2099, Jun 2019.

\bibitem{marcelli2021new}
Giovanna Marcelli, Gianluca Panati, and Stefan Teufel.
\newblock A new approach to transport coefficients in the quantum spin hall effect.
\newblock In {\em Annales Henri Poincar{\'e}}, volume~22, pages 1069--1111. Springer, 2021.

\bibitem{monaco2020spin}
Domenico Monaco and Lara Ul{\v{c}}akar.
\newblock Spin hall conductivity in insulators with nonconserved spin.
\newblock {\em Physical Review B}, 102(12):125138, 2020.

\bibitem{qiu2024mathematical}
Jiayu Qiu and Hai Zhang.
\newblock A mathematical theory of integer quantum hall effect in photonics.
\newblock {\em arXiv preprint arXiv:2405.17200}, 2024.

\bibitem{qiu2025bulk}
Jiayu Qiu and Hai Zhang.
\newblock Bulk-edge correspondence in finite photonic structure.
\newblock {\em arXiv preprint arXiv:2501.15531}, 2025.

\bibitem{QB2024bec_dirac_functional}
Solomon Quinn and Guillaume Bal.
\newblock Asymmetric transport for magnetic dirac equations.
\newblock {\em Pure and Applied Analysis}, 6(2):353--377, 2024.

\bibitem{shi2006proper}
Junren Shi, Ping Zhang, Di~Xiao, and Qian Niu.
\newblock Proper definition of spin current in spin-orbit coupled systems.
\newblock {\em Physical review letters}, 96(7):076604, 2006.

\bibitem{simon2005trace}
Barry Simon.
\newblock {\em Trace ideals and their applications}.
\newblock Number 120. American Mathematical Soc., 2005.

\bibitem{stone1992quantum}
Michael Stone.
\newblock {\em Quantum Hall Effect}.
\newblock World Scientific, 1992.

\bibitem{sun2024nonconserved}
Hao Sun, Alexander Kazantsev, Alessandro Principi, and Giovanni Vignale.
\newblock Nonconserved density accumulations in orbital hall transport: Insights from linear response theory.
\newblock {\em arXiv preprint arXiv:2410.20668}, 2024.

\bibitem{taarabt2014landau+functional}
Amal Taarabt.
\newblock Equality of bulk and edge hall conductances for continuous magnetic random schr\"odinger operators, 2014.

\bibitem{thiang2023transfer}
Guo~Chuan Thiang and Hai Zhang.
\newblock Bulk-interface correspondences for one-dimensional topological materials with inversion symmetry.
\newblock {\em Proceedings of the Royal Society A}, 479(2270):20220675, 2023.

\bibitem{tong2016lecturesquantumhalleffect}
David Tong.
\newblock Lectures on the quantum hall effect, 2016.

\bibitem{vonKlitzing2020forty_years}
Klaus von Klitzing, Tapash Chakraborty, Philip Kim, Vidya Madhavan, Xi~Dai, James McIver, Yoshinori Tokura, Lucile Savary, Daria Smirnova, Ana~Maria Rey, Claudia Felser, Johannes Gooth, and Xiaoliang Qi.
\newblock 40 years of the quantum hall effect.
\newblock {\em Nature Reviews Physics}, 2(8):397--401, Aug 2020.

\bibitem{wu2017valley}
Xiaoxiao Wu, Yan Meng, Jingxuan Tian, Yingzhou Huang, Hong Xiang, Dezhuan Han, and Weijia Wen.
\newblock Direct observation of valley-polarized topological edge states in designer surface plasmon crystals.
\newblock {\em Nature communications}, 8(1):1304, 2017.

\bibitem{xiao2021conserved}
Cong Xiao and Qian Niu.
\newblock Conserved current of nonconserved quantities.
\newblock {\em Physical Review B}, 104(24):L241411, 2021.

\bibitem{xiao2007valley}
Di~Xiao, Wang Yao, and Qian Niu.
\newblock Valley-contrasting physics in graphene: magnetic moment and topological transport.
\newblock {\em Physical review letters}, 99(23):236809, 2007.

\bibitem{vzelezny2018spin_torque}
J~{\v{Z}}elezn{\`y}, P~Wadley, K~Olejn{\'\i}k, A~Hoffmann, and H~Ohno.
\newblock Spin transport and spin torque in antiferromagnetic devices.
\newblock {\em Nature Physics}, 14(3):220--228, 2018.

\bibitem{zworski2012semiclassical}
M.~Zworski.
\newblock {\em Semiclassical Analysis}.
\newblock Graduate studies in mathematics. American Mathematical Society, 2012.

\end{thebibliography}

\end{document}